\documentclass[11pt,journal,onecolumn]{IEEEtran}

\usepackage[T1]{fontenc}
\usepackage[utf8]{inputenc}
\usepackage{lmodern,microtype}
\usepackage{amsmath,amssymb,amsthm,mathtools,bm,mathrsfs}
\usepackage{booktabs,longtable,array,tabularx,graphicx,float}
\usepackage{enumitem,cite,xurl}
\usepackage[hidelinks]{hyperref}
\allowdisplaybreaks[2]
\setlist{leftmargin=*,itemsep=2pt,topsep=3pt}
\AtBeginDocument{\DeclareFontShape{T1}{lmr}{bx}{sc}{<->ssub*lmr/bx/n}{}}
\newtheorem{theorem}{Theorem}
\newtheorem{lemma}[theorem]{Lemma}
\newtheorem{proposition}[theorem]{Proposition}
\newtheorem{corollary}[theorem]{Corollary}
\theoremstyle{definition}

\theoremstyle{remark}

\newcommand{\E}{\mathbb E}
\newcommand{\PP}{\mathbb P}

\newcommand{\Var}{\operatorname{Var}}
\newcommand{\Cov}{\operatorname{Cov}}
\newcommand{\eps}{\varepsilon}
\newcommand{\one}{\mathbf 1}

\newcommand{\supp}{\operatorname{supp}}
\newcommand{\TV}{\operatorname{d_{TV}}}
\newcommand{\diag}{\operatorname{diag}}

\newcommand{\cN}{\mathcal N}

\graphicspath{{figures/}}

\newcounter{manuscriptpart}


\hypersetup{pdftitle={Second-Order Image Characterization with Distortion and Secrecy},pdfauthor={Daming Cao}}

\ifdefined\HighlightEdits
  \DeclareRobustCommand{\edit}[1]{{\bfseries #1}}
\else
  \DeclareRobustCommand{\edit}[1]{#1}
\fi
\newcommand{\tr}{\operatorname{tr}}

\begin{document}

\title{Second-Order Image Characterization with Distortion and Secrecy}
\author{Daming Cao}
\markboth{SECOND-ORDER IMAGE CHARACTERIZATION WITH DISTORTION AND SECRECY}{}
\maketitle
\begin{abstract}
\edit{Image characterization relates coding budgets to the probability carried by source cells and their output images. At second order, refining a cell can reveal unavailable information, while a change of reference measure can alter the relevant likelihood fluctuations. We establish image bounds that account for both effects. A finite-block partition theorem bounds the cost of merging component images by the joint conditional information spectrum of their label. It preserves continuous second-order profiles when that cost is negligible on the square-root blocklength scale. For reconstruction constraints, cell-dependent reference measures retain the actual distortion deficits in a single likelihood. A Gaussian evaluation gives matching second-order rates for Heegard--Berger coding with one uninformed decoder in the strict both-active regime, common reconstruction under a joint agreement and distortion criterion, and two-stage refinement with side information absent at the first decoder. A structured binary source illustrates partition stability even when its global density and conditional-divergence deficits are linear in blocklength. Posterior-mass images also sharpen extraction after a public refinement by charging its residual information rather than its alphabet size. These results preserve the joint constraints and the observations available for choosing the local coding experiment.}
\end{abstract}

\section{Introduction}
\edit{A message produced by a source encoder identifies a cell of possible source sequences. The descriptions available at the decoders constrain the size of an output image, or the size of each of its sections. Image characterization connects these constraints to the probability of successful reconstruction and provides a common route to coding converses~\cite[Chs.~15--16]{ck}. Weighted image measures extend the same viewpoint to distortion and secrecy~\cite[Ch.~17]{ck}.}

\edit{At a fixed nonzero error probability, the relevant refinement of a message budget is of order $\sqrt n$. A first-order image estimate with an unspecified $o(n)$ loss does not resolve this term. Moreover, several constraints can be active on the same successful event. Their joint fluctuation depends on the physical channel, and a coding choice may respond only to information available at the terminal that makes it. An encoder can adapt to its observed source type or energy, but not to decoder-only noise.}

\edit{These restrictions create two obstacles. An arbitrary message cell may not have the composition or covariance needed for a local approximation. Partitioning it into regular components introduces a label whose availability must be accounted for. Reconstruction imposes a different difficulty: acceptable sets may overlap, and a reference that ignores one decoder's reproduction can discard an active distortion constraint. The question is therefore how to refine or reweight an image while preserving its second-order operational meaning.}

\edit{We address the first obstacle by merging component images through posterior lists. The additional section budget is controlled by the information still needed to identify the component after the section index is known. This gives a two-sided finite-block comparison, not only a converse. For reconstruction, we attach the output reference to the reproduction assigned to each source cell. The resulting likelihood contains both distortion deficits, and a pair of likelihood comparisons cancels the same analytical description cost. In both routes, the joint Gaussian approximation is performed only after the admissible source-side choice has been fixed.}

\newpage
\subsection{Contributions}
\begin{enumerate}
\item \edit{We prove stability of simultaneous image constraints under recoverable source partitions. The comparison retains one joint exceptional event, projected sections, and prescribed actual costs. A common selector header implements the partition when the encoder observes its label and the required communication links are available. An explicit structured source attains a matching second-order rate beyond the small-density-deficit and small-conditional-divergence regimes.}
\item \edit{We characterize a two-level Gaussian distortion image using reconstruction-dependent references. Its converse starts from arbitrary measurable cells and matches a source-adaptive shell construction. The resulting rate applications are Gaussian Heegard--Berger coding with one uninformed decoder, common reconstruction under a joint agreement/distortion error criterion, and the corresponding successive-refinement corner.}
\item \edit{We extend the image bounds to coupled outputs, overlapping reconstruction policies, and posterior-mass constraints. The coupled-output bounds retain compatible helper sections. For extraction, the conditional information in a public refinement replaces a worst-case alphabet-size penalty and yields a second-order invariance result.}
\end{enumerate}

\subsection{Relation to existing methods}
\edit{The type and image methods of Csisz\'ar and K\"orner~\cite{ck} provide the first-order setting. Cao's single-output posterior characterization~\cite{imagecomp}, recalled as Lemma~\ref{comp:hi:thm:input}, supplies the contact identities and optimizer family at a prescribed rate pair. The regular weighted refinement is given in~\cite{wak}. The joint posterior extension uses this common geometry for every active section, whereas the conditional-composition and partition bounds are proved directly.}

\edit{Controlled auxiliaries provide another way to obtain regular image components. Tan~\cite{tanbc} combines Gaussian entropy dissipation with the strong functional representation lemma of Li and El Gamal~\cite{sfrl} to construct a message partition for the scalar Gaussian broadcast channel. The conditional label cost and the centered likelihood approximation have separate roles. Our partition comparison isolates the image operation justified by such a construction and applies to arbitrary finite output alphabets. The simultaneous seed-selection lemma ensures that one partition controls all specified outputs.}

\edit{For reconstruction with side information, likelihood bounds, Poisson matching, and type-deviation convergence give nonasymptotic and second-order constructions~\cite{wkt,poisson,lili}. Zhou~\cite{consolidated} establishes the discrete Wyner--Ziv characterization in the setting considered there. Xu and Zhou~\cite{xuzhou} obtain the scalar Gaussian characterization using source-dependent shells and a converse based on posterior regularization. Their reconstruction repair and smoothing estimates underlie the Gaussian analysis below. The additional cell-dependent reference preserves the uninformed reconstruction, and its distortion deficit must be included in the entropy and enhancement calculations. Common reconstruction~\cite{steinbergcr,ahmadi} and successive refinement~\cite{steinbergmerhav} then impose agreement or additional message-budget constraints on this experiment.}

\edit{Weighted-reference methods are also related to Brascamp--Lieb duality and reverse hypercontractivity~\cite{blfr,blinfo,lvhv}. For Gaussian Gray--Wyner coding, Zhou and Yu~\cite{zhouyu} establish the bivariate positive-common-rate characterization through covariance-shell covering. Section~\ref{sec:gw-comparison} compares its supporting-weight dispersion with the quadratic weighted-image expression. The identity identifies the same source fluctuation in the two representations on their shared bivariate regime.}

\edit{The order of operational choices matters beyond source coding. Oufkir and Berta~\cite{oufkir} show a second-order advantage of joint non-signalling channel conversion over concatenated coding and simulation. The allowable assistance differs from the source observations here, but both settings require one admissible transformation for the whole experiment. For security, extraction must likewise be distinguished from reliable agreement~\cite{htw}. Our posterior-mass results concern the actual extracted key and public information; communication-constrained key generation additionally requires a compatible reliable partition~\cite{hentila}.}

\subsection{Organization}
Sections~\ref{sec:reference}--\ref{sec:security} develop the image methods in Table~\ref{tab:methods}. Their matching applications, Gaussian comparisons, and nonmatching bounds appear in Sections~\ref{sec:tight}--\ref{sec:nontight} and Tables~\ref{tab:tight}--\ref{tab:gaps}. The appendices supply supporting proofs.

\begin{table}[H]
\centering\small
\caption{Image methods and the constraint retained.}\label{tab:methods}
\renewcommand{\arraystretch}{1.15}
\begin{tabularx}{\textwidth}{@{}>{\raggedright\arraybackslash}p{.29\textwidth}>{\raggedright\arraybackslash}X@{}}
\toprule Method & Operational content \\
\midrule
Single-output reference, Lemma~\ref{comp:hi:thm:input} & Posterior contacts and the optimizer family at a fixed rate pair. \\
Joint and marked images, Theorems~\ref{comp:hi:lem:vector-transfer},~\ref{new:dense} & Coupled outputs, projected sections, and actual costs on one event. \\
Recoverable partitions, Theorem~\ref{thm:partition-stability} & Two-sided merging controlled by the label's joint conditional information spectrum. \\
Spectral stability, Proposition~\ref{prop:spectral-stability} & Separate control of observable means, likelihood remainders, and joint covariance. \\
Reconstruction-constrained image, Theorem~\ref{thm:gaussian-distortion-image} & Both distortion deficits and the same analytical cost in paired likelihood bounds. \\
Policy and private images, Theorem~\ref{v5:policy-lp}, Proposition~\ref{prop:public-refinement} & Overlapping reconstruction policies; posterior mass after public information. \\
\bottomrule
\end{tabularx}
\end{table}

\begin{table}[H]
\centering\small
\caption{Direct matching rate applications.}\label{tab:tight}
\renewcommand{\arraystretch}{1.15}
\begin{tabularx}{\textwidth}{@{}>{\raggedright\arraybackslash}p{.32\textwidth}>{\raggedright\arraybackslash}X@{}}
\toprule Application & Scope and result \\
\midrule
Gaussian Heegard--Berger & One uninformed decoder; $0<a<p$, $0<b<aq/(a+q)$; Corollary~\ref{cor:hb-gaussian-rate}. \\
Gaussian common reconstruction & Joint agreement/distortion error; scalar and nested cases; Theorem~\ref{thm:cr-rates}. \\
Two-stage refinement & Side information absent at the first decoder; both-active corner; Corollary~\ref{cor:sr-image}. \\
Structured conditional source & Exponentially many regular components with a recoverable label; Proposition~\ref{prop:structured-cells}. \\
Compatible helper intersections & Active sections satisfying the stated mean and likelihood conditions; Theorem~\ref{comp:hi:thm:mixed}. \\
Conditional-type extraction & Prescribed types and public refinements of negligible conditional information; Theorem~\ref{thm:type-extraction}, Corollary~\ref{cor:public-extraction}. \\
\bottomrule
\end{tabularx}
\end{table}

\begin{table}[H]
\centering\small
\caption{Established Gaussian comparisons.}\label{tab:benchmarks}
\renewcommand{\arraystretch}{1.15}
\begin{tabularx}{\textwidth}{@{}>{\raggedright\arraybackslash}p{.32\textwidth}>{\raggedright\arraybackslash}X@{}}
\toprule Model & Characterization used \\
\midrule
Scalar Gaussian Wyner--Ziv & Source-observable joint Gaussian profile at fixed interior distortion~\cite{xuzhou}; Lemma~\ref{v5:wz-inner}. \\
Bivariate Gaussian Gray--Wyner & Supporting-weight characterization at positive common rate~\cite{zhouyu}; Section~\ref{sec:gw-comparison}. \\
\bottomrule
\end{tabularx}
\end{table}

\begin{table}[H]
\centering\small
\caption{Nonmatching bounds and the missing image estimate.}\label{tab:gaps}
\renewcommand{\arraystretch}{1.15}
\begin{tabularx}{\textwidth}{@{}>{\raggedright\arraybackslash}p{.32\textwidth}>{\raggedright\arraybackslash}X@{}}
\toprule Scope & Remaining requirement \\
\midrule
General helper codes & Recoverable regular components or small conditional-likelihood losses, with compatible means and encoder operations. \\
General both-lossy reconstruction & A matching asymptotic evaluation of the block reconstruction-policy inverse. \\
Two noisy reconstruction stages & Joint residual geometry under the original decoder observations. \\
Reliable-key generation & One public partition attaining both agreement and the private-image spectrum. \\
\bottomrule
\end{tabularx}
\end{table}


\section{Image Experiments and a Single-Output Reference}\label{sec:reference}
\subsection{Operational conventions}
All logarithms are natural. Alphabets are finite and fixed unless a Gaussian model is specified. The source distribution has full support after null letters are removed. A channel is a fixed stochastic kernel; zeros are allowed unless positivity is explicitly required. Write $T_n(q)$ for an $n$-type class and $\mathcal T_n(Q_{X|U}\mid u^n)$ for a conditional type class. A probability measure on posterior distributions and an auxiliary stochastic matrix are different objects. For a positive scalar $v$, $\phi_{vI_n}$ denotes the density of $\mathcal N(0,vI_n)$.

For $G\sim\mathcal N(0,\Sigma)$ define $\Phi_\Sigma(b)=\Pr\{G\le b\}$ componentwise, and write $\Phi$ and $Q=1-\Phi$ for the standard normal distribution and tail. Positive marginal variances make rectangle boundaries null even when $\Sigma$ is singular. For a finite source $X\sim p$,
\[
 Z_n=\sqrt n(\widehat p_{X^n}-p)\Rightarrow Z\sim\mathcal N(0,\diag(p)-pp^{\mathsf T}),
 \qquad H'_p(z)=-\sum_xz(x)\log p(x).
\]
A source-coding coefficient vector $L$ describes budgets $\log M_{i,n}=nr_i+\sqrt nL_i+o(\sqrt n)$. At a zero first-order rate, $L_i\ge0$. Closed coefficient regions allow favorable $o(\sqrt n)$ budget adjustments. Error is the probability that any required decoding or actual block-distortion constraint fails. Total variation is $\TV(P,Q)=\frac12\|P-Q\|_1$ and is evaluated for the actual extracted key, including decoding-error outcomes and all public information. We write $R_n=o_P(a_n)$ when $R_n/a_n$ converges to zero in probability and $R_n=O_P(a_n)$ when this ratio is tight. An eventwise statement on $E_n$ means the corresponding statement for $1_{E_n}R_n$ under the original law.

To obtain a rate theorem from an image bound, the admissible sets must be realized by an encoder and its decoders. In a conditional-type experiment, $u^n$ is fixed and known to the terminals that index their codebooks by it. In a helper problem, the helper message supplies this description.

\subsection{Intuitive proof outline and key insights}
\edit{The starting point is a code-induced cell and the image constraints imposed by its decoders. If the cell has a fixed conditional composition, each image score is a sum of independent channel increments after the source sequence is fixed. The main issue is then to retain all scores on the same event. The posterior method reaches a similar local experiment for arbitrary cells by controlling their distance from the optimizer family at the prescribed first-order rates.}

\edit{A cell may remain too irregular for either description. Refinement separates it into components, but gives the component-image problem an extra label. Posterior lists measure how much of that label is missing at each section. Merging only the labels on those lists recovers an admissible image with a controlled increase in its section budgets. A decoder that needs the label itself uses a common selector header, provided the coding architecture supplies the necessary link.}

\edit{For distortion, counting alone is insufficient because several source sequences may share an acceptable reproduction. A reference centered at the reproduction assigned to a cell converts its distortion constraint into a likelihood deficit. In the Gaussian two-level problem, both deficits enter with positive multipliers. This identifies the posterior geometry needed for a converse and explains why the coarse reconstruction cannot be replaced by an unconditional output reference.}

\edit{The final step separates observable variation from residual channel noise. A likelihood mean remains part of the admissible allocation, while only a controlled centered remainder is discarded. The construction must implement the same allocation before the channel noise is observed. This order is also essential for private images: public information changes the posterior mass available to extraction and must be included before the spectrum is evaluated.}

\subsection{The reference case \texorpdfstring{$U-X-Y$}{U-X-Y}}
The single-output posterior characterization of Cao~\cite{imagecomp} describes the image of a source cell at a prescribed first-order rate pair. We recall its contact identities and optimizer family, which will determine the mean and covariance constraints in the joint-image extension. The related regular weighted formulation appears in~\cite{wak}.
Let $X$ and $Y$ take values in finite alphabets $\mathcal X$ and
$\mathcal Y$, and write their law as $p(x)W(y|x)$.  Throughout this
section,
\begin{equation}
 p(x)>0,\qquad W(y|x)>0\qquad(x\in\mathcal X,y\in\mathcal Y).
 \label{comp:hi:eq:positive}
\end{equation}
Fix $c>1$, and define
\begin{equation}
 \phi_c=\min_{U-X-Y}\{I(U;X)+cH(Y|U)\}.
 \label{comp:hi:eq:value}
\end{equation}
The channel $W$ is fixed in this minimization.  Fix an optimizer-induced
pair $\boldsymbol r=(r_h,r_\Sigma)$ satisfying
\begin{equation}
 r_h>0,\qquad r_h+cr_\Sigma=\phi_c,
 \qquad a_0=H(p)-r_h.
 \label{comp:hi:eq:informative}
\end{equation}
In particular, both coordinates, rather than only their weighted sum,
are prescribed.

For a probability vector $q$ on $\mathcal X$, set
\begin{equation}
 \begin{aligned}
 F_c(q)&=D(q\Vert p)+cH(qW),\\
 h(q,y)&=-\log(qW)(y),\qquad
 m(q,x)=\sum_yW(y|x)h(q,y).
 \end{aligned}
 \label{comp:hi:eq:posterior-functions}
\end{equation}
For an auxiliary $T=P_{U|X}$, its posterior measure is
$\mu_T=\sum_uP_U(u)\delta_{P_{X|U=u}}$.  Conversely, a finite posterior
measure $\sum_uw_u\delta_{q_u}$ with barycenter $p$ defines
$T(u|x)=w_uq_u(x)/p(x)$.

\begin{lemma}[Single-output posterior image characterization {\cite{imagecomp}}]
\label{comp:hi:thm:input}
Under \eqref{comp:hi:eq:positive}--\eqref{comp:hi:eq:informative}, there is a unique
vector $j:\mathcal X\to\mathbb R$ such that
\begin{equation}
 \langle q,j\rangle\le F_c(q),\qquad
 \langle p,j\rangle=\phi_c.
 \label{comp:hi:eq:dual}
\end{equation}
Let $d(q)=F_c(q)-\langle q,j\rangle$ and
$\mathcal C=\{q:d(q)=0\}$.  The contact set $\mathcal C$ is compact
and contained in the interior of the probability simplex.  At every
contact,
\begin{equation}
 \log\frac{q(x)}{p(x)}+cm(q,x)=j(x)
 \qquad(x\in\mathcal X).
 \label{comp:hi:eq:contact}
\end{equation}
The complete optimizer fiber at the prescribed rate pair is
\begin{equation}
 \begin{split}
 \mathcal M_{\boldsymbol r}=\bigg\{\mu:\ &\supp\mu\subseteq\mathcal C,
 \quad\int q\,d\mu=p,\\
 &\int H(q)\,d\mu=a_0\bigg\}.
 \end{split}
 \label{comp:hi:eq:fiber}
\end{equation}
Its elements are probability measures.  They also satisfy
$\int H(qW)\,d\mu=r_\Sigma$.  Define
\begin{equation}
 \begin{aligned}
 v(q)&=c^2\sum_xq(x)\Var_{W(\cdot|x)}h(q,Y),\\
 \tau_-^2&=\min_{\mu\in\mathcal M_{\boldsymbol r}}\int v\,d\mu,\\
 \tau_+^2&=\max_{\mu\in\mathcal M_{\boldsymbol r}}\int v\,d\mu.
 \end{aligned}
 \label{comp:hi:eq:variance-range}
\end{equation}
Then $0<\tau_-\le\tau_+<\infty$, and all intermediate variances are
attainable by finite auxiliaries preserving $\boldsymbol r$.

The scalar image assertion has the following uniform form.  Suppose $p_n\to p$, $\mathcal X_n\subseteq T_n(p_n)$, and
output budgets $b_n$ satisfy
\begin{equation}
 \begin{gathered}
 \frac{\log|\mathcal X_n|}{n}\to a_0,\qquad
 \frac{b_n}{n}\to r_\Sigma,\\
 s_n=\frac{cb_n-\log|\mathcal X_n|
        -n\langle j+\log p,p_n\rangle}{\sqrt n}\to s.
 \end{gathered}
 \label{comp:hi:eq:input-scaling}
\end{equation}
For every $\mathcal D_n$ with $\log|\mathcal D_n|\le b_n$,
\begin{equation}
 \limsup_n\frac{1}{|\mathcal X_n|}
      \sum_{x^n\in\mathcal X_n}W^n(\mathcal D_n|x^n)
 \le f_{\boldsymbol r}(s),
 \label{comp:hi:eq:scalar-image}
\end{equation}
where
\begin{equation}
 f_{\boldsymbol r}(s)=\sup_{\tau\in[\tau_-,\tau_+]}\Phi(s/\tau)
 =\begin{cases}\Phi(s/\tau_-),&s\ge0,\\
                 \Phi(s/\tau_+),&s<0.
   \end{cases}
 \label{comp:hi:eq:scalar-envelope}
\end{equation}
The bound is uniform when the three first-order errors in
\eqref{comp:hi:eq:input-scaling} vanish uniformly and $s_n$ lies in a fixed
bounded interval.  Separate bounded second-order offsets for the input
and output cardinalities are not required.
\end{lemma}

For several output tests, the contact identity supplies a common mean decomposition and the complete fixed-pair family constrains covariance selection to the prescribed first-order rates. The next section combines this geometry with the joint conditional output law to determine the correlations among those tests.


\section{Joint Posterior Images and Compatible Sections}\label{comp:hi:sec:tools}\label{sec:posterior}
Apply Lemma~\ref{comp:hi:thm:input} with input $A=X$ and joint output $B=(B_1,\ldots,B_k)=Y$. Thus $p,W,c,j,\mathcal C$, and $\mathcal M_{\boldsymbol r}$ keep their reference-lemma meanings under this relabeling. The output coordinates need not be conditionally independent. A single image bound controls the aggregate rate, whereas separately encoded targets impose additional section constraints. All these constraints must be retained on one joint success event.

\subsection{Counting on the original success event}

Write $B=(B_1,\ldots,B_k)$ and let $S$ be a nonempty proper subset of
$[k]$.  Put $D=B_S$ and $C=B_{S^c}$.  A set
$\mathcal D\subseteq\mathcal B^n$ has section bound $M_S$ if
\begin{equation}
 \big|\{d^n:(d^n,c^n)\in\mathcal D\}\big|\le M_S
 \qquad\text{for every }c^n.
 \label{comp:hi:eq:section-cardinality}
\end{equation}
This property follows from separate encoding: after the helper index and
the complement source sequence are fixed, there are at most
$M_S=\prod_{j\in S}M_j$ remaining index tuples that can be decoded
correctly.

\begin{lemma}[Ancillary section spectrum]
\label{comp:hi:lem:ancillary}
Suppose $A-C-D$, and put
\begin{equation}
 t_S(b)=-\log P_{D|C}(d|c),\qquad
 m_S(a)=\E[t_S(B)|A=a].
 \label{comp:hi:eq:ancillary-score}
\end{equation}
If \eqref{comp:hi:eq:section-cardinality} holds, then for every $a^n$ and
$\gamma>0$,
\begin{equation}
 \PP\left\{B^n\in\mathcal D,
       \sum_i t_S(B_i)>\log M_S+\gamma\,\middle|\,A^n=a^n\right\}
 \le e^{-\gamma}.
 \label{comp:hi:eq:ancillary-spectrum}
\end{equation}
The same inequality holds after restricting $A^n$ to any nonempty
source subset and averaging under any law on that subset.
\end{lemma}
\begin{proof}
Condition additionally on $C^n=c^n$.  The Markov condition and the
memoryless channel give the original product law
$P_{D|C}^n(\cdot|c^n)$, independently of the fixed $a^n$.  Each word in
the indicated tail has conditional probability smaller than
$e^{-\gamma}/M_S$.  There are at most $M_S$ correct words in this
section.  Summing their probabilities and then averaging over
$C^n|a^n$ proves the assertion.  Restricting the input law does not
change either the section cardinality or this conditional product law.
\end{proof}

\begin{lemma}[Weighted section spectrum and one prediction correction]
\label{comp:hi:lem:weighted-spectrum}
Let $A^n$ have any law, let $B^n|A^n$ follow $W^n$, and denote its
induced output law by $Q_{B^n}$.  If
\eqref{comp:hi:eq:section-cardinality} holds, then every probability law
$R_{C^n}$ satisfies
\begin{equation}
 \PP\left\{B^n\in\mathcal D,
       -\log\frac{Q_{B^n}(B^n)}{R_{C^n}(C^n)}
              >\log M_S+\gamma\right\}\le e^{-\gamma}.
 \label{comp:hi:eq:weighted-spectrum}
\end{equation}
For a positive product reference $R_{C^n}=R_C^n$, define the source-only
posteriors $q_i=Q(A_i\in\cdot|A^{i-1})$.  Then
\begin{multline}
 \PP\biggl\{B^n\in\mathcal D,
       \sum_i[h(q_i,B_i)+\log R_C(C_i)]
                \\>\log M_S+2\gamma\biggr\}\le2e^{-\gamma}.
 \label{comp:hi:eq:weighted-prediction}
\end{multline}
For any finite collection of sections, all these predictions use the
same exceptional event for the joint output $B^n$.
\end{lemma}
\begin{proof}
Use the measure $\nu_R(d^n,c^n)=R_{C^n}(c^n)$.  Its mass on
$\mathcal D$ is at most $M_S$, by
\eqref{comp:hi:eq:section-cardinality}.  On the event in
\eqref{comp:hi:eq:weighted-spectrum}, each output probability is smaller
than $e^{-\gamma}\nu_R(d^n,c^n)/M_S$.  Summation proves that inequality.

For the prediction step, set
\begin{equation}
 \begin{aligned}
 s_i(b)&=(q_iW)(b),\\
 a_i(b)&=Q(B_i=b|B^{i-1}),\\
 \Lambda_B&=\sum_i\log\frac{s_i(B_i)}{a_i(B_i)}.
 \end{aligned}
 \label{comp:hi:eq:prediction-def}
\end{equation}
Under the filtration $\mathcal F_i=\sigma(A^i,B^i)$, the conditional
law of $B_i$ given $\mathcal F_{i-1}$ is $s_i$.  Indeed, after the
source prefix is fixed, past channel randomness gives no further
information about the next source symbol.  Consequently,
\begin{equation}
 \E\left[\frac{a_i(B_i)}{s_i(B_i)}\,\middle|\,
              \mathcal F_{i-1}\right]=\sum_{b:s_i(b)>0}a_i(b)\le1.
 \label{comp:hi:eq:prediction-martingale}
\end{equation}
The product of these ratios is a nonnegative supermartingale, and is a martingale when $W$ is positive;
therefore $\PP\{\Lambda_B<-\gamma\}\le e^{-\gamma}$.  The exact
identity
\begin{equation}
 \sum_i[h(q_i,B_i)+\log R_C(C_i)]
   =-\log\frac{Q_{B^n}(B^n)}{R_C^n(C^n)}-\Lambda_B
 \label{comp:hi:eq:weighted-identity}
\end{equation}
proves \eqref{comp:hi:eq:weighted-prediction}.  Since $\Lambda_B$ does not
depend on the chosen section, only the counting exceptional events
must be union-bounded separately.
\end{proof}

\subsection{A multivariate posterior transfer}

Fix two finite collections of sections, $\mathscr A$ and
$\mathscr W$.  For $S\in\mathscr A$, use the ancillary mark in
\eqref{comp:hi:eq:ancillary-score}.  For $S\in\mathscr W$, use
\begin{equation}
 \begin{aligned}
 C_S&=B_{S^c},\\
 t_{C_S}(b)&=-\log P_{C_S}(b_{S^c}),\\
 m_{C_S}(a)&=\E[t_{C_S}(B)|A=a].
 \end{aligned}
 \label{comp:hi:eq:complement-score}
\end{equation}
Define the vector of output marks
\begin{equation}
 \boldsymbol\psi(q,b)=\left(ch(q,b),\ (t_S(b))_{S\in\mathscr A},\
       (c[h(q,b)-t_{C_S}(b)])_{S\in\mathscr W}\right)
 \label{comp:hi:eq:mark-vector}
\end{equation}
and the covariance-valued continuous function
\begin{equation}
 \begin{aligned}
 V(q)&=\sum_aq(a)\Cov_{W(\cdot|a)}\boldsymbol\psi(q,B),\\
 \Sigma_\mu&=\int V(q)\,d\mu(q).
 \end{aligned}
 \label{comp:hi:eq:mark-covariance}
\end{equation}
All its entries are bounded and continuous under
\eqref{comp:hi:eq:positive}.  In the corresponding source vector,
\begin{equation}
 K(a)=\left(j(a),\ (m_S(a))_{S\in\mathscr A},\
                 (j(a)-cm_{C_S}(a))_{S\in\mathscr W}\right),
 \label{comp:hi:eq:source-vector}
\end{equation}
the only auxiliary-dependent quantities have disappeared.  This fact
will make source-type-dependent optimizer selection operational.

\begin{theorem}[Simultaneous posterior image bound]
\label{comp:hi:lem:vector-transfer}
Assume the input geometry of Lemma~\ref{comp:hi:thm:input}.  Let
$A^n$ be uniform on $\mathcal A_n\subseteq T_n(p_n)$, where
$p_n\to p$ and $n^{-1}\log|\mathcal A_n|\to a_0$.  Suppose a correct
set $\mathcal D_n$ has total cardinality at most $e^{b_{0,n}}$ and
section cardinality at most $e^{b_{S,n}}$.  Assume the ancillary Markov
conditions for $\mathscr A$.  Define normalized thresholds by
\begin{equation}
 \begin{aligned}
 s_{0,n}&=\frac{cb_{0,n}-\log|\mathcal A_n|
                 -n\langle j+\log p,p_n\rangle}{\sqrt n},\\
 s_{S,n}^{\rm a}&=\frac{b_{S,n}-n\langle m_S,p_n\rangle}{\sqrt n},\\
 s_{S,n}^{\rm w}&=\frac{cb_{S,n}-\log|\mathcal A_n|
              -n\langle j+\log p-cm_{C_S},p_n\rangle}{\sqrt n}.
 \end{aligned}
 \label{comp:hi:eq:vector-thresholds}
\end{equation}
If $s_n=(s_{0,n},(s_{S,n}^{\rm a}),(s_{S,n}^{\rm w}))$ stays bounded,
the simultaneous image bound is
\begin{equation}
 \PP\{B^n\in\mathcal D_n\}
 \le \sup_{\mu\in\mathcal M_{\boldsymbol r}}
                 \Phi_{\Sigma_\mu}(s_n)+o(1).
 \label{comp:hi:eq:vector-image}
\end{equation}
Here and below the displayed probability formula applies directly when
all retained Gaussian marginal variances are positive.  For zero
marginals it denotes the associated closed coefficient bound, obtained
by arbitrarily small positive threshold shifts.  The error is uniform
on bounded threshold windows and vanishing first-order neighborhoods.
\end{theorem}
\begin{proof}[Proof outline]
The aggregate posterior score decomposes into a deterministic type term, a nonnegative dual gap, a centered source residual, and a conditionally independent output vector. Each weighted section uses the same gap and source residual. Source-past martingale bounds remove the residual on low-gap histories and localize the posterior measure to the complete fixed-rate-pair family. Only then is the whole source block conditioned on to evaluate the joint output Gaussian law. High-gap successful histories have vanishing mass. Appendix~\ref{app:posterior} gives the identities, stopping estimate, and uniformity argument.
\end{proof}

\subsection{Preserving section means under an observable helper choice}

The weighted spectrum is valid without a Markov assumption, but its
mean need not equal the Slepian--Wolf conditional entropy.  The precise
condition for equality is useful both analytically and computationally.

\begin{lemma}[Complete-fiber mean compatibility]
\label{comp:hi:lem:mean-compatibility}
For an auxiliary $U-A-(D,C)$,
\begin{equation}
 H(D|C,U)-\E[-\log P(B|U)+\log P_C(C)]=I(U;C).
 \label{comp:hi:eq:mean-gap}
\end{equation}
Consequently the two means agree for every optimizer in
$\mathcal M_{\boldsymbol r}$ exactly when
\begin{equation}
 (q-p)W_C=0\qquad
   \text{for every }q\in\mathcal K_{\boldsymbol r},
 \quad
 \mathcal K_{\boldsymbol r}=
       \bigcup_{\mu\in\mathcal M_{\boldsymbol r}}\supp\mu.
 \label{comp:hi:eq:nullspace}
\end{equation}
Equivalently, the span of all these contact differences is contained
in the nullspace of $W_C$ acting on source tangent vectors.
\end{lemma}
\begin{proof}
The weighted expectation is $H(B|U)-H(C)$, whereas
$H(D|C,U)=H(B|U)-H(C|U)$.  Subtraction proves
\eqref{comp:hi:eq:mean-gap}.  The posterior $q=P_{A|U=u}$ induces complement
law $qW_C$, and the unconditional complement law is $pW_C$.  Thus
$U\perp C$ holds for a posterior measure precisely when these two
vectors agree on its support.  Taking the union of the supports proves
\eqref{comp:hi:eq:nullspace}.  The linear-span formulation is immediate.
\end{proof}

The Gaussian probability is optimized over the complete fiber, so the compatibility certificate must hold throughout that family.  If the contact differences span the entire
source tangent space, \eqref{comp:hi:eq:nullspace} forces $W_C$ to have
identical rows.  Dependence between $A$ and $C$ is permitted when the
optimal posterior variations span a proper subspace.

The matching helper construction also needs the \emph{conditional likelihoods}, not only their means, to agree at second order. Lemma~\ref{comp:hi:lem:implementation} in Appendix~\ref{app:posterior} proves this for every fixed finite family of admissible optimizer branches. A type header describes the selected branch, while the target encoders continue to use fixed local binning maps.


\section{Conditional-Composition Images with Actual-Distortion Marks}\label{new:images}\label{sec:conditional}
The preceding posterior theorem handles arbitrary cells through a supporting entropy geometry. A second description is useful when a conditional composition is specified. It permits different output projections to be counted and keeps actual reconstruction costs without replacing them by their conditional means.

\subsection{Projected images and marked success}
Fix a joint $n$-type $Q_{UX}$, $u^n\in\mathcal T_n(Q_U)$, and a nonempty cell
\[
 \mathcal A_n\subseteq\mathcal T_n(Q_{X|U}\mid u^n)=:\mathcal T_n,
 \qquad \delta_n=\log\frac{|\mathcal T_n|}{|\mathcal A_n|}.
\]
The input $X^n$ is uniform on $\mathcal A_n$. The finite memoryless channel $W_{B|X}$ has joint output $B=(B_1,\ldots,B_k)$. Its coordinates can be dependent given $X$. Let $P^{\mathcal A}$ denote the resulting joint law, and let $R$ be the output of the unconditioned product input $\prod_iQ_{X|U=u_i}$:
\[
 R(b^n\mid u^n)=\prod_{i=1}^n(Q_{X|U=u_i}W)(b_i).
\]
For each of finitely many rows $j$, select disjoint coordinate sets $S_j,T_j\subseteq[k]$, with $S_j\ne\varnothing$, and write $V_j=B_{S_j}$, $C_j=B_{T_j}$. An image $\mathcal D\subseteq\mathcal B^n$ is admissible at budgets $b=(b_j)$ if the projection of $\mathcal D$ onto $(V_j^n,C_j^n)$ contains at most $e^{b_j}$ different $v_j^n$ for each fixed $c_j^n$. Thus $T_j=\varnothing$ specifies a marginal-image budget, while $S_j\cup T_j=[k]$ specifies a section of the full image.

Define the reference scores
\begin{equation}
 J_{j,n}=\sum_{i=1}^n j_{j,Q}(B_i,u_i),\qquad
 j_{j,Q}(b,u)=-\log Q_{V_j|C_j,U}(v_j\mid c_j,u).
 \label{new:dense-score}
\end{equation}
All scores are evaluated on their sampling support. For fixed bounded nonnegative functions $d_\ell(u,x,b)$, define the actual marks
\[
 D_{\ell,n}=\sum_{i=1}^n d_\ell(u_i,X_i,B_i),\qquad
 \mathcal E_n(t)=\{D_{\ell,n}\le t_\ell\text{ for every }\ell\}.
\]
For example $d_\ell(u,x,b)=d(x,g_\ell(u,b))$ records the actual distortion of a fixed reconstruction rule. These marks need not be determined by the output alone. Let
\[
 P_{\rm im}^*(b,t)=\sup_{\mathcal D\text{ admissible}}
    P^{\mathcal A}\{B^n\in\mathcal D,\ \mathcal E_n(t)\}.
\]
Omitting the marks recovers the ordinary image problem.

\begin{theorem}[Simultaneous projected and marked images]\label{new:dense}
Let $s$ be the number of image rows and put $\Delta_n=\delta_n+|\mathcal X||\mathcal U|\log(n+1)$. For every $\gamma>0$,
\begin{align}
 \Pr\{J_{j,n}\le b_j\ \forall j,\ \mathcal E_n(t)\}
 &\le P_{\rm im}^*(b,t),\label{new:dense-lower}\\
 P_{\rm im}^*(b,t)
 &\le\Pr\{J_{j,n}\le b_j+\Delta_n+2\gamma\ \forall j,\ \mathcal E_n(t)\}
       +2s e^{-\gamma}.\label{new:dense-upper}
\end{align}
There is no independence assumption between $U$ and any complement $C_j$, and channel transitions may vanish.

For compatible $Q_n\to Q$, let $\xi_Q(u,x,b)$ be the vector consisting of all $j_{j,Q}(b,u)$ and all $d_\ell(u,x,b)$. Set
\begin{equation}
 m(Q)=\E_{QW}\xi_Q(U,X,B),\qquad
 \Gamma(Q)=\sum_{u,x}Q(u,x)\Cov_{W(\cdot|x)}\xi_Q(u,x,B).
 \label{eq:marked-covariance}
\end{equation}
The image-coordinate means are $H_Q(V_j|C_j,U)$. Suppose all retained marginal variances in $\Gamma(Q)$ are positive. If $\delta_n=o(\sqrt n)$ and
\[
 (b_n,t_n)=nm(Q_n)+\sqrt nL+o(\sqrt n),
\]
then $P_{\rm im}^*(b_n,t_n)\to\Phi_{\Gamma(Q)}(L)$. Singular joint covariance is permitted. If instead $\limsup\delta_n/\sqrt n\le\beta$, the corresponding lower and upper limiting bounds are $\Phi_{\Gamma(Q)}(L)$ and $\Phi_{\Gamma(Q)}(L+\beta e_{\rm im})$, where $e_{\rm im}$ is one on image coordinates and zero on distortion coordinates.
\end{theorem}

\begin{proof}
Under the product input every member of $\mathcal T_n$ has probability $e^{-nH_Q(X|U)}$. Hence the conditioning event $\{X^n\in\mathcal A_n\}$ has probability $\pi\ge e^{-\Delta_n}$. For every output projection $O$, marginalization preserves
\[
 \frac{P_O^{\mathcal A}}{R_O}\le e^{\Delta_n}.
\]
Also $P\{\log(P/R)<-\gamma\}\le e^{-\gamma}$ for probability laws $P\ll R$. The exact conditional identity for row $j$ is
\begin{equation}
 J_{j,n}=-\log P^{\mathcal A}(V_j^n\mid C_j^n)
       +\log\frac{P_{V_j^nC_j^n}^{\mathcal A}}{R_{V_j^nC_j^n}}
       -\log\frac{P_{C_j^n}^{\mathcal A}}{R_{C_j^n}}.
 \label{eq:projected-identity}
\end{equation}
For each $c_j^n$, an admissible projection contains at most $e^{b_j}$ values. Its mass where $-\log P^{\mathcal A}(V_j^n|C_j^n)>b_j+\gamma$ is at most $e^{-\gamma}$. Averaging over $C_j^n$ and excluding the lower tail of the final likelihood ratio in \eqref{eq:projected-identity} proves \eqref{new:dense-upper}. All exclusions intersect the same original image event; retaining $\mathcal E_n(t)$ only reduces their probabilities.

For the lower bound take $\mathcal D_0=\{b^n:J_{j,n}(b^n)\le b_j\text{ for every }j\}$ on the reference support. For fixed $c_j^n$, every admitted $v_j^n$ has conditional reference probability at least $e^{-b_j}$. There are at most $e^{b_j}$ such values. Every projected section of $\mathcal D_0$ is a subset of this threshold section, so $\mathcal D_0$ is admissible. Intersecting its event with the actual marks gives \eqref{new:dense-lower}.

Given any $x^n\in\mathcal A_n$, the vector of scores and marks is a sum of independent channel increments. Since $(u^n,x^n)$ has a fixed joint type, its distribution is the same for every such input, up to permutation of the coordinates. Its mean and covariance are exactly $nm(Q_n)$ and $n\Gamma(Q_n)$. If $w_0$ is the smallest positive transition of $W$, then on a sampled $(u,x,b)$,
\[
 0\le j_{j,Q}(b,u)\le-\log Q(x|u)-\log w_0.
\]
The averaged third moments are uniformly bounded on a finite simplex, and bounded actual marks preserve this property. The identity $q\log^2(1/q)\to0$ as $q\downarrow0$, together with the same conditional moment bound for a vanishing $U$ column, proves covariance continuity. The independent-array central limit theorem is therefore uniform on compact limiting covariance sets with positive diagonal. Such Gaussians assign zero probability to rectangle boundaries even if jointly singular. Taking $\gamma=3\log(n+1)$ in the finite-block bounds proves the stated limits and the deficit bracket.
\end{proof}

\begin{corollary}[Two output images]\label{cor:two-images}
In the experiment $U-X-(Y,Z)$, impose separate image sizes $e^{b_Y}$ and $e^{b_Z}$. For a conditional-type cell with $\delta_n=o(\sqrt n)$, put
\[
 b_Y=nH_{Q_n}(Y|U)+\sqrt nL_Y+o(\sqrt n),\quad
 b_Z=nH_{Q_n}(Z|U)+\sqrt nL_Z+o(\sqrt n).
\]
Then the optimal probability that both outputs fall in their respective images has limit
\[
 \Phi_\Gamma(L_Y,L_Z),\quad
 \Gamma=\E_{Q_{UX}}\Cov_{W(\cdot|X)}
 \begin{pmatrix}-\log Q_{Y|U}(Y|U)\\-\log Q_{Z|U}(Z|U)\end{pmatrix},
\]
provided its diagonal is positive. Additional bounded actual-distortion marks have the joint limit specified in Theorem~\ref{new:dense}.
\end{corollary}
\begin{proof}
Use rows $(V_1,C_1)=(Y,\varnothing)$ and $(V_2,C_2)=(Z,\varnothing)$. Every admissible joint image is contained in the product of its two projections; replacing it by that product preserves both budgets and cannot decrease the success probability. Theorem~\ref{new:dense} therefore applies to the original separate-image problem.
\end{proof}
\edit{If a distortion mark depends only on $(u,x)$, its sum is fixed by the joint type. An infeasible deterministic budget makes the joint success probability zero; a feasible one can be removed before applying the positive-diagonal Gaussian statement. Thus no success probability is assigned by a Gaussian approximation to that deterministic coordinate.} The off-diagonal covariance is determined by the physical joint channel. Applying a scalar image theorem twice and assigning two separate error tolerances would not give this joint limit.

\subsection{Conditional likelihood control for irregular cells}
Let $\widetilde P$ be the unconditioned product input and channel law and write $\mathsf A=\{X^n\in\mathcal A_n\}$. For each row define
\begin{align}
 \ell_{j,n}&=\log\frac{P^{\mathcal A}(V_j^n\mid C_j^n)}{R(V_j^n\mid C_j^n)},\label{v5:conditional-likelihood}\\
 k_{j,n}(c)&=-\log\widetilde P(\mathsf A\mid C_j^n=c),\nonumber\\
 K_{j,n}&=D(P_{V_j^nC_j^n}^{\mathcal A}\|R_{V_j^nC_j^n})
              -D(P_{C_j^n}^{\mathcal A}\|R_{C_j^n})\ge0.\label{v5:conditional-deficit}
\end{align}
The last quantity is conditional relative entropy. A complement can explain a large part of the restriction defining the cell, making this quantity much smaller than $\delta_n$.

\begin{theorem}[Screened projected-image comparison]\label{v5:screened}
For every nonempty conditional-type cell, every $\gamma>0$, and deterministic $\eta_j>0$,
\begin{align}
 P_{\rm im}^*(b,t)
 &\le\Pr\{J_{j,n}\le b_j+\eta_j+\gamma\ \forall j,\ \mathcal E_n(t)\}
       +\sum_j\left(e^{-\gamma}+\Pr\{\ell_{j,n}>\eta_j\}\right),\label{v5:screened-bound}\\
 P_{\rm im}^*(b,t)
 &\le\Pr\{J_{j,n}\le b_j+k_{j,n}(C_j^n)+\gamma\ \forall j,\ \mathcal E_n(t)\}
       +s e^{-\gamma}.\label{v5:posterior-bound}
\end{align}
Furthermore,
\begin{equation}
 \Pr\{\ell_{j,n}>\eta\}\le\frac{K_{j,n}+1}{\eta}.\label{v5:conditional-tail}
\end{equation}
The Gaussian conclusion of Theorem~\ref{new:dense} remains valid if its density assumption is replaced by $\max_jK_{j,n}=o(\sqrt n)$, or more generally by $(\ell_{j,n})_+/\sqrt n\to0$ in probability for every retained row.
\end{theorem}
\begin{proof}
Bayes' rule gives
\[
 \ell_{j,n}=\log\widetilde P(\mathsf A\mid V_j^n,C_j^n)
              -\log\widetilde P(\mathsf A\mid C_j^n)
          \le k_{j,n}(C_j^n).
\]
Use conditional section counting as in Theorem~\ref{new:dense} and the identity $J_{j,n}=-\log P^{\mathcal A}(V_j^n|C_j^n)+\ell_{j,n}$. This proves both probability bounds without changing the actual-distortion event. Conditional likelihood-ratio summation yields $\Pr\{\ell_{j,n}<-u\}\le e^{-u}$. Hence $\E(\ell_{j,n})_-\le1$. The relative-entropy chain rule gives $\E\ell_{j,n}=K_{j,n}$, so $\E(\ell_{j,n})_+\le K_{j,n}+1$ and Markov's inequality proves \eqref{v5:conditional-tail}. If $K_{j,n}=o(\sqrt n)$, choose $\eta_{j,n}=[(K_{j,n}+1)\sqrt n]^{1/2}$ and $\gamma=3\log(n+1)$. Both the normalized threshold loss and the exceptional probability vanish. Under convergence in probability, a deterministic diagonal choice of $\eta_{j,n}=o(\sqrt n)$ does the same. The feasible-image lower bound and the conditional-composition Gaussian law are unchanged.
\end{proof}

For example, take input $(V,T)$ and a cell
\[
 \mathcal A_n=\{(v^n,t^n):v^n\in\mathcal C_n\subseteq\mathcal T_n(Q_V),\quad
                   t^n\in\mathcal T_n(Q_{T|V}\mid v^n)\}.
\]
If the complement observes $V^n$, only the conditional-type event remains after it is fixed. Thus $k_{j,n}(V^n)\le|\mathcal V||\mathcal T|\log(n+1)$ even when $\log|\mathcal T_n(Q_V)|-\log|\mathcal C_n|=\Theta(n)$. The other output can retain information about $T$ beyond $V$: conditioning on the complement removes the sparse restriction on $V^n$, while the conditional-type uncertainty remains. A total-image row has an empty complement and requires its own bound.



\section{Image Stability Under Recoverable Source Partitions}\label{sec:partitions}
An irregular source cell need not admit a useful product reference. A refinement can separate it into regular components, but revealing the component to every terminal may change the problem. The relevant question is therefore quantitative: how much image budget is needed to remove that extra knowledge? The answer is determined by the conditional information of the component label at each section, rather than by the number of components alone.

\subsection{One partition and simultaneous section budgets}
Let $X$ be an arbitrary source, $B$ its joint output on a finite alphabet, and $J=f(X)$ a countably valued source partition. For each of a fixed number $s$ of constraints, let $(V_r,C_r)=\pi_r(B)$ be a projection of the output. For an output set $\mathcal D$, its $r$th section is
\[
 \mathcal D_r(c)=\{v:(v,c)\in\pi_r(\mathcal D)\}.
\]
A budget vector $\boldsymbol b$ permits $|\mathcal D_r(c)|\le e^{b_r}$ for every $r,c$. A total-image constraint is included by taking $C_r$ constant and $V_r=B$. If an actual-cost event $\mathcal E\subseteq\mathcal X\times\mathcal B$ is prescribed, the same event is retained throughout. Define
\begin{align}
 \Gamma(\boldsymbol b)&=\sup_{\mathcal D}\Pr\{B\in\mathcal D,\ (X,B)\in\mathcal E\},\nonumber\\
 \Gamma^J(\boldsymbol b)&=\sup_{(\mathcal D_j)}
       \Pr\{B\in\mathcal D_J,\ (X,B)\in\mathcal E\},\label{eq:partition-values}
\end{align}
where every image in either supremum satisfies the same section budgets. All quantities refer to the original source--channel law. A common side-information variable already known to the relevant terminals can be conditioned on in these definitions.

Put $T_r=-\log P(J\mid C_r)$, evaluated at $(J,C_r)$. The following finite-block bound does not impose independence among the sections.
\begin{theorem}[Image stability under a recoverable partition]\label{thm:partition-stability}
For every positive vector $\boldsymbol\tau$,
\begin{equation}
 \Gamma(\boldsymbol b)\le\Gamma^J(\boldsymbol b)
 \le\Gamma(\boldsymbol b+\boldsymbol\tau)
      +\Pr\{\exists r:T_r>\tau_r\}.
 \label{eq:partition-sandwich}
\end{equation}
In particular, the final term is at most
$\sum_r H(J\mid C_r)/\tau_r$. The same assertion holds for section measures in place of cardinalities, provided that for each $r,c$ the same measure is used for every component image.
\end{theorem}
\begin{proof}
Define the posterior list
$\mathcal L_r(c)=\{j:P(j\mid c)\ge e^{-\tau_r}\}$;
its size is at most $e^{\tau_r}$. Given component images $(\mathcal D_j)$, merge them as
\begin{equation}
 \mathcal D=\{b:\text{some }j\text{ satisfies }b\in\mathcal D_j
       \text{ and }j\in\bigcap_r\mathcal L_r(C_r(b))\}.
 \label{eq:partition-merge}
\end{equation}
For fixed $r,c$, the projected section of $\mathcal D$ is contained in the union of the corresponding sections of at most $e^{\tau_r}$ component images. Its size is therefore at most $e^{b_r+\tau_r}$. The same argument uses subadditivity for a common section measure. Whenever $B\in\mathcal D_J$ and all true-label lists contain $J$, the merged image also contains $B$; the prescribed event $\mathcal E$ is unchanged. Taking suprema proves the upper bound. Choosing every component image equal proves the lower bound. Finally, $T_r\ge0$ and $\E T_r=H(J\mid C_r)$ give the entropy estimate.
\end{proof}

The spectrum in \eqref{eq:partition-sandwich} is the primary condition. Conditional entropy supplies a convenient sufficient estimate, but is not necessary.
\begin{corollary}[Preservation of a second-order image profile]\label{cor:partition-profile}
For a sequence of these experiments, suppose
\begin{equation}
 \max_{r\le s}\frac{-\log P(J_n\mid C_{r,n})}{\sqrt n}
       \longrightarrow0\quad\text{in probability}.
 \label{eq:partition-small}
\end{equation}
Let $\boldsymbol b_n=n\boldsymbol r+\sqrt n\boldsymbol L$ and suppose
$\Gamma^{J_n}_n(\boldsymbol b_n)\to F(\boldsymbol L)$, locally with respect to fixed offsets, where $F$ is continuous at $\boldsymbol L$. Then
$\Gamma_n(\boldsymbol b_n)\to F(\boldsymbol L)$.
A sufficient condition for \eqref{eq:partition-small} is
$\max_r H(J_n\mid C_{r,n})=o(\sqrt n)$.
\end{corollary}
\begin{proof}
A diagonal choice gives deterministic $\tau_n=o(\sqrt n)$ such that
$\Pr\{\max_rT_{r,n}>\tau_n\}\to0$. Theorem~\ref{thm:partition-stability} at budget $\boldsymbol b_n-\tau_n\one$ implies
\[
 \Gamma_n^{J_n}(\boldsymbol b_n-\tau_n\one)-o(1)
 \le\Gamma_n(\boldsymbol b_n)\le\Gamma_n^{J_n}(\boldsymbol b_n).
\]
For any fixed $\eta>0$, monotonicity bounds the left side below by the component profile at offset $\boldsymbol L-\eta\one$ for large $n$. Take $n\to\infty$ and then $\eta\downarrow0$. Markov's inequality proves the sufficient entropy condition.
\end{proof}
A total-image row in this theorem charges $-\log P(J_n)$, since its complement is constant. A partition recoverable only from a nonconstant complement does not make that total-image cost disappear.

\subsection{When a label must be decoded}
Merging sets preserves image mass, but a reconstruction policy may require the label itself. The next statement supplies one common header when the encoder knows $J$ and each relevant decoder observes its own $C_r$.
\begin{lemma}[A common selector header]\label{lem:selector-header}
For $t,\gamma>0$, a deterministic hash of $J$ into
$L=\lceil e^{t+\gamma}\rceil$ labels permits all $s$ decoders to recover $J$ with joint error at most
\begin{equation}
 \Pr\{\exists r:-\log P(J\mid C_r)>t\}+s e^{-\gamma}.
 \label{eq:selector-error}
\end{equation}
Thus \eqref{eq:partition-small} gives an $o(\sqrt n)$-nat header with vanishing error, when the communication architecture allows that header.
\end{lemma}
\begin{proof}
Assign independent uniform labels in $[L]$ to the possible values of $J$. Decoder $r$ searches its posterior list with threshold $e^{-t}$, which contains at most $e^t$ labels. Outside the true-label exclusion event, any competing label collides with probability $1/L$, so the expected joint collision probability is at most $s e^t/L\le s e^{-\gamma}$. Average and fix one hash. This remains valid for a countable label alphabet. A slow threshold sequence and $\gamma_n\to\infty$ with $t_n+\gamma_n=o(\sqrt n)$ prove the last assertion.
\end{proof}
The header is a message of the encoder that observes $J$. It does not reveal a distributed source to another encoder. For distortion policies, the original joint distortion event is retained whenever the label is recovered; its failure probability is added only once.

\subsection{Constructing one controlled partition}
The previous results reduce irregular-cell analysis to a partition construction. A synthetic observation can provide such a partition through the strong functional representation lemma of Li and El Gamal~\cite{sfrl}. The same seed must work for all active outputs. Tan's Gaussian broadcast converse uses this principle to preserve both component geometry and the cost of reassembling an output image~\cite{tanbc}.

Consider $T-X-(B_1,\ldots,B_m,E_1,\ldots,E_s)$ on standard Borel spaces with $I(X;T)<\infty$. Fix reference kernels $Q_{j|t}$ and finite quantities
\[
 d_j=\E_TD(P_{B_j|T}\|Q_{j|T}),\qquad
 h_r=I(X;T\mid E_r),\qquad
 \kappa=\log\left(1+\frac{I(X;T)}{\log2}\right)+4\log2.
\]
\begin{lemma}[Simultaneous auxiliary-to-partition transfer]\label{lem:simultaneous-partition}
There is a deterministic countably labelled partition $J=f(X)$, and a reference parameter $t(J)$, such that, with $C=2(m+s)$,
\begin{align}
 H(J\mid E_r)&\le C(h_r+\kappa),\quad r\le s,\label{eq:partition-SFRL-h}\\
 \E_JD(P_{B_j|J}\|Q_{j|t(J)})&\le C(d_j+\kappa),\quad j\le m.
 \label{eq:partition-SFRL-d}
\end{align}
Additional nonnegative approximation errors can be included in the same seed-selection argument by increasing $C$ by twice their number.
\end{lemma}
The proof is in Appendix~\ref{app:partition-tools}. It generates all physical outputs conditionally on $X$, independently of the representation seed, and then fixes one seed for the entire experiment. In particular, $I(X^n;T_n)=O(n)$ and $h_{r,n}=o(\sqrt n)$ yield the label condition in Corollary~\ref{cor:partition-profile}, while $d_{j,n}=o(n)$ is preserved. A component-image theorem must still establish its fluctuation approximation; an $o(n)$ divergence is not itself that approximation.

\subsection{Separating drift from centered errors}\label{sec:spectral-control}
Distributional approximation has two different tasks. A predictable likelihood drift must be retained in the admissible mean or canceled by a reference identity. Only the centered remainder is removed by a variance estimate. This separation is needed even when the component law is close in relative entropy per symbol.

For example, take $P_n(0)=P_n(1)=1/2$ and $Q_n(1)=e^{-\sqrt n}/2$. For $\ell_n=\log(dP_n/dQ_n)$,
\begin{equation}
 D(P_n\|Q_n)=\tfrac12\sqrt n+O(1)=o(n),\qquad
 \Var_{P_n}(\ell_n)/n\longrightarrow\tfrac14.
 \label{eq:KL-counterexample}
\end{equation}
Thus even centering the likelihood does not make this error negligible at second order. \edit{A change of reference is useful only after its predictable mean has been included or canceled. The Gaussian reconstruction argument below performs that cancellation through its explicit distortion identity. The remaining random error is controlled separately.}

The following criterion records what suffices. Let $\mathcal H_n$ contain the information available before residual channel noise. Let $m_n,\Sigma_n$ be $\mathcal H_n$-measurable, with
$\min_j(\Sigma_n)_{jj}\ge\sigma_*^2>0$, and suppose
\[
 \E\sup_z\left|P(Z_n\le z\mid\mathcal H_n)-\Phi_{\Sigma_n}(z)\right|\le\beta_n.
\]
For the actual score vector write
$S_n=m_n+Z_n+(R_n+T_n)/\sqrt n$, where
$\E\|T_n\|_\infty\le h_n$ and $\E\|R_n\|^2\le v_n$.
\begin{proposition}[Quantitative spectral stability]\label{prop:spectral-stability}
For every $\eta>0$ and every $\mathcal H_n$-measurable threshold $b_n$,
\begin{equation}
 \left|P(S_n\le b_n)-\E\Phi_{\Sigma_n}(b_n-m_n)\right|
 \le\beta_n+C\eta+\frac{h_n}{\eta\sqrt n}
                    +\frac{v_n}{n\eta^2},
 \label{eq:spectral-stability}
\end{equation}
where $C$ depends only on the fixed dimension and $\sigma_*$. In particular, $h_n=o(\sqrt n)$, $v_n=o(n)$, and $\beta_n\to0$ imply a vanishing comparison error.
\end{proposition}
\begin{proof}
Outside events of probability at most $h_n/(\eta\sqrt n)+v_n/(n\eta^2)$, each score threshold moves by at most $2\eta$. Apply the conditional rectangle comparison at the enlarged and reduced thresholds. A Gaussian marginal with variance at least $\sigma_*^2$ assigns at most $4\eta/(\sqrt{2\pi}\sigma_*)$ mass to a strip of width $4\eta$; a union bound over coordinates proves \eqref{eq:spectral-stability}. No independence between the error terms and $Z_n$ is required.
\end{proof}
Conditional centering of $R_n$ is normally how its second-moment bound is established. A Gaussian covariance may be singular as long as every retained diagonal entry has the displayed positive lower bound. Deterministic coordinates require a separate atom-sensitive argument. The proposition establishes a second-order transfer, not a logarithmic rate remainder.


\section{Distorted Images and Weighted Image Measures}\label{v5:distortion}\label{sec:distortion}
A distortion requirement changes which source words can share one image. The methods in this section retain either the complete reconstruction policy or an inequality for the weighted mass of every successful subset. Both use the actual block acceptance relation. The conditional marks in Theorem~\ref{new:dense} evaluate distortion fluctuations at a fixed source--reconstruction composition. Optimizing over the reconstruction policies requires the additional image construction developed below.

\subsection{Reconstruction policies and a fractional cell bound}
In this subsection $x,y,a,z$ may each denote an entire finite block. Let $p(x)W(y|x)$ be its joint source and side-information law. One encoder observes $x$ and sends a common index in $[M]$. The uninformed decoder produces $a$, while the informed decoder produces $z=\psi(y)$. Let $b_0(x,a),b_1(x,z)\in\{0,1\}$ indicate satisfaction of the two prescribed \emph{actual} distortion constraints. The alphabets and indicators are arbitrary and finite.

A policy $s=(a,\psi)$ has source-dependent success probability
\[
 q_s(x)=b_0(x,a)\sum_yW(y|x)b_1(x,\psi(y)).
\]
For a nonempty cell $C\subseteq\mathcal X$, define
\begin{align}
 v(C)&=\max_s\sum_{x\in C}p(x)q_s(x)\label{v5:cell-value}\\
 &=\max_a\sum_y\max_z\sum_{x\in C}
 p(x)W(y|x)b_0(x,a)b_1(x,z).\nonumber
\end{align}
The order of the maxima is essential: the weak reconstruction is fixed before $y$, and only the informed reconstruction can depend on $y$. Distinct acceptable sets are not required to be disjoint.

For a real budget $m>0$ let
\begin{equation}
 \mathsf F(m)=\max_{w_C\ge0}
 \left\{\sum_Cw_Cv(C):\ \sum_Cw_C\le m,
          \ \sum_{C\ni x}w_C\le1\ \text{for every }x\right\}.
 \label{v5:fractional-program}
\end{equation}
This finite linear program enumerates all nonempty source cells. Its size grows with the number of source subsets; its role here is to characterize the finite-block bound and identify the quantity to evaluate asymptotically.

\begin{theorem}[Distortion-policy relaxation and rounding]\label{v5:policy-lp}
Let $P_{\mathrm c}^*(M)$ be optimal joint success with $M$ messages. For every integer $M\ge1$ and every $c>0$,
\begin{equation}
 P_{\mathrm c}^*(M)\le\mathsf F(M),\qquad
 P_{\mathrm c}^*(\lceil cM\rceil)\ge(1-e^{-c})\mathsf F(M).
 \label{v5:lp-sandwich}
\end{equation}
The rounding assertion holds also for a real input budget $m$ in place of $M$. The exact dual is
\begin{equation}
 \mathsf F(m)=\min_{\alpha_x,\beta\ge0}
 \left\{\sum_x\alpha_x+m\beta:
       \beta+\sum_{x\in C}\alpha_x\ge v(C)\ \forall C\right\}.
 \label{v5:lp-dual}
\end{equation}
Equivalently, its constraints are
\begin{equation}
 \beta\ge\sup_s\sum_x\bigl(p(x)q_s(x)-\alpha_x\bigr)_+.
 \label{v5:policy-dual}
\end{equation}
Let $\mathsf M(\eps)=\inf\{m\ge1:\mathsf F(m)\ge1-\eps\}$. For $0<\delta<\eps<1$,
\begin{equation}
 \mathsf M(\eps)\le M^*(\eps)
 \le\left\lceil\log(1/\delta)\,\mathsf M(\eps-\delta)\right\rceil.
 \label{v5:lp-rate-bracket}
\end{equation}
Infinite values are allowed when the distortion constraints are infeasible.
\end{theorem}

\begin{proof}
An encoder partitions $\mathcal X$ into at most $M$ nonempty cells. Optimizing each pair of decoders separately gives the exact identity
\[
 P_{\mathrm c}^*(M)=\max_{\text{partitions }\{C_j\}_{j\le M}}
                           \sum_jv(C_j).
\]
A partition is feasible in \eqref{v5:fractional-program} by assigning weight one to its cells. This proves the upper bound. Linear programming duality gives \eqref{v5:lp-dual}. Maximizing
$v(C)-\sum_{x\in C}\alpha_x$ first over the policy and then over $C$ selects exactly the positive summands, proving \eqref{v5:policy-dual}.

To prove rounding, fix feasible weights and a maximizing policy $s_C$ for every cell. Independently draw $N=\lceil cm\rceil$ cell-policy pairs, choosing $C$ with probability $w_C/m$ and an empty dummy with the remaining probability. These lists, including the policies, are fixed codebooks. At source $x$, the encoder selects the drawn pair containing $x$ with the largest $q_{s_C}(x)$; a fixed default is used if none contains $x$. This uses $x$ alone, not the realized $y$.

For $t\in[0,1]$ set
\[
 u_x(t)=\sum_{C\ni x:\ q_{s_C}(x)\ge t}w_C\le1.
\]
The probability that a drawn list contains a policy of success at least $t$ for this $x$ is
\[
 1-(1-u_x(t)/m)^N\ge1-e^{-cu_x(t)}
                         \ge(1-e^{-c})u_x(t).
\]
The final inequality is the chord bound for the concave function $1-e^{-cu}$ on $[0,1]$. Integrating over $t$ and summing with weights $p(x)$ shows that the expected success of this code is at least
$(1-e^{-c})\sum_Cw_Cv(C)$. Fixing one list attains at least this average. Maximizing over feasible weights proves rounding.

For \eqref{v5:lp-rate-bracket}, use $c=\log(1/\delta)$ and a fractional solution with success $1-\eps+\delta$. Its rounded success is at least
$(1-\delta)(1-\eps+\delta)\ge1-\eps$, since $\delta\le\eps$. Compactness of the finite program gives the infimum when finite. The lower bound follows from the relaxation.
\end{proof}

For a memoryless block problem, choosing $\delta_n=n^{-2}$ makes the extra logarithmic message budget $O(\log\log n)$. The associated error argument is $\eps-n^{-2}$, so transferring the bound to a limiting coefficient requires control of the profile at this shifted tolerance. In particular, if the fractional programs have a locally uniform, strictly crossing limit
\[
 \mathsf F_n(e^{nR+\sqrt nL})\longrightarrow F(L),
\]
then Theorem~\ref{v5:policy-lp} transfers that same second-order profile to operational coding. Evaluating this limiting fractional-image profile is the remaining step toward a general both-lossy Heegard--Berger second-order formula.

\subsection{Weighted acceptance and successful-subset certificates}
The policy formulation also explains the connection with distortion-tilted information. If an integrable score $j$ satisfies
\begin{equation}
 \sum_xp(x)e^{j(x)}q_s(x)\le1\qquad\text{for every admissible policy }s,
 \label{v5:weighted-acceptance}
\end{equation}
then choosing $\alpha_x=p(x)\one\{j(x)\le\log M+\gamma\}$ and $\beta=e^{-\gamma}/M$ in \eqref{v5:policy-dual} gives
\begin{equation}
 P_{\mathrm c}^*(M)\le\Pr\{j(X)\le\log M+\gamma\}+e^{-\gamma}.
 \label{v5:weighted-spectrum}
\end{equation}
For ordinary lossy coding, the policy is a reproduction and $q_s$ is a distortion-ball indicator. The usual tilted-information exponential certificate implies \eqref{v5:weighted-acceptance}; this is the familiar distortion-sensitive spectrum \cite{rd}. With decoder side information, a reconstruction policy may couple decisions across coordinates. The certificate must therefore hold for every admissible \emph{block policy} under the actual block-distortion constraints.

The following form is particularly useful when an entropy inequality controls every successful subset directly.
\begin{lemma}[Successful-subset spectrum]\label{v5:entropy-spectrum}
Let $\mathcal E$ be the success event under a probability law $P$. Suppose a real budget $b$ and an integrable score $J$ satisfy
\begin{equation}
 b\ge\E[J\mid A]+\log P(A)
 \label{v5:successful-certificate}
\end{equation}
for every measurable $A\subseteq\mathcal E$ of positive probability. Then
\[
 P(\mathcal E)\le P\{J\le b+\gamma\}+e^{-\gamma}\qquad(\gamma>0).
\]
Several certificates may be imposed simultaneously on the same event, with the sum of their exponential penalties.
\end{lemma}
\begin{proof}
Apply \eqref{v5:successful-certificate} to
$A=\mathcal E\cap\{J>b+\gamma\}$. If it has positive probability, then
$\log P(A)<-\gamma$. Removing these successful violations proves the claim, and a union bound proves the simultaneous form.
\end{proof}



\section{Cell-Dependent References and Reconstruction-Constrained Images}\label{sec:cell-reference}
A fixed output reference describes an ordinary channel image. A reconstruction constraint provides additional structure: each source cell is assigned a reproduction, and the reference output measure can be centered at that reproduction. This section retains that information before taking a Gaussian limit. The same likelihood then accounts for both the number of cells and the actual distortions.

\subsection{Normalized mass of a labelled image}
Let $X\sim P_X$ and $Y\mid X\sim W$, on standard Borel spaces with the densities used below. A deterministic partition $M=f(X)$ has at most $M_0$ nonempty cells. For each label $m$, choose a nonnegative reference density $r_m(y)$ of total mass at most one, positive on the sampled output support. The densities may depend on the partition and its prescribed reconstructions, but are fixed functions before the random experiment. Write $p_{M,Y}$ for the joint mass--density under the original law and define
\begin{equation}
 H_r(M,Y)=\log\frac{r_M(Y)}{p_{M,Y}(M,Y)}.\label{eq:cell-H}
\end{equation}
All density ratios are evaluated on their sampling support.

\begin{lemma}[Cell-dependent image mass]\label{lem:cell-reference}
For every such partition and every $\gamma>0$,
\begin{equation}
 \E e^{H_r}\le M_0,\qquad
 \Pr\{H_r>\log M_0+\gamma\}\le e^{-\gamma}.\label{eq:cell-budget}
\end{equation}
Let $s_{m,y}(x)$ be any normalized conditional density that is positive on the relevant source support, and put
\[
 Z_s=\log\frac{p_{X\mid M,Y}(X\mid M,Y)}{s_{M,Y}(X)}.
\]
Then $\Pr\{Z_s<-\gamma\}\le e^{-\gamma}$ and
\begin{equation}
 H_r=Z_s+\log\frac{s_{M,Y}(X)}{p_X(X)}
                  +\log\frac{r_M(Y)}{W(Y\mid X)}.\label{eq:cell-identity}
\end{equation}
A finite family of reference pairs gives simultaneous tests on one success event, with the sum of their exceptional probabilities. A reference mass smaller than one is kept unnormalized, since its normalization would change the likelihood offset.
\end{lemma}
\begin{proof}
Direct integration gives
\[
 \E e^{H_r}=\sum_m\int_{\{p_{M,Y}(m,y)>0\}}r_m(y)\,dy\le M_0.
\]
For each $(m,y)$, integration of $s_{m,y}/p_{X\mid M,Y}$ under the posterior is at most one. Markov's inequality gives both tails. Since $M=f(X)$, Bayes' rule gives
$p_{X\mid M,Y}(x\mid m,y)=p_X(x)W(y\mid x)/p_{M,Y}(m,y)$ on the cell, proving the identity. Apply the union bound only after intersecting each exceptional event with the original success event.
\end{proof}

The choice $r_m=p_Y$ gives $H_r=-\log P(M\mid Y)$. A reference depending on a reconstruction $u(m)$ instead retains a distortion ball inside the likelihood. The proof uses no independence of the source coordinates after restriction to a cell.

\subsection{Images with an agreed reproduction}
Suppose the encoder also selects $U=h(X)$ and the decoder returns $\widehat U=g(M,Y)$. For each $u$, choose a backward source reference $q_{X\mid u}$ and its output $q_{Y\mid u}=q_{X\mid u}W$. Define
\begin{equation}
 g_q(x,u,y)=\log\frac{q_{X\mid u}(x)}{p_X(x)}
                 -\log\frac{q_{Y\mid u}(y)}{p_Y(y)}.\label{eq:cr-reference-score}
\end{equation}
\begin{lemma}[Agreed-reproduction image]\label{lem:agreed-image}
For every $E\subseteq\{U=\widehat U\}$,
\begin{equation}
 \E[e^{g_q(X,U,Y)}\one_E]\le M_0,
 \qquad
 \Pr(E)\le\Pr\{g_q\le\log M_0+\gamma\}+e^{-\gamma}.\label{eq:cr-reference-test}
\end{equation}
More generally, let a coarse reproduction $u_0(m)$ be known from the message. If $v(X)$ is recovered on $E$ and $q_{X\mid v}$ induces $q_{Y\mid v}$, the same conclusion holds for
\begin{equation}
 g(x,u_0,v,y)=\log\frac{q_{X\mid v}(x)}{p_X(x)}
                         +\log\frac{r_{u_0}(y)}{q_{Y\mid v}(y)},\label{eq:nested-agreement-score}
\end{equation}
where $r_{u_0}$ is a probability density for every coarse reconstruction.
\end{lemma}
\begin{proof}
For fixed $m,y$, agreement fixes $u=g(m,y)$. Multiplication by the source--channel density in \eqref{eq:cr-reference-score} gives $p_Y(y)q_{X\mid u,Y=y}(x)$. Its integral over the successful source section is at most $p_Y(y)$. Sum over $m$ and integrate over $y$. For \eqref{eq:nested-agreement-score}, the same integral is bounded by $r_{u_0(m)}(y)$, whose integral is one. The tail bound follows by removing successful points with $g>\log M_0+\gamma$.
\end{proof}

For real block vectors, agreement permits an encoder-observable geometric adjustment at logarithmic cost. Project $x^n$ onto the span of the encoder-selected reproduction and quantize the projection coefficient. On agreement, the decoder has the same vector and applies the transmitted coefficient. For nested reconstructions known at the stronger decoder, projection onto nested spans fixes their cross moment as well. Appendix~\ref{app:cr-details} proves the required bounded-header version, including nearly collinear reproductions.

\subsection{Cancellation of common likelihood costs}
The next lemma separates the image reduction from its Gaussian evaluation. The sigma-field $\mathcal F_n$ contains only information fixed before the residual channel experiment.
\begin{lemma}[Matched-cost image reduction]\label{lem:matched-cost}
Let $E_n$ be the original successful event. Suppose there are $\mathcal F_n$-measurable $A_n,c_n$, a likelihood $H_n$, and residuals $J_n,B_n$ such that, on $E_n$ outside events of probability tending to zero,
\begin{align}
 H_n&\le nr_n+c_n+\sqrt nL+o_P(\sqrt n),\label{eq:matched-upper}\\
 H_n&\ge nr_n+c_n+\sqrt n(A_n+J_n+B_n)-o_P(\sqrt n).\label{eq:matched-lower}
\end{align}
Suppose also that
\begin{equation}
 H_n-\E[H_n\mid\mathcal F_n]=\sqrt nJ_n+o_P(\sqrt n)\label{eq:matched-center}
\end{equation}
on a source-measurable localization retaining $E_n$ up to probability $o(1)$. On that localization assume the conditional CDF of $(J_n,B_n)$ converges uniformly over all thresholds to $\Phi_\Sigma$, where $\Sigma\succ0$. If $A_n\Rightarrow A$ under the original source law, then
\begin{equation}
 \limsup_n\Pr(E_n)\le
 \E\left[\sup_{t\in\mathbb R}\Phi_\Sigma(L-A-t,t)\right].\label{eq:matched-envelope}
\end{equation}
No bound $c_n=o(\sqrt n)$ is required.
\end{lemma}
\begin{proof}
Define the initially measurable allocation
\[
 T_n=\frac{\E[H_n\mid\mathcal F_n]-nr_n-c_n}{\sqrt n}-A_n.
\]
Substitution into \eqref{eq:matched-upper}--\eqref{eq:matched-lower} gives, on success,
$J_n\le L-A_n-T_n+o_P(1)$ and $B_n\le T_n+o_P(1)$.
For fixed $\eta>0$, discard the two violations exceeding $\eta$, condition on $\mathcal F_n$, and apply the uniform conditional CDF bound. The resulting upper bound is
$\E h(L-A_n+2\eta)+o(1)$, where
$h(z)=\sup_t\Phi_\Sigma(z-t,t)$.
A positive definite Gaussian law makes $h$ bounded and continuous; its maximizers are finite and locally bounded, since either marginal tail tends to zero as $|t|\to\infty$ on a bounded $z$-window. Weak convergence and then $\eta\downarrow0$ prove the claim. The same $c_n$ cancels in both substitutions, regardless of its magnitude.
\end{proof}

\edit{Let $(J,B)\sim\mathcal N(0,\Sigma)$ be independent of the limiting observable score $A$. A source-observable allocation fixes the two thresholds before $(J,B)$ is realized. With $\Sigma\succ0$ and finite $L$, this restriction gives}
\[
 \E\sup_t\Phi_\Sigma(L-A-t,t)
 <\Pr\{A+J+B\le L\}.
\]
For scalar Gaussian Wyner--Ziv coding, Xu and Zhou~\cite{xuzhou} prove that the left side is the operational optimum. Their posterior smoothing and sampled-description estimates establish the premises of Lemma~\ref{lem:matched-cost}. We next evaluate a cell-dependent reference that also retains an uninformed reconstruction. In its proof, the $o(n)$ second-moment estimate for the centered likelihood and the uniform conditional CDF estimate verify the two distinct requirements of Proposition~\ref{prop:spectral-stability}.

\subsection{A two-level Gaussian distortion image}\label{sec:gaussian-image-local}
Let $X^n\sim\cN(0,pI_n)$ and $Y^n=X^n+Z^n$, with $Z^n\sim\cN(0,qI_n)$ independent. A labelled distortion image consists of at most $M$ source cells, a center $u_m^n$ for each cell, and a measurable reconstruction policy $v_m^n(y^n)$. Its successful mass is
\begin{equation}
 \Pr\{\|X^n-u_M^n\|^2\le na,
          \ \|X^n-v_M^n(Y^n)\|^2\le nb\}.\label{eq:two-level-success}
\end{equation}
Let $\mathcal I_n(M;a,b)$ be the supremum of this mass over all such images. Fix
\begin{equation}
 0<a<p,\qquad 0<b<\frac{aq}{a+q},\qquad
 R_{a,b}=\frac12\log\frac{pq}{b(a+q)}.\label{eq:hb-regime}
\end{equation}
Put $\rho_a=a/(a+q)$, $\beta=b/q$, and
\begin{equation}
 \Sigma_{a,b}=\begin{pmatrix}
 (\rho_a-\beta)(1-(\rho_a-\beta)/2)&-\beta(\rho_a-\beta)/2\\
 -\beta(\rho_a-\beta)/2&\beta(1-\beta/2)
 \end{pmatrix}.\label{eq:hb-covariance}
\end{equation}
For $G_0\sim\cN(0,1/2)$, define
\begin{equation}
 h_{a,b}(z)=\sup_t\Phi_{\Sigma_{a,b}}(z-t,t),\qquad
 F_{a,b}(L)=\E h_{a,b}(L-G_0).\label{eq:hb-profile}
\end{equation}

\begin{theorem}[Gaussian reconstruction-constrained image]\label{thm:gaussian-distortion-image}
Under \eqref{eq:hb-regime}, for every fixed $L\in\mathbb R$ and every budget sequence
$\log M_n=nR_{a,b}+\sqrt nL+o(\sqrt n)$,
\begin{equation}
 \lim_{n\to\infty}\mathcal I_n(M_n;a,b)=F_{a,b}(L).\label{eq:hb-image-limit}
\end{equation}
The profile is continuous and strictly increasing from zero to one. The construction needs only $O(\log n)$ additional description nats at a fixed source cutoff; a diagonal removal of the cutoff preserves the second-order limit.
\end{theorem}
\begin{proof}[Proof architecture]
For the converse choose the cell reference
$r_m(y^n)=\phi_{(a+q)I_n}(y^n-u_m^n)$ and the posterior reference
$s_{m,y^n}(x^n)=\phi_{bI_n}(x^n-v_m^n(y^n))$.
Writing $w^n=X^n-u_M^n$, the exact identity is
\begin{align}
 H_r={}&Z_s+nR_{a,b}+\sqrt nA_n+N_{0,n}
       +\frac{na-\|w^n\|^2}{2(a+q)}
       +\frac{nb-\|X^n-v_M^n(Y^n)\|^2}{2b},\label{eq:hb-exact-identity}\\
 A_n={}&\frac{\|X^n\|^2-np}{2p\sqrt n},\qquad
 N_{0,n}=-\frac{\langle w^n,Z^n\rangle}{a+q}
       +\frac{a(\|Z^n\|^2-nq)}{2q(a+q)}.\nonumber
\end{align}
Both distortion deficits have positive coefficients on success. Lemma~\ref{lem:cell-reference} first localizes those deficits. Appendix~\ref{app:hb-image} then proves the two likelihood comparisons, their common sampled-description cost, and the conditional covariance \eqref{eq:hb-covariance}; Lemma~\ref{lem:matched-cost} gives the upper profile.

For the construction, a coarse source-shell cover fixes
$\|x^n-u^n\|^2=na+O(1)$ at budget
$n\log(p/a)/2+\sqrt nA_n(x^n)+O(\log n)$.
A conditional refinement shell allocates the remaining information and actual-distortion thresholds from $L-A_n(x^n)$. Its conditional Gaussian law has covariance \eqref{eq:hb-covariance}. The shell construction and finite-alphabet branch accounting are proved in Appendix~\ref{app:hb-image}. They attain $F_{a,b}(L)$ with logarithmic slack. For any specified $o(\sqrt n)$ perturbation of the budget, apply the construction at $L-\eta$, then let $\eta\downarrow0$.

The determinant is
$\beta(\rho_a-\beta)(1-\rho_a/2)>0$.
The compact-maximizer argument in Lemma~\ref{lem:matched-cost} proves continuity of $h_{a,b}$. Increasing its first threshold strictly increases the Gaussian probability, so $h_{a,b}$ and $F_{a,b}$ are strictly increasing. Their endpoint limits follow by the halfspace upper bound and the feasible choice $t=z/2$.
\end{proof}

The source variance in \eqref{eq:hb-profile} is $1/2$. It arises from the source density in \eqref{eq:hb-exact-identity}. Using a common unconditional output reference would discard the coarse reconstruction and give a different source score. This is why the reference measure is part of the image geometry rather than an arbitrary cosmetic choice.


\section{Private Images and Posterior-Mass Constraints}\label{sec:security}
For extraction, the relevant image is the set of key values generated by high-posterior-mass source points. Such a set can have substantial actual probability but too few points to have comparable probability under the ideal uniform key law. This gives the secrecy counterpart of image counting.

\subsection{Two-sided private-image spectra}
Let $(V,E)$ have any finite joint law. An extractor maps $V$ to $K$ values, deterministically or using an independent public seed $S$ included in the observer's information. Write
\[
 H_{V|E}(V,E)=-\log P_{V|E}(V|E),\quad
 \Delta=\TV(P_{KES},\operatorname{Unif}[K]\times P_{ES}).
\]
Conditioned on the public seed, the key is a deterministic function of the source; the randomness being extracted is that of $V$.

\begin{lemma}[Private-image spectrum]\label{lem:private-spectrum}
For every $\gamma>0$, every extractor satisfies
\begin{equation}
 \Delta\ge\Pr\{H_{V|E}<\log K-\gamma\}-e^{-\gamma}.
 \label{eq:pa-converse}
\end{equation}
A public two-universal hash family satisfies
\begin{equation}
 \Delta\le\Pr\{H_{V|E}<\log K+\gamma\}+\tfrac12e^{-\gamma/2}.
 \label{eq:pa-direct}
\end{equation}
In particular, a fixed deterministic hash attains the upper bound for the given source law.
\end{lemma}
\begin{proof}
For fixed $e$, let $A_e=\{v:P(v|e)>e^\gamma/K\}$. Its size is less than $Ke^{-\gamma}$. Its image under any extractor therefore has at most that many key values. The actual probability of the event that the key belongs to this image is at least $P\{V\in A_E\}$, while its ideal probability is at most $e^{-\gamma}$. The variational characterization of total variation proves \eqref{eq:pa-converse}. Apply the same argument at each independent seed value to include public seeding.

For achievability truncate the original source to the subprobability measure
\[
 r(v,e)=P(v,e)\one\{P(v|e)\le e^{-\gamma}/K\},
 \qquad \beta=1-\sum_{v,e}r(v,e).
\]
For a two-universal random hash $H$, the collision expansion gives
\[
 \E_H\sum_{k,e}\frac{(r_H(k,e)-r_E(e)/K)^2}{P_E(e)}
 \le\sum_{v,e}\frac{r(v,e)^2}{P_E(e)}\le\frac{e^{-\gamma}}K.
\]
Cauchy--Schwarz yields expected variation at most $\frac12e^{-\gamma/2}$ between the hashed subprobability and its own uniform-key product. Restoring the missing mass contributes at most $\beta$ in total variation, since each restored measure has mass $\beta$ and total variation includes the factor $1/2$. This proves \eqref{eq:pa-direct}. Including the public seed averages over its values; alternatively fix a hash whose variation is no larger than that average.
\end{proof}

Truncation is performed on the full joint source law. Consequently, when extraction is composed with a decoder, the bound includes both agreement and disagreement outcomes.

\subsection{Public messages and private-image composition}
A public message changes the posterior mass available to extraction. Its effect can be retained without assuming that its bins have equal probability.

\begin{lemma}[Extraction after a fixed public message]\label{lem:public-composition}
Let $F=f(V)$ take at most $M$ values. An independent two-universal hash to $K$ values, or a suitable fixed member of the family, satisfies
\begin{equation}
 \TV(P_{KFE},\operatorname{Unif}[K]\times P_{FE})
 \le\Pr\{H_{V|E}<\log M+\log K+\gamma\}+\tfrac12 e^{-\gamma/2}.
 \label{eq:public-pa}
\end{equation}
The public-seed version includes the seed on both sides. The bound holds for every deterministic $f$, with no uniformity assumption on $F$.
\end{lemma}
\begin{proof}
Truncate at $P(v|e)\le e^{-\gamma}/(MK)$ and denote the retained measure by $r$. In each fixed public bin the collision calculation is unchanged. Using weight $P_E(e)/M$ gives
\[
 \E_H\sum_{k,f,e}
 \frac{(r_H(k,f,e)-r_{FE}(f,e)/K)^2}{P_E(e)/M}
 \le M\sum_{v,e}\frac{r(v,e)^2}{P_E(e)}\le\frac{e^{-\gamma}}K.
\]
The sum of these weights over $(k,f,e)$ is $K$. Cauchy--Schwarz and restoration of the omitted mass prove \eqref{eq:public-pa}. Averaging or fixing the hash proceeds as in Lemma~\ref{lem:private-spectrum}.
\end{proof}
The lemma quantifies the effect of the public partition on the extraction budget. Applying it to a reconciliation message gives the public-rate/key-rate construction in Corollary~\ref{cor:reconciliation}.

\subsection{Public refinements measured by their residual information}\label{sec:public-refinement}
An alphabet-size charge is useful without further information about a public message. When the observer already predicts most of that message, the relevant private-image loss is its conditional surprisal. Let $J=f(V)$ be countably valued and put
\[
 H_{V|E}=-\log P(V\mid E),\qquad T_{J|E}=-\log P(J\mid E).
\]
Since $J$ is a deterministic source function, Bayes' rule gives
\begin{equation}
 -\log P(V\mid E,J)=H_{V|E}-T_{J|E}.
 \label{eq:public-spectrum-identity}
\end{equation}
The identity uses the actual joint law, including the public message on every outcome. Write $\Delta(K)$ for total variation from a uniform key independent of $(E,J)$ and, if used, the independent public hash seed; the seed is included on both sides of the comparison.
\begin{proposition}[Private-image bounds after a public refinement]\label{prop:public-refinement}
For finite $V$, every independent-public-seed extractor to $K$ values obeys
\begin{equation}
 \Delta(K)\ge\Pr\{H_{V|E}<\log K-\gamma\}-e^{-\gamma}.
 \label{eq:public-spectrum-converse}
\end{equation}
A two-universal hash, or a suitable fixed member, attains
\begin{align}
 \Delta(K)
 &\le\Pr\{H_{V|E}-T_{J|E}<\log K+\gamma\}
                    +\tfrac12e^{-\gamma/2}\label{eq:public-spectrum-direct}\\
 &\le\Pr\{H_{V|E}<\log K+\gamma+t\}
       +\Pr\{T_{J|E}>t\}+\tfrac12e^{-\gamma/2}.
 \label{eq:public-spectrum-split}
\end{align}
The second probability in \eqref{eq:public-spectrum-split} is at most $H(J\mid E)/t$.
\end{proposition}
\begin{proof}
Apply Lemma~\ref{lem:private-spectrum} with observer $(E,J)$ and substitute \eqref{eq:public-spectrum-identity}. For the lower bound, $T_{J|E}\ge0$, so the event $H_{V|E}<\log K-\gamma$ is contained in its conditional-posterior counterpart. For the upper bound, split according to $T_{J|E}\le t$. Markov's inequality yields the final estimate. The seed is independent and public throughout.
\end{proof}
In particular, a public alphabet can be exponential while $T_{J|E}=o_P(\sqrt n)$. Conversely, a label cheaply recoverable by a legitimate decoder need not be predictable by the observer. These are different conditional spectra.


\section{Tight Second-Order Rate Applications}\label{sec:tight}
We first apply the reconstruction-constrained image theorem to Gaussian coding. We then use partition stability for a structured finite-alphabet source, the joint posterior bound for compatible helper constraints, and posterior-mass bounds for extraction. In every case, the construction respects the same observations as the converse.


\subsection{Gaussian Heegard--Berger coding with one uninformed decoder}\label{sec:hb-gaussian-rate}
An encoder observing $X^n\sim\cN(0,pI_n)$ sends one common message. The first decoder reconstructs $U^n(M)$; the second observes $Y^n=X^n+Z^n$, $Z^n\sim\cN(0,qI_n)$, and reconstructs $V^n(M,Y^n)$. The joint excess-distortion probability is the probability that either constraint in \eqref{eq:two-level-success} fails. This is the Heegard--Berger architecture in which side information may be absent~\cite{heegard,nit}; the encoder does not observe $Y^n$.

\begin{corollary}[Both-active Gaussian Heegard--Berger rate]\label{cor:hb-gaussian-rate}
For the fixed interior regime \eqref{eq:hb-regime} and $0<\eps<1$,
\begin{equation}
 \log M^*_{\rm HB}(n,a,b,\eps)
 =nR_{a,b}+\sqrt nF_{a,b}^{-1}(1-\eps)+o(\sqrt n).
 \label{eq:hb-rate-exact}
\end{equation}
\end{corollary}
\begin{proof}
Every code gives a labelled image as in \eqref{eq:two-level-success}, and every such image specifies an encoder and both decoders. Theorem~\ref{thm:gaussian-distortion-image} therefore bounds the same operational success probability in both directions. Its continuous strictly increasing profile has a unique inverse. Offsets strictly above that inverse achieve a strict error margin; offsets below it violate the converse. A diagonal favorable adjustment proves the fixed-$\eps$ expansion.
\end{proof}

The refinement boundary is $b<s_a$, not merely $b<a$. The side information alone can refine a Gaussian coarse residual of variance $a$ to $s_a=aq/(a+q)$. At $b=s_a$ the refinement layer disappears and the positive-residual-variance argument changes. All results in this subsection use fixed strict inequalities.

For $p=4$, $q=1$, $a=1$, and $b=1/4$, one has $R_{a,b}=\frac12\log8$ and
\[
 \Sigma_{a,b}=\begin{pmatrix}7/32&-1/32\\-1/32&7/32\end{pmatrix}.
\]
Equal diagonal entries make $J+B$ and $J-B$ independent. Conditionally on their sum, the rectangle of largest probability is centered at $J-B=0$, so $t=z/2$. With independent standard Gaussians $G,H$,
\begin{equation}
 F_{a,b}(L)=\Pr\left\{\sqrt{7/8}\,G+\frac{|H|}{\sqrt2}\le L\right\}.
 \label{eq:hb-folded-example}
\end{equation}
At $\eps=0.1$ its inverse is approximately $1.88793694$ nats. The source and noise sum has variance $7/8$, but its scalar halfspace quantile is not the inverse of \eqref{eq:hb-folded-example}. The absolute-value term records the restriction that the allocation is chosen before the decoder noise.

\subsection{Gaussian common reconstruction}\label{sec:cr-gaussian}
A common-reconstruction requirement asks the decoder to reproduce the encoder-selected reconstruction, rather than to choose an arbitrary reproduction using its side information~\cite{steinbergcr,ahmadi}. We use one fixed joint error allowance: disagreement or violation of any required distortion constitutes error. \edit{Thus agreement and distortion share the prescribed block-error allowance; a separately vanishing disagreement requirement defines a different criterion.}

For one informed decoder, let $U^n=h(X^n)$ be selected by the encoder and let $\widehat U^n=g(M,Y^n)$. Success requires $U^n=\widehat U^n$ and $\|X^n-U^n\|^2\le nD$.
For the two-decoder variant, the weak decoder reconstructs $U_0^n(M)$ without side information. The encoder selects a fine reconstruction $U_1^n=h(X^n)$ and the informed decoder must agree with it. Success requires both squared-error bounds $na,nb$, as well as fine agreement.

\begin{theorem}[Gaussian agreed-reconstruction rates]\label{thm:cr-rates}
Let $p,q>0$ and $0<\eps<1$.
For the single-reconstruction problem with $0<D<p$, put
\begin{align}
 R_{\rm CR}(D)&=\frac12\log\frac{p(D+q)}{D(p+q)},\nonumber\\
 c&=\frac q{p+q},\quad d_q=\frac q{D+q},\quad
 V_{\rm CR}=\frac{c^2}{2}+(d_q-c)-\frac{(d_q-c)^2}{2}.
 \label{eq:cr-rate-parameters}
\end{align}
Then
\begin{equation}
 \log M^*_{\rm CR}(n,D,\eps)
 =nR_{\rm CR}(D)+\sqrt{nV_{\rm CR}}\,Q^{-1}(\eps)+o(\sqrt n).
 \label{eq:cr-rate}
\end{equation}
For the two-decoder problem with $0<b<a<p$, put
\begin{equation}
 R_{\rm HB,CR}=\frac12\log\frac{p(b+q)}{b(a+q)},\quad
 \delta=\frac q{b+q}-\frac q{a+q},\quad
 V_{\rm HB,CR}=\frac12+\delta-\frac{\delta^2}{2}.
 \label{eq:hbcr-parameters}
\end{equation}
Its minimum message size satisfies
\begin{equation}
 \log M^*_{\rm HB,CR}(n,a,b,\eps)
 =nR_{\rm HB,CR}+\sqrt{nV_{\rm HB,CR}}\,Q^{-1}(\eps)+o(\sqrt n).
 \label{eq:hbcr-rate}
\end{equation}
\end{theorem}
\begin{proof}[Image argument]
For \eqref{eq:cr-rate}, apply Lemma~\ref{lem:agreed-image} with
$q_{X\mid u}=\cN(u,DI_n)$. An encoder-observable projection adjustment costs $O(\log n)$ and enforces
$\langle X^n,X^n-U^n\rangle=\|X^n-U^n\|^2+o(n)$ on the retained source set. The positive distortion multiplier localizes the residual norm at $nD$. Conditional Gaussian evaluation then gives variance \eqref{eq:cr-rate-parameters}. A backward Gaussian shell fixes actual distortion and attains the same score.

For \eqref{eq:hbcr-rate}, choose $r_{u_0}=\cN(u_0,(a+q)I_n)$ and $q_{X\mid u_1}=\cN(u_1,bI_n)$ in \eqref{eq:nested-agreement-score}. Projections onto the coarse reproduction span and the nested coarse--fine span are implementable because the strong decoder knows the coarse reproduction from the common message. Their quantized amplitudes cost $O(\log n)$. The image bound localizes the two residual energies and their inner product at $(a,b,b)$ per coordinate. Its source fluctuation has variance $1/2$ and its conditional channel fluctuation has variance $\delta-\delta^2/2$. Nested shells attain that scalar spectrum. Appendix~\ref{app:cr-details} provides the projection, likelihood, and construction details for both claims.
\end{proof}

At $(p,q,a,b)=(4,1,1,1/4)$, the common-reconstruction center and variance are $\frac12\log10$ and $0.755$. Its coefficient at $\eps=0.1$ is approximately $1.11354959$ nats. Comparison with \eqref{eq:hb-folded-example} uses different first-order centers. The smaller common-reconstruction coefficient does not imply a smaller total rate.

\subsection{Successive refinement with side information absent at the first decoder}\label{sec:sr-image}
Let the weak decoder receive $M_1$ and the informed decoder receive $(M_1,M_2,Y^n)$. At the both-active corner in \eqref{eq:hb-regime}, set
\[
 R_1=\frac12\log\frac pa,\qquad
 R_2=\frac12\log\frac{s_a}{b}.
\]
The criterion is the probability that either actual distortion constraint fails. This is a specialization of successive refinement with degraded decoder side information~\cite{steinbergmerhav}.

\begin{corollary}[Joint two-stage image region]\label{cor:sr-image}
At this corner the exact closed second-order region is
\begin{equation}
 \E\left[\one\{G_0\le L_1\}\,
 h_{a,b}(L_1+L_2-G_0)\right]\ge1-\eps,
 \qquad G_0\sim\cN(0,1/2).
 \label{eq:sr-image-region}
\end{equation}
\end{corollary}
\begin{proof}
\edit{For the converse, put
$\jmath_a(x^n)=nR_1+(\|x^n\|^2-np)/(2p)$.
Multiplying the Gaussian source density by $e^{\jmath_a}$ gives the constant $(2\pi ea)^{-n/2}$. The weighted mass of a radius-$\sqrt{na}$ distortion ball is therefore
\[
 \omega_n=\frac{(n/(2e))^{n/2}}{\Gamma(n/2+1)}\le1.
\]
The first message selects at most $M_1$ such balls, so the mass of joint success with $\jmath_a(X^n)>\log M_1+\gamma$ is at most $\omega_ne^{-\gamma}$. Taking $\gamma\to\infty$ with $\gamma=o(\sqrt n)$ gives $A_n\le L_1+o_P(1)$ on joint success.} Apply the arbitrary-image converse of Theorem~\ref{thm:gaussian-distortion-image} to the total message $(M_1,M_2)$, retaining that source-measurable restriction before its final conditional Gaussian expectation. This proves the upper profile in \eqref{eq:sr-image-region}.

\edit{For the direct part, retain a fixed source-energy window and branches with upper endpoints $\alpha_j\le L_1-\eta$, where $\eta>0$. Let $B_n$ be the polynomial number of branches and $N_j$ the coarse alphabet in branch $j$. Define
\[
 H_j=\left\lfloor\frac{M_1}{B_nN_j}\right\rfloor.
\]
The first message enumerates the triples consisting of a branch, a coarse index, and one of $H_j$ refinement subindices. Their total number is at most $M_1$. The second message carries the remaining refinement subindex in an alphabet of size $M_2$. The strict branch margin ensures $H_j\ge1$ for large $n$, and
\[
 \log(H_jM_2)\ge nR_2+
 \sqrt n(L_1+L_2-\alpha_j)-O(\log n).
\]
Apply the construction in Appendix~\ref{app:hb-image} with a fixed positive second-order margin to absorb this logarithmic overhead. It attains the profile in \eqref{eq:sr-image-region} with that margin, averaged over the retained source branches. The weak decoder reads only the branch and coarse index. Let the rate margin decrease, then $\eta\downarrow0$, and expand the source window; a diagonal sequence gives the asserted closed region. All branch and refinement codebooks are fixed before transmission.}
\end{proof}

For the parameters of \eqref{eq:hb-folded-example}, $\eps=0.1$, and $L_1=1.2$, the boundary has $L_1+L_2\approx2.03890685$. When $L_1=Q^{-1}(0.1)/\sqrt2$, the coarse decoder alone uses the whole error allowance, and no finite $L_1+L_2$ attains joint success $0.9$. The common source score must remain inside the joint probability rather than being replaced by two separately optimized quantiles.



\subsection{Structured source cells with a recoverable label}\label{sec:structured-cells}
The partition theorem becomes useful when a whole source cell fails both product-reference tests of Section~\ref{sec:conditional}, but its components have regular conditional compositions. The following construction has exponentially many components, so a bound charging the full label alphabet would change the first-order rate.

Let $h$ denote binary entropy in nats. Fix rational parameters
\[
 0<t,q<1/2,\quad 2(t+q)<\delta<1/2,
 \qquad 0<R<\log2-h(\delta).
\]
For compatible even $n$, choose $L_n=\lfloor e^{nR}\rfloor$ balanced binary centers $a^n(j)$ at pairwise Hamming distance greater than $n\delta$. Such a sequence exists by constant-composition greedy packing. Around each center, form the conditional-type cell $\mathcal A_j$ by flipping exactly a fraction $t$ of its zeros and a fraction $t$ of its ones. Choose $J$ uniformly, $X^n$ uniformly on $\mathcal A_J$, and observe $C^n=X^n\oplus N^n$ through a BSC of crossover $q$. These disjoint cells make $J$ source-measurable. An encoder observes $X^n$ and a decoder must recover it from one fixed-length message and $C^n$.

Set $v=t+q-2tq$, $Q_A=\operatorname{Bern}(1/2)$, $Q_{X|A}=\operatorname{BSC}(t)$, and $W_{C|X}=\operatorname{BSC}(q)$. Define
\[
 H_*=h(t)+h(q)-h(v),\qquad
 V_*=\E_{Q_{AX}}\Var_{W(\cdot|X)}[-\log Q_{X|C,A}(X\mid C,A)].
\]
\begin{proposition}[A matching irregular-cell application]\label{prop:structured-cells}
For this source sequence, $V_*>0$ and
\begin{equation}
 \log M_n^*(\eps)=nH_*+\sqrt{nV_*}\,Q^{-1}(\eps)+o(\sqrt n).
 \label{eq:structured-rate}
\end{equation}
Although $H(J)=nR+O(1)$, its conditional entropy obeys
$H(J\mid C^n)=O(ne^{-cn})$ for some $c>0$. Relative to the fair-input product reference, both the full source-cell density deficit and the conditional image divergence grow linearly in $n$.
\end{proposition}
\begin{proof}[Mechanism]
The observation identifies the center with exponentially small error because the center distance exceeds twice the typical source-plus-channel error weight. Conditioned on that center, the source is a regular conditional type. Its conditional image spectrum has mean $nH_*$, variance $nV_*$, and an $O_P(\log n)$ reference correction. The component profile is therefore Gaussian. Corollary~\ref{cor:partition-profile} removes the label without second-order loss. Conditional binning, with the same alphabet reused for every center, realizes the image bound. Appendix~\ref{app:structured-example} proves the packing, likelihood, and rate statements and computes the two linear global deficits.
\end{proof}
This is a characterization of the specified structured source sequence, rather than an assertion for every iid helper network. Its purpose is to exhibit an image class strictly beyond the whole-cell small-deficit hypotheses. In a helper network, the same reduction applies only when the responsible encoder can form the partition and all active section budgets satisfy the recovery condition.


\subsection{Several sources with one rate-limited helper}\label{sec:helper-app}

The helper observes $A^n$, encoder $j$ observes $B_j^n$, and a
single decoder reconstructs all $B_1^n,\ldots,B_k^n$.  Write
$r(S)=\sum_{j\in S}r_j$ and $L(S)=\sum_{j\in S}L_j$.  The
first-order region is the union, over $U-A-B$, of
\begin{equation}
 \begin{aligned}
 r_h&\ge I(U;A),\\
 r(S)&\ge H(B_S|B_{S^c},U)
             \quad(\varnothing\ne S\subseteq[k]).
 \end{aligned}
 \label{comp:hi:eq:multi-wak-first-order}
\end{equation}
Fix the informative exposed aggregate pair
\eqref{comp:hi:eq:informative}, and fix target rates summing to
$r_\Sigma$.  The following hypotheses specify which active sections
are included.

\begin{enumerate}
\item For every $S\in\mathscr A$, assume $A-B_{S^c}-B_S$ and
      $r(S)=H(B_S|B_{S^c})$.
\item For every $S\in\mathscr W$, assume $U\perp B_{S^c}$ for
      every posterior measure in $\mathcal M_{\boldsymbol r}$, and
      $r(S)=r_\Sigma-H(B_{S^c})$.
\item The two collections are disjoint.  Every other proper subset
      has a uniform positive margin:
      \begin{equation}
       \inf_{\mu\in\mathcal M_{\boldsymbol r}}
          [r(S)-H_\mu(B_S|B_{S^c},U)]>0.
       \label{comp:hi:eq:uniform-slack}
      \end{equation}
\end{enumerate}

The mean and covariance vectors $K$ and $\Sigma_\mu$ are those in
\eqref{comp:hi:eq:source-vector} and \eqref{comp:hi:eq:mark-covariance}.  Define
\begin{equation}
 \begin{aligned}
 b(\boldsymbol L)&=\left(L_h+cL([k]),\
       (L(S))_{S\in\mathscr A},\
       (L_h+cL(S))_{S\in\mathscr W}\right),\\
 G&\sim\cN(0,\Cov K(A)),\\
 F(\boldsymbol L)&=\E_G\sup_{\mu\in\mathcal M_{\boldsymbol r}}
                \Phi_{\Sigma_\mu}(b(\boldsymbol L)-G).
 \end{aligned}
 \label{comp:hi:eq:mixed-profile}
\end{equation}

\begin{theorem}[Mixed ancillary and weighted helper intersections]
\label{comp:hi:thm:mixed}
Assume that all retained diagonal entries of $\Sigma_\mu$ are positive for every $\mu\in\mathcal M_{\boldsymbol r}$. Assume also \eqref{comp:hi:eq:positive},
\eqref{comp:hi:eq:informative}, the three hypotheses above, and the hypotheses of
Lemma~\ref{comp:hi:thm:input}.  The closed local second-order region
for joint reconstruction error at most $\eps$ is
\begin{equation}
 \mathcal L_\eps=\{\boldsymbol L:F(\boldsymbol L)\ge1-\eps\}.
 \label{comp:hi:eq:mixed-region}
\end{equation}
Every code has limiting upper success bounded by $F(\boldsymbol L)$, and this profile is attained. The operational restrictions at zero message rates are included. The joint covariance may be singular.
\end{theorem}
\begin{proof}
\emph{Converse.}
For fixed helper index $m$, let $\mathcal D_m$ be the actual set of
tuples reconstructed correctly by the decoder.  The encoder-index
tuple is injective on this set, so
\begin{equation}
 |\mathcal D_m|\le\prod_jM_j,
 \qquad
 |\{b_S^n:(b_S^n,b_{S^c}^n)\in\mathcal D_m\}|
                  \le\prod_{j\in S}M_j.
 \label{comp:hi:eq:network-sections}
\end{equation}
Condition on a helper type $p_n$ in a bounded $\sqrt n(p_n-p)$ window. With $\gamma_n=3\log(n+1)$ put
\begin{equation}
 N_n=\left\lfloor\frac{|T_n(p_n)|e^{-\gamma_n}}{M_h}\right\rfloor.
 \label{comp:hi:eq:trim-size}
\end{equation}
Cells smaller than $N_n$ contain at most an $e^{-\gamma_n}$ fraction of the type. In every other cell keep the $N_n$ source words having the largest actual values $W^n(\mathcal D_m\mid a^n)$. The set $\mathcal D_m$ is fixed before trimming, so its section bounds and the physical channel are unchanged. The retained average success bounds the original cell average from above. Since $a_0>0$, $N_n$ is exponentially large and
\begin{equation}
 \log N_n=na_0+\sqrt n[-\langle\log p,z_n\rangle-L_h]+o(\sqrt n),
 \quad z_n=\sqrt n(p_n-p).
 \label{comp:hi:eq:trim-expansion}
\end{equation}
Rebuild the source posteriors on each retained cell. This also preserves each ancillary Markov law.
Theorem~\ref{comp:hi:lem:vector-transfer} therefore applies to the retained
set in every relevant cell.

The normalization of each row matters.  Substitute
\eqref{comp:hi:eq:trim-expansion} and the message budgets into
\eqref{comp:hi:eq:vector-thresholds}.  The aggregate row gives
$L_h+cL([k])-\langle j,z_n\rangle+o(1)$.  An ancillary row gives
$L(S)-\langle m_S,z_n\rangle+o(1)$.  A weighted row gives
\begin{equation}
 L_h+cL(S)-\langle j-cm_{C_S},z_n\rangle+o(1),
 \label{comp:hi:eq:weighted-type-budget}
\end{equation}
because its order-$n$ term is
$c[r(S)-r_\Sigma+H(C_S)]=0$.  Thus every weighted row contains
the same helper offset $L_h$.  No section introduces an additional
allocation variable.  The joint bound is consequently the integrand
of \eqref{comp:hi:eq:mixed-profile}, evaluated at the empirical source
fluctuation.

Under the memoryless source,
$n^{-1/2}\sum_i[K(A_i)-\E K(A)]\Rightarrow G$.  Average the
uniform bound on bounded source-type windows and enlarge the
windows.  The full fixed-pair covariance family remains inside the
expectation because the converse permits the source history to
select any member of that family.  This proves the outer bound.

\emph{Achievability.}
On a bounded source-type window, approximate the covariance envelope
by a finite family of posterior measures. This is possible uniformly by compactness and continuity of Gaussian orthant probabilities under the positive-diagonal hypothesis.  Each covariance
matrix has a finite posterior representation: preserve the
$|\mathcal A|-1$ independent barycenter coordinates, the entropy,
and the $d(d+1)/2$ covariance entries, where $d$ is the number of
rows in \eqref{comp:hi:eq:mark-vector}, and apply the finite-dimensional
convex hull theorem.  At most
$|\mathcal A|+d(d+1)/2+1$ contacts suffice.  Thus the entire finite
family can be implemented with finite auxiliary alphabets.

Select its member using only the observed $A$-type and apply
Lemma~\ref{comp:hi:lem:implementation}.  The decoder learns $u^n$, its
joint type with $A^n$, and the branch from the helper message.  It
uses the product reference $Q_n(b^n|u^n)$.  Every target encoder
uses one independent random binning map, fixed across all helper
branches and all source types.  The decoder maximizes the joint
reference probability among tuples having the received indices.

For a transmitted tuple, classify competitors according to the
nonempty set $S$ of coordinates that differ.  The complement is
fixed.  There are at most
$1/Q_n(b_S^n|b_{S^c}^n,u^n)$ competitors with at least the transmitted
reference probability, and the probability of a simultaneous bin
collision for their differing coordinates is $1/M(S)$.  Therefore
the average error outside
\begin{equation}
 -\log Q_n(B_S^n|B_{S^c}^n,U^n)
                   \le\log M(S)-\gamma_n
       \quad(\varnothing\ne S\subseteq[k])
 \label{comp:hi:eq:joint-binning-event}
\end{equation}
is bounded by the probability that this joint event fails plus
$(2^k-1)e^{-\gamma_n}$.  This keeps all active thresholds on the
same event and requires no target encoder to know the helper type.

For $S\in\mathscr A$, the conditional reference in
\eqref{comp:hi:eq:joint-binning-event} is exactly
$P_{B_S|B_{S^c}}^n$, for every rounded helper type.  Its mean is
$n\langle m_S,p_n\rangle$.  For $S\in\mathscr W$, the exact
single-letter identity is
\begin{equation}
 \begin{split}
 -\log Q_n(B_S|C_S,U)
  ={}&-\log Q_n(B|U)+\log P_{C_S}(C_S)\\
     &+\log\frac{Q_n(C_S|U)}{P_{C_S}(C_S)}.
 \end{split}
 \label{comp:hi:eq:weighted-reference-split}
\end{equation}
Lemma~\ref{comp:hi:lem:implementation} makes the sum of the last term
$o_{\PP}(\sqrt n)$, uniformly on the window.  Subtracting the mean
$cn\langle m_{C_S},p_n\rangle$ from
\eqref{comp:hi:eq:aggregate-mean-ach} gives precisely
\eqref{comp:hi:eq:weighted-type-budget}.  The aggregate, ancillary, and
weighted centered channel vectors have covariance
$\Sigma_\mu+o(1)$ under the fixed compatible $(a^n,u^n)$.

For every omitted subset, \eqref{comp:hi:eq:uniform-slack} gives a fixed
positive gap between its rate and its conditional entropy, uniformly
over the chosen family and nearby types.  Its conditional score
increments are bounded and independent once $(a^n,u^n)$ is fixed;
Chebyshev's inequality already makes its violation probability
$O(n^{-1})$.  The retained rows have the conditional joint Gaussian
limit in \eqref{comp:hi:eq:mixed-profile}.  Averaging over source types
therefore attains the finite-family approximation to that profile.

Choose deterministic bin maps with average joint error no greater
than the random ensemble average.  Finally enlarge the source-type
window, refine the finite covariance family, and remove inward rate
shifts along a blocklength diagonal.  All headers and type-rounding
losses are $O(\log n)$ for each fixed stage, so they vanish after
division by $\sqrt n$.  This proves the closed region and, in the
regular case, the asserted profile equality.
\end{proof}

The covariance supremum in the converse ranges over the complete optimizer fiber. Accordingly, the weighted-section condition $U\perp B_{S^c}$ is required throughout that fiber. A smaller family gives the same characterization whenever it attains the same supremum at every relevant threshold.


\subsection{Privacy amplification from a conditional-type image}\label{sec:pa-app}
Fix public $u^n$ and compatible joint types $Q_n\to Q$ on finite alphabets. The source $X^n$ is uniform on $\mathcal T_n(Q_{n,X|U}\mid u^n)$, and the observer sees $Z^n$ through $W_{Z|X}^{\otimes n}$. The retained conditional support is fixed and its positive $Q_{n,X|U}$ entries are bounded away from zero. A fixed hash or an independent public-seed extractor maps $X^n$ to $K$ key values. Its total variation is measured against a uniform key independent of $(Z^n,u^n)$ and the seed. Let $K_n^*(\delta)$ be the largest admissible alphabet at variation at most $\delta$.

Define
\begin{equation}
 V_Z(Q)=\sum_{u,x}Q(u,x)\Var_{W(\cdot|x)}
       \log\frac{W(Z|x)}{(Q_{X|u}W)(Z)}.
 \label{eq:private-variance}
\end{equation}
This is the conditional information variance remaining after the source composition has been fixed, rather than the unconditional information variance of an iid pair.

\begin{theorem}[Conditional-type extraction rate]\label{thm:type-extraction}
For $0<\delta<1$,
\begin{equation}
 \log K_n^*(\delta)
 =nH_{Q_nW}(X|Z,U)+\sqrt{nV_Z(Q)}\,\Phi^{-1}(\delta)+o(\sqrt n).
 \label{eq:type-extraction}
\end{equation}
When $V_Z(Q)=0$ the second-order coefficient is zero, under the coefficient-closure convention. Channel zeros are allowed.
\end{theorem}
\begin{proof}
Write $T=\mathcal T_n(Q_{n,X|U}\mid u^n)$, and let $R_{Z^n}=\prod_i(Q_{n,X|u_i}W)$ be the product output reference. Under the product input, the probability $\pi$ of $T$ is at least $(n+1)^{-|\mathcal X||\mathcal U|}$. The output-image likelihood identity is
\[
 \frac{P_{Z^n}(z^n)}{R_{Z^n}(z^n)}
 =\frac{\widetilde P\{X^n\in T\mid Z^n=z^n\}}\pi
 \le(n+1)^{|\mathcal X||\mathcal U|}.
\]
Its logarithm also has lower tail at most $e^{-\gamma}$. Thus this correction is $O_{\PP}(\log n)$, by the same conditioning comparison as in Theorem~\ref{new:dense}. Bayes' rule now gives the private-image spectrum exactly:
\begin{align}
 -\log P_{X^n|Z^n,u^n}(X^n|Z^n,u^n)
 &=\log|T|-\sum_i\log\frac{W(Z_i|X_i)}{(Q_{n,X|u_i}W)(Z_i)}
       +\log\frac{P_{Z^n}(Z^n)}{R_{Z^n}(Z^n)}.\label{eq:type-private-identity}
\end{align}
The first term is $nH_{Q_n}(X|U)+O(\log n)$. Given any source word in $T$, the middle sum has the same independent-increment law, with mean $nI_{Q_nW}(X;Z|U)$ and variance $nV_Z(Q_n)$. The smallest positive channel probability and the positive conditional input probabilities bound all sampled scores on the retained support. Consequently the centered expression in \eqref{eq:type-private-identity}, divided by $\sqrt n$, converges to $\mathcal N(0,V_Z(Q))$.

Apply the lower and upper private-image tests of Lemma~\ref{lem:private-spectrum} with $\gamma=3\log(n+1)$. At positive variance the scalar Gaussian distribution is continuous and strictly increasing, so coefficients below and above $\sqrt{V_Z(Q)}\Phi^{-1}(\delta)$ give matching direct and converse bounds. If the variance is zero, Chebyshev's inequality and coefficients $-\eta$ and $\eta$ give coefficient zero after $\eta\downarrow0$.
\end{proof}

For a binary constant-composition source with half zeros and half ones and a BSC observer of crossover $p\in(0,1/2)$, the formula is
\[
 \log K_n^*(\delta)
 =nh(p)+\sqrt{np(1-p)}\log\frac{1-p}{p}\,\Phi^{-1}(\delta)+o(\sqrt n)
\]
along compatible blocklengths, where $h(p)=-p\log p-(1-p)\log(1-p)$. The observer and any public seed are included in the secrecy criterion. This is an exact extraction theorem for the prescribed source experiment; reliable agreement from arbitrary iid observations additionally needs an admissible reconciliation step.

\begin{corollary}[Extraction invariant under predictable public refinement]\label{cor:public-extraction}
Let $(V_n,E_n)$ be finite-source experiments with
\[
 \frac{-\log P(V_n\mid E_n)-a_n}{\sqrt n}
       \Rightarrow\cN(0,V),\qquad V>0,
\]
and let $J_n=f_n(V_n)$ be public with
$-\log P(J_n\mid E_n)=o_P(\sqrt n)$. Then
\begin{equation}
 \log K_n^*(\delta;J_n)
       =a_n+\sqrt{nV}\,\Phi^{-1}(\delta)+o(\sqrt n).
 \label{eq:public-extraction-rate}
\end{equation}
A sufficient condition is $H(J_n\mid E_n)=o(\sqrt n)$.
\end{corollary}
\begin{proof}
Choose $t_n=o(\sqrt n)$ with $P(T_{J_n|E_n}>t_n)\to0$, and choose $\gamma_n\to\infty$ with $\gamma_n=o(\sqrt n)$. The bounds of Proposition~\ref{prop:public-refinement} at coefficients on either side of $\sqrt V\Phi^{-1}(\delta)$ converge to the same strictly increasing normal CDF. A slow favorable diagonal adjustment gives the fixed-$\delta$ optimum. Markov's inequality proves the sufficient condition.
\end{proof}
At positive conditional varentropy, applying this corollary to Theorem~\ref{thm:type-extraction} preserves its center and conditional varentropy after every such public refinement. For example, let $V_n=(A^n,B^n)$ with independent components, let $E_n=A^n$, and publish $J_n=A^n$. The message may have exponentially many values, but its conditional surprisal is zero and extraction from $B^n$ is unchanged. The same second-order conclusion holds if $E_n$ reveals all but $o(\sqrt n)$ symbols of $A^n$ on a fixed finite alphabet. If instead a public header has only $L_n$ values, counting gives
\[
 P\{-\log P(J_n\mid E_n)>\log L_n+\gamma\}\le e^{-\gamma}.
\]
Thus an $o(\sqrt n)$-nat selector header also preserves the extraction coefficient. These statements concern extraction; agreement and a prescribed public-rate budget remain additional constraints in key generation.


\section{Comparison with Gaussian Coding Theorems}\label{sec:benchmarks}
The preceding applications use an image constraint as the main converse object. Gaussian Wyner--Ziv and Gray--Wyner coding provide useful comparisons: the former requires a source-observable allocation, whereas the latter admits a supporting-weight characterization on its positive-common-rate boundary.

\subsection{Bivariate Gaussian Gray--Wyner coding}\label{sec:gw-comparison}
\edit{A centralized encoder can observe the empirical covariance of a Gaussian source pair. This enables a covariance-shell construction, whereas the decoder-only observation in Wyner--Ziv coding leaves a residual allocation problem. The distinction explains why the two exact second-order formulas have different forms.}

Let $Z_t=(X_t,Y_t)^{\mathsf T}\sim\mathcal N(0,K)$ independently, where $K\succ0$, and let the two quadratic distortion thresholds be $D_1,D_2>0$. The encoder sends a common message and one private message to each decoder. At a boundary point with positive common rate, Zhou and Yu~\cite{zhouyu} characterize the closed second-order region by
\begin{equation}
 L_0+\lambda_1 L_1+\lambda_2 L_2
 \ge\sqrt{V_{\rm GW}}\,Q^{-1}(\varepsilon),\qquad
 V_{\rm GW}=\frac{1+(\lambda_1+\lambda_2-1)^2}{2},
 \label{eq:v10-gw-benchmark}
\end{equation}
together with nonnegativity of $L_i$ when its first-order private rate is zero. Here $\lambda_i$ is the supporting weight associated with private rate $r_i$, and error means that either actual distortion constraint fails.

\edit{The weighted-image calculation expresses the same variance through a quadratic source score. Put $a_i=D_ie^{2r_i}$, and let $S_*$ maximize $\log\det S$ over $0\prec S\preceq K$ and $S_{ii}\le a_i$. For supporting multipliers $\Lambda\succeq0$ and $\mu_i\ge0$, the stationarity and complementarity relations are}
\begin{equation}
 S_*^{-1}=\Lambda+\operatorname{diag}(\mu),\qquad
 \Lambda(K-S_*)=0,\qquad \mu_i(a_i-S_{*,ii})=0.
 \label{eq:v10-gw-kkt}
\end{equation}
Set $\lambda_i=\mu_i a_i=\mu_iS_{*,ii}$ and $A=(K^{-1}-\Lambda)/2$. The centered score
$\operatorname{tr}[A(ZZ^{\mathsf T}-K)]$ has variance $2\operatorname{tr}[(AK)^2]$. At a positive-common-rate point $K-S_*\ne0$, so \eqref{eq:v10-gw-kkt} implies $\operatorname{rank}\Lambda\le1$. Thus one eigenvalue of
$H=K^{1/2}AK^{1/2}$ is $1/2$. Moreover,
\[
 \operatorname{tr}(\Lambda K)=\operatorname{tr}(\Lambda S_*)
       =2-\lambda_1-\lambda_2,
 \qquad \operatorname{tr}H=\frac{\lambda_1+\lambda_2}{2}.
\]
The other eigenvalue is $(\lambda_1+\lambda_2-1)/2$, proving
\begin{equation}
 2\operatorname{tr}[(AK)^2]
 =\frac{1+(\lambda_1+\lambda_2-1)^2}{2}.
 \label{eq:v10-gw-equivalence}
\end{equation}
\edit{The identities compare two representations of the same bivariate fluctuation. A weighted-image converse retains this source score on every successful subset, as in Lemma~\ref{v5:entropy-spectrum}; the coding theorem additionally requires a construction. The uniform covariance-shell covering lemma of Zhou and Yu~\cite{zhouyu} supplies that construction and explicitly resolves the bivariate boundary regimes.}

For example, $K=\left(\begin{smallmatrix}1&0.7\\0.7&1\end{smallmatrix}\right)$, $D_1=D_2=0.2$, and $r_1=r_2=\tfrac12\log3$ give
$S_*=\left(\begin{smallmatrix}0.6&0.3\\0.3&0.6\end{smallmatrix}\right)$ and
$\lambda_1=\lambda_2=2/3$. Both expressions yield $V_{\rm GW}=5/9$. The agreement in \eqref{eq:v10-gw-equivalence} concerns the positive-common-rate bivariate model; limits in which a rate vanishes require the corresponding boundary convention in the coding theorem.



\subsection{The scalar Gaussian Wyner--Ziv benchmark}\label{v5:wz-section}
The reconstruction image of Section~\ref{sec:cell-reference} retains an uninformed reproduction in addition to a decoder using side information. The scalar Wyner--Ziv benchmark isolates the information--distortion constraint when only the informed reconstruction is required. Xu and Zhou~\cite{xuzhou} establish its exact profile by a source-dependent shell construction and an arbitrary-code converse.

Use the scaled observation model
\begin{equation}
 X\sim\cN(0,P),\quad Y=(1-c)X+N,\quad
 N\perp X,\quad N\sim\cN(0,v(1-c)),\quad c=v/P\in(0,1).
 \label{v5:wz-model}
\end{equation}
Then $\E[X\mid Y]=Y$ and $\Var(X\mid Y)=v$. The encoder observes $X^n$ only, and the decoder uses its message and $Y^n$. Fix $0<D<v$, write $d=D/v$, $\rho=1-c$, $\beta=\rho d$, and $R=\frac12\log(v/D)$. Define
\begin{equation}
 V_s=c^2/2,\qquad
 \Sigma=\begin{pmatrix}
 (\rho-\beta)(1-(\rho-\beta)/2)&-\beta(\rho-\beta)/2\\
 -\beta(\rho-\beta)/2&\beta(1-\beta/2)
 \end{pmatrix}.
 \label{v5:wz-Sigma}
\end{equation}
The second coordinate is the distortion fluctuation multiplied by $1/(2D)$. For independent $G_s\sim\cN(0,V_s)$ and $(J,B)\sim\cN(0,\Sigma)$, put
\begin{equation}
 F_{\rm G}(L)=\E\sup_t\Phi_\Sigma(L-G_s-t,t),\qquad
 L_{\rm G}(\eps)=F_{\rm G}^{-1}(1-\eps).
 \label{v5:wz-profile}
\end{equation}
\begin{lemma}[Scalar Gaussian Wyner--Ziv characterization {\cite{xuzhou}}]\label{v5:wz-inner}
For fixed interior parameters in \eqref{v5:wz-model},
\begin{equation}
 \log M^*(n,D,\eps)=nR+\sqrt nL_{\rm G}(\eps)+o(\sqrt n),\qquad 0<\eps<1.
 \label{v5:wz-upper}
\end{equation}
The profile is continuous and strictly increasing, and
\begin{equation}
 F_{\rm G}(L)<\Phi(\sqrt2L),\qquad
 L_{\rm G}(\eps)>Q^{-1}(\eps)/\sqrt2.
 \label{v5:wz-interval}
\end{equation}
\end{lemma}
\begin{proof}[Identification and interpretation]
Rescaling the observed variable to the additive model of~\cite{xuzhou} gives regression coefficient $\rho$ and distortion-to-noise ratio $\beta$. Equations~\eqref{v5:wz-Sigma}--\eqref{v5:wz-profile} are therefore its Theorem~1 and Eq.~(23). A strictly larger offset gives a strict error margin; a slow favorable adjustment converts the asymptotic definition to the fixed-$\eps$ expansion. Each finite maximizing rectangle excludes positive Gaussian mass from its containing halfspace. Since $\Var(G_s+J+B)=1/2$, this proves the strict comparison, as in~\cite[Corollary~1]{xuzhou}.
\end{proof}
For $(P,v,D,\eps)=(2,1,0.2,0.1)$ the profile gives $L_{\rm G}\approx1.5121711120$, whereas the both-terminal coefficient is $0.9061938024$~\cite{conditionalrd,rd}. At $D=1/2$ with $P=2,v=1$, it simplifies to $F_{\rm G}(L)=\Phi(L)^2$~\cite[App.~B]{xuzhou}. The optimization is performed after observing the source score and before observing $(J,B)$. A union over allocations after the residual noise would describe the different halfspace experiment.

The two-level image in Theorem~\ref{thm:gaussian-distortion-image} has the same form of residual optimization but source variance $1/2$. That variance follows from retaining the coarse distortion in the cell-reference identity. \edit{Its proof therefore requires the coarse-distortion entropy and enhancement identities in Appendix~\ref{app:hb-image}, in addition to the residual Gaussian analysis.}


\section{Nonmatching Rate Bounds}\label{sec:nontight}
General helper cells introduce a conditional likelihood loss; unrestricted lossy reconstruction leads to a fractional-image inverse; two noisy reconstruction stages require compatible residual geometry; and reliable-key generation couples decoding to a public key partition. The following bounds preserve those objects and state the additional estimate needed for each model.

\subsection{Subset-rate bounds for general helper cells}\label{sec:nested-bounds}
For a helper code, fix its helper index and any source-type restriction before forming the input law. Let $A^n$ have the resulting law and pass it through $W_{B|A}^{\otimes n}$. Write $B=(D,C)$ for a target subset and its complement, and let $\mathcal D$ be the actual correctly reconstructed set. Separate encoders imply at most $M(S)=\prod_{j\in S}M_j$ target words in each section indexed by $c^n$.

The original-law source predictors are $q_i=P(A_i\in\cdot\mid A^{i-1})$. Define
\begin{align}
 T_{S,n}&=\sum_i-\log\frac{(q_iW)(B_i)}{(q_iW_C)(C_i)},\label{eq:nested-score}\\
 \Lambda_{C,n}&=\sum_i\log\frac{(q_iW_C)(C_i)}{P(C_i\mid C^{i-1})},\qquad
 \mathcal I_{C,n}=\sum_i I(C_i;A^{i-1}\mid C^{i-1}).\label{eq:predictive-deficit}
\end{align}
For $\gamma,\eta>0$ the subset-rate converse is
\begin{equation}
 P\{B^n\in\mathcal D\}
 \le P\{T_{S,n}\le\log M(S)+2\gamma+\eta\}
        +2e^{-\gamma}+\frac{1+\mathcal I_{C,n}}\eta.
 \label{eq:nested-rate-bound}
\end{equation}
Several such constraints can be retained on the same joint success event, with the corresponding sum of exceptional probabilities.

To verify \eqref{eq:nested-rate-bound}, let $\Lambda_{B,n}$ be the analogous joint-output prediction correction. The exact identity is
\[
 T_{S,n}=-\log P(D^n\mid C^n)-\Lambda_{B,n}+\Lambda_{C,n}.
\]
Section counting controls the first term on success; the source-past prediction supermartingale in Lemma~\ref{comp:hi:lem:weighted-spectrum} gives $P\{\Lambda_{B,n}<-\gamma\}\le e^{-\gamma}$. The same comparison for $C$ gives $P\{\Lambda_{C,n}<-t\}\le e^{-t}$. Moreover the original memoryless channel yields $\E\Lambda_{C,n}=\mathcal I_{C,n}$, so $\E(\Lambda_{C,n})_+\le1+\mathcal I_{C,n}$. Markov's inequality proves the displayed rate bound. All identities and tail estimates are evaluated under the original cell-induced source--channel law.

Taking $\gamma=3\log(n+1)$ and $\eta=t\sqrt n$ exhibits a subset-rate loss of $t\sqrt n+O(\log n)$ and an exceptional probability at most $(1+\mathcal I_{C,n})/(t\sqrt n)+o(1)$. If the predictive deficit is uniformly $o(\sqrt n)$ on the relevant cells, a vanishing normalized threshold loss is possible. When a conditional type is available, Theorem~\ref{v5:screened} instead uses the conditional divergence $K_{j,n}$ and can be substantially sharper. A matching second-order characterization for arbitrary helper codes requires a uniform small-deficit estimate over code-induced cells and a common encoder-observable choice realizing the conditional means. Alternatively, Corollary~\ref{cor:partition-profile} permits regular component images whenever one source-measurable refinement label has negligible conditional information at every active section. Proposition~\ref{prop:structured-cells} demonstrates that this condition can hold when the whole-cell conditional divergence is linear. A general helper-code theorem still has to construct that partition uniformly, preserve the complete first-order means, and implement any header at the correct encoder. Equation~\eqref{eq:nested-rate-bound} isolates the prediction loss for the direct, unpartitioned route.

\subsection{General finite-alphabet both-lossy Heegard--Berger coding}\label{sec:lossy-hb-bounds}
Let $(X_i,Y_i)$ be iid on finite alphabets. The encoder observes $X^n$ and sends one common message. The uninformed decoder reconstructs within distortion $D_0$, and the informed decoder uses $Y^n$ to reconstruct within distortion $D_1$. Error is the probability that either block constraint fails. Encoder knowledge of $Y^n$ gives the distinct Kaspi setting~\cite{kaspi}. The Gaussian evaluation in Corollary~\ref{cor:hb-gaussian-rate} uses quadratic geometry; here the exact operational reduction is the policy-image inverse.

Use these two acceptance indicators in the reconstruction-policy image of Theorem~\ref{v5:policy-lp}. Denote its blocklength-$n$ value by $\mathsf F_n$ and its inverse by $\mathsf M_n(\varepsilon)$. For every feasible blocklength and $0<\delta_n<\varepsilon$,
\begin{equation}
 \log\mathsf M_n(\varepsilon)
 \le\log M_n^*(D_0,D_1,\varepsilon)
 \le\log\mathsf M_n(\varepsilon-\delta_n)
                  +\log\log(1/\delta_n)+O(1).
 \label{eq:hb-operational-inverse}
\end{equation}
The constant absorbs the integer ceiling. For any chosen first-order center $R$, put
\[
 L_n^{\rm frac}(\varepsilon)
 =\frac{\log\mathsf M_n(\varepsilon)-nR}{\sqrt n},\qquad
 L_n^*(\varepsilon)
 =\frac{\log M_n^*(D_0,D_1,\varepsilon)-nR}{\sqrt n}.
\]
With $\delta_n=n^{-2}$,
\begin{equation}
 L_n^{\rm frac}(\varepsilon)
 \le L_n^*(\varepsilon)
 \le L_n^{\rm frac}(\varepsilon-n^{-2})
                   +O\!\left(\frac{\log\log n}{\sqrt n}\right).
 \label{eq:hb-second-order-bracket}
\end{equation}
Taking $R$ to be the model's first-order optimal rate centers these inequalities at its coding boundary. The fractional images retain the complete block reconstruction policies and the overlap between acceptable sets. The remaining asymptotic question is whether their success profiles have a locally uniform, strictly crossing limit. Such a limit would identify the operational second-order coefficient. Comparisons with refined direct bounds such as~\cite{lili} use the same observation pattern and joint excess-distortion criterion.



\subsection{Nested images for two degraded observations}\label{sec:degraded-images}
Consider a common message and two physically degraded observations $X-Y_2-Y_1$. The encoder selects reproductions $U_0=h_0(X)$ and $U_1=h_1(X)$, which must agree with the two decoder outputs on success. A reference joint law $Q_{X,U_0,U_1}$ induces its observation marginals through the original degraded channel. Define
\begin{align}
 \imath_0&=\log\frac{q_{X\mid U_0}(X\mid U_0)}{p_X(X)}
           -\log\frac{q_{Y_1\mid U_0}(Y_1\mid U_0)}{p_{Y_1}(Y_1)},\nonumber\\
 \imath_1&=\log\frac{q_{X\mid U_0,U_1}(X\mid U_0,U_1)}{q_{X\mid U_0}(X\mid U_0)}
           -\log\frac{q_{Y_2\mid U_0,U_1}(Y_2\mid U_0,U_1)}{q_{Y_2\mid U_0}(Y_2\mid U_0)}.
 \label{eq:degraded-nested-score}
\end{align}
\begin{proposition}[Nested agreement-image test]\label{prop:degraded-image}
For every joint agreement event $E$ of an $M$-message code,
\begin{equation}
 \E[e^{\imath_0+\imath_1}\one_E]\le M,
 \qquad
 \Pr(E)\le\Pr\{\imath_0+\imath_1\le\log M+\gamma\}+e^{-\gamma}.
 \label{eq:degraded-image-bound}
\end{equation}
\end{proposition}
\begin{proof}
Multiply the original density by the exponential score. The result is
\[
 p_{Y_1}(y_1)\,
 q_{X\mid U_0,U_1,Y_2}(x\mid u_0,u_1,y_2)\,
 q_{Y_2\mid U_0,Y_1}(y_2\mid u_0,y_1).
\]
For fixed $(m,y_1)$, agreement fixes $u_0$. For fixed $y_2$, agreement fixes $u_1$, so integrating over its successful source section gives at most one. The remaining conditional density of $Y_2$ integrates to one. Finally integrate over $Y_1$ and sum over $m$. The degraded reference channel is essential to this factorization.
\end{proof}

For two noisy Gaussian observations, evaluating this test still requires the cross moment of the two code-induced reconstruction residuals. Even if individual repairs give
$\|w_j\|^2/n\to D_j$ and $\langle X,w_j\rangle/n\to D_j$, Cauchy--Schwarz only restricts their limiting cross moment to
\begin{equation}
 \frac{D_0D_1}{p}\ \pm\
 \sqrt{\left(D_0-\frac{D_0^2}{p}\right)
       \left(D_1-\frac{D_1^2}{p}\right)}.\label{eq:degraded-cross-interval}
\end{equation}
A stronger decoder need not know the weaker decoder's selected reproduction. Hence the nested projection used in Appendix~\ref{app:cr-details} is not automatically implementable in this model. The missing estimate is a joint cell-geometry characterization under the original observations, not a separate marginal Gaussian approximation.

For ordinary successive refinement with $Y_j=X+Z_j$ and physical degradation $q_1>q_2$, giving $Y_1^n$ also to the encoder and stronger decoder gives an explicit outer image region. Subtract the shared conditional mean and use the independent innovation with parameters
\begin{equation}
 p_{\rm eff}=\frac{pq_1}{p+q_1},\qquad
 q_{\rm eff}=\frac{q_1q_2}{q_1-q_2}.\label{eq:degraded-genie}
\end{equation}
The stronger innovation observation is obtained from
$(Y_2-(q_2/q_1)Y_1)/(1-q_2/q_1)$ by subtracting $\E[X\mid Y_1]$.
Its added noise is independent of the source innovation and has variance $q_{\rm eff}$. If $0<a<p_{\rm eff}$ and $0<b<aq_{\rm eff}/(a+q_{\rm eff})$, Corollary~\ref{cor:sr-image} with these parameters bounds the original two-noisy-decoder region from outside. A matching original-code construction and converse need a nested, source-observable allocation for both noisy reconstructions.


\subsection{Reconciliation followed by extraction}\label{sec:reconcile-app}
Consider the prescribed conditional-type source of Section~\ref{sec:pa-app}, now with a joint channel $W_{YZ|X}$. Alice observes $X^n$, Bob observes $Y^n$, and the observer sees $Z^n$. The sequence $u^n$ is fixed and public. Alice sends one public message and both parties use the same fixed hash; the observer sees that message. No encoder sees either channel output. Let
\[
 H_Y(Q)=H_{QW}(X|Y,U),\quad H_Z(Q)=H_{QW}(X|Z,U),
\]
and let $V_Y(Q),V_Z(Q)$ be the two conditional variances in \eqref{eq:private-variance}, using the respective marginal channels. Assume the same fixed-support regularity and $H_Z(Q)>H_Y(Q)$.

\begin{corollary}[A direct public-rate/key-rate inner bound]\label{cor:reconciliation}
For $0<\varepsilon,\delta<1$, the following offsets are achievable by reconstruction followed by hashing:
\begin{align}
 \log M_n&=nH_Y(Q_n)+\sqrt n A+o(\sqrt n),\nonumber\\
 \log K_n&=n[H_Z(Q_n)-H_Y(Q_n)]+\sqrt n L+o(\sqrt n),\label{eq:reconcile-budgets}\\
 A&\ge\sqrt{V_Y(Q)}\,Q^{-1}(\varepsilon),\qquad
 A+L\le\sqrt{V_Z(Q)}\,\Phi^{-1}(\delta).\label{eq:reconcile-region}
\end{align}
The key-disagreement probability is at most $\varepsilon$ and the total variation of Alice's actual key from uniformity, independent of the observer and public message, is at most $\delta$. Favorable sub-$\sqrt n$ adjustments apply at boundary coefficients. The achievable region corresponds to public reconstruction followed by hashing.
\end{corollary}
\begin{proof}
For an arbitrary finite pair $(V,Y)$, independent uniform bin labels and maximum-posterior decoding give a deterministic reconciliation map with error at most
\[
 \Pr\{-\log P_{V|Y}(V|Y)>\log M-\gamma\}+e^{-\gamma}.
\]
Indeed, for fixed $(v,y)$, at most $1/P(v|y)$ source words have conditional probability at least $P(v|y)$. A collision with any such competing word has probability $1/M$. On the complementary information event their total collision probability is at most $e^{-\gamma}$. Average and fix one map.

Apply this statement to $(X^n,Y^n)$ with public $u^n$. The conditional-image identity \eqref{eq:type-private-identity}, now for Bob's output, gives the first Gaussian inequality in \eqref{eq:reconcile-region}. For the chosen reconciliation map, Lemma~\ref{lem:public-composition} applies uniformly without any assumption on its public-bin masses. Applying \eqref{eq:type-private-identity} for the observer gives the second inequality. Fix a hash attaining the variation bound. If Bob reconstructs $X^n$ correctly, his hash agrees with Alice's, so hashing does not increase the disagreement probability. The secrecy calculation includes every outcome, not only successful reconstruction. Choose logarithmic $\gamma$ and take the closed coefficient limit.
\end{proof}
Reliability and secrecy are measured separately, which gives the two spectral inequalities in \eqref{eq:reconcile-region}. Establishing tightness for the reliable-key task requires a converse that also accounts for source preprocessing and the full public transcript.


\section{Discussion and Conclusion}\label{sec:discussion}
\subsection{What the image must retain}
\edit{The results distinguish an image bound from the code that realizes it. A section-cardinality constraint records what a decoder can distinguish after its complement is fixed. A reconstruction-dependent reference records the actual reproduction assigned to a source cell. A private image records the probability mass remaining after public information. The corresponding second-order expressions differ because these operational objects permit different choices before the residual experiment.}

\edit{Recoverable partitions enlarge the class of cells that can be analyzed without requiring the whole cell to resemble a single product law. The structured-source example makes the distinction explicit: both global deficits grow linearly, yet the component label is recoverable from the observation. The finite-block merging bound charges the missing conditional information and preserves the component spectrum. A total-image constraint still pays for the label when no complement is present.}

\subsection{Conclusions and remaining image constraints}
\edit{Cell-dependent Gaussian references retain two active distortion deficits in one likelihood. The paired converse comparisons cancel their common analytical costs and lead to the same source-observable profile as the shell construction. This yields the stated Heegard--Berger, common-reconstruction, and successive-refinement results under their respective joint error criteria. Posterior-mass images give a parallel refinement for extraction: the relevant public cost is the label's conditional information, not necessarily the logarithm of its alphabet.}

\edit{The nonmatching models identify the limits of these mechanisms. Arbitrary helper codes require regular components with compatible encoder operations; two noisy reconstruction stages require joint residual geometry for one admissible description; and reliable-key generation requires one partition satisfying both reliability and secrecy. Establishing those properties, rather than substituting an independently optimized dispersion for each constraint, is the remaining operational step.}


\appendices
\renewcommand{\thetheorem}{\thesection.\arabic{theorem}}
\renewcommand{\theequation}{\thesection.\arabic{equation}}
\makeatletter
\@addtoreset{theorem}{section}
\@addtoreset{equation}{section}
\makeatother
\renewcommand{\theHtheorem}{appendix.\thesection.\arabic{theorem}}
\renewcommand{\theHequation}{appendix.\thesection.\arabic{equation}}

\section{Posterior Localization and Helper Implementation}\label{app:posterior}
\subsection{The simultaneous image bound}
\begin{proof}[Proof of Theorem~\ref{comp:hi:lem:vector-transfer}]
We give the posterior localization argument because it is the step
that permits additional marks to be appended to the scalar theorem.
Write
\begin{equation}
 \begin{aligned}
 r(q,a)&=\log\frac{q(a)}{p(a)}+cm(q,a)-j(a)-d(q),\\
 D_n&=\sum_i d(q_i),\qquad R_n=\sum_i r(q_i,A_i),\\
 N_n&=\frac1{\sqrt n}\sum_i
       \big[\boldsymbol\psi(q_i,B_i)-\E[\boldsymbol\psi(q_i,B_i)|A_i,q_i]\big].
 \end{aligned}
 \label{comp:hi:eq:posterior-sums}
\end{equation}
The source innovation has mean zero under $q$, and it vanishes on
$\mathcal C$.  Since $\prod_iq_i(A_i)=|\mathcal A_n|^{-1}$ and the
source type is $p_n$, the aggregate score satisfies the exact identity
\begin{equation}
 \begin{aligned}
 c\sum_i h(q_i,B_i)&=\log|\mathcal A_n|
       +n\langle j+\log p,p_n\rangle\\
 &\quad+D_n+R_n+\sqrt nN_{0,n}.
 \end{aligned}
 \label{comp:hi:eq:aggregate-identity}
\end{equation}
Subtracting $c\sum_it_{C_S}(B_i)$ gives each weighted row, with
exactly the same $D_n$ and $R_n$.  Every ancillary row has deterministic
mean $n\langle m_S,p_n\rangle$.  Apply the ordinary total-cardinality
spectrum and Lemmas~\ref{comp:hi:lem:ancillary} and
\ref{comp:hi:lem:weighted-spectrum} before any conditioning on success.
With $\gamma_n=3\log(n+1)$, they imply that success, outside a
probability of at most
$(2+|\mathscr A|+|\mathscr W|)e^{-\gamma_n}$, lies in
\begin{equation}
 \begin{cases}
 N_{0,n}+(D_n+R_n)/\sqrt n\le s_{0,n}+o(1),\\
 N_{S,n}^{\rm a}\le s_{S,n}^{\rm a}+o(1),&S\in\mathscr A,\\
 N_{S,n}^{\rm w}+(D_n+R_n)/\sqrt n
                      \le s_{S,n}^{\rm w}+o(1),&S\in\mathscr W.
 \end{cases}
 \label{comp:hi:eq:original-event}
\end{equation}

First control the source residual on low-gap histories.  The function
$v_r(q)=\sum_aq(a)r(q,a)^2$ is bounded and continuous, including on
simplex faces because $t\log^2t\to0$.  It vanishes on $\mathcal C$.
Retain innovation $i$ while $\sum_{j<i}d(q_j)\le K\sqrt n$; this
decision is measurable from $A^{i-1}$.  If $R_n^{(K)}$ is the retained
sum, source martingale orthogonality gives, for every $u>0$,
\begin{equation}
 \begin{gathered}
 \frac1n\E[(R_n^{(K)})^2]
 \le\omega(u)+\frac{C(K\sqrt n+d_{\max})}{nu},\\
 \omega(u)=\sup_{d(q)\le u}v_r(q)\downarrow0.
 \end{gathered}
 \label{comp:hi:eq:stopped-residual}
\end{equation}
Indeed, retained indices with $d(q_i)>u$ number at most the retained
gap divided by $u$; all other indices contribute at most $\omega(u)$
each.  Letting $n\to\infty$ and then $u\downarrow0$ proves
$R_n=o_{\PP}(\sqrt n)$ on $\{D_n\le K\sqrt n\}$, under the
original source law.

Next retain every first-order moment needed to identify the fiber.
Set $\widehat\mu_n=n^{-1}\sum_i\delta_{q_i}$.  Source martingale
differences give
\begin{equation}
 \begin{aligned}
 \int q\,d\widehat\mu_n-p_n&=o_{\PP}(1),\\
 \int H(q)\,d\widehat\mu_n-n^{-1}\log|\mathcal A_n|
      &=o_{\PP}(1),\\
 \Cov(N_n|A^n)-\int V(q)\,d\widehat\mu_n&=o_{\PP}(1).
 \end{aligned}
 \label{comp:hi:eq:full-moments}
\end{equation}
Each error has second moment $O(n^{-1})$.  For the second line, subtract
$H(q_i)$ from $-\log q_i(A_i)$ and use the exact source chain rule.
For the last line, subtract the $q_i$-average of each bounded covariance
entry from its value at $A_i$.  These are predictable source
martingale differences.  On $D_n\le K\sqrt n$,
$\int d\,d\widehat\mu_n\le K/\sqrt n$.  Any weak subsequential
limit of the empirical measures therefore has barycenter $p$, entropy
$a_0$, and support on $\mathcal C$.  It belongs to
$\mathcal M_{\boldsymbol r}$.  Compactness makes this conclusion
uniform, and hence
\begin{equation}
 \operatorname{dist}\left(\Cov(N_n|A^n),
       \{\Sigma_\mu:\mu\in\mathcal M_{\boldsymbol r}\}\right)
       =o_{\PP}(1)
 \label{comp:hi:eq:cov-localization}
\end{equation}
on the low-gap event.  Keeping only the weighted optimality condition
would not yield the entropy restriction in this conclusion.

Condition now on the entire $A^n$.  The $q_i$ are fixed and the
coordinates $B_i$ are independent.  The centered vector increments in
\eqref{comp:hi:eq:posterior-sums} are uniformly bounded, so their
characteristic-function expansion has a uniform $O(n^{-1/2})$ cubic
remainder at each fixed argument.  Compactness of the covariance family
then gives a uniform conditional multivariate central limit theorem.
If the limiting marginal variances are positive, every boundary of a
lower orthant has zero probability, even when the joint covariance is
singular.  Gaussian orthant probabilities are consequently continuous
uniformly on the compact covariance and threshold windows.  If a
marginal can vanish, apply the comparison first with a positive
threshold enlargement and a continuous approximation to the orthant;
send the enlargement to zero only in the closed coefficient region.

On low-gap histories, remove $R_n$ using
\eqref{comp:hi:eq:stopped-residual} and drop the nonnegative $D_n$ in all
rows in which it occurs.  This enlarges the joint lower-orthant event.
Equations \eqref{comp:hi:eq:cov-localization} and the conditional Gaussian
comparison give the right side of \eqref{comp:hi:eq:vector-image}.
Finally, $\E R_n^2\le Cn$ and $\E N_{0,n}^2\le C$.  If
$|s_{0,n}|\le B$, the high-gap part of the aggregate constraint with
$D_n>K\sqrt n$ has probability at most $C/(K-B-1)^2$.  Let $n$ tend
to infinity at fixed $K$, remove the small threshold enlargement, and
then let $K\to\infty$.  This proves the bound and its asserted
uniformity.  No Gaussian comparison has been conditioned on the
terminal success event.
\end{proof}
\subsection{Observable helper branches}
\begin{lemma}[Local helper implementation and weighted likelihood matching]
\label{comp:hi:lem:implementation}
Fix finitely many finite auxiliaries with posterior measures in
$\mathcal M_{\boldsymbol r}$.  On any bounded window
$z_n=\sqrt n(p_n-p)$, each auxiliary admits conditional types $T_n$
satisfying
\begin{equation}
 T_n=T+O(n^{-1/2}),\qquad
 I(p_n,T_n)=r_h+\frac{L_h}{\sqrt n}+O(n^{-1}).
 \label{comp:hi:eq:helper-prescription}
\end{equation}
The helper can select the auxiliary from its observed type, cover the
type class by compatible $u^n$, and describe that choice using
$nr_h+\sqrt nL_h+O(\log n)$ nats.  Under
$Q_n=p_nT_nW$,
\begin{equation}
 cnH_{Q_n}(B|U)=ncr_\Sigma+
       \sqrt n\,[\langle j,z_n\rangle-L_h]+O(1).
 \label{comp:hi:eq:aggregate-mean-ach}
\end{equation}
If $U\perp C$ at the limiting auxiliary, then for every compatible
$(a^n,u^n)$,
\begin{equation}
 \sum_i\log\frac{Q_n(C_i|u_i)}{P_C(C_i)}
                  =o_{\PP}(\sqrt n)
 \label{comp:hi:eq:weighted-matching}
\end{equation}
under the original channel $W^n(\cdot|a^n)$, uniformly on the window.
\end{lemma}
\begin{proof}
All active entries of $T$ are positive, by
Lemma~\ref{comp:hi:thm:input} and Bayes' rule.  There is a row-sum-zero
direction $H$ with $D_TI(p,T)[H]\ne0$.  Otherwise, transferring mass
between any two active entries of each row would force
$\log[T(u|a)/P_U(u)]$ to be constant in $u$ for each $a$.  Row
normalization would then imply $T(u|a)=P_U(u)$ and $I(U;A)=0$,
contradicting \eqref{comp:hi:eq:informative}.  The implicit function theorem
along $H$ sets the mutual information in
\eqref{comp:hi:eq:helper-prescription}; row-wise rounding changes a positive
channel by $O(n^{-1})$ and its mutual information by $O(n^{-1})$.
These estimates are uniform over the fixed finite auxiliary family.

For the resulting joint type, a random auxiliary word of its marginal
type is compatible with a fixed $a^n$ with probability at least
$\exp[-nI(p_n,T_n)-C\log(n+1)]$.  Draw sufficiently many such words
that this probability times the number of draws exceeds
$\log|T_n(p_n)|+1$.  The union bound leaves positive probability that
every source word is covered.  Thus a deterministic cover of size at
most $\exp[nI(p_n,T_n)+C'\log(n+1)]$ exists.  A type header costs
$O(\log n)$, and a finite branch index costs $O(1)$.

On this positive face, ordinary differentiation of the fixed-channel
objective gives
\begin{equation}
 I(p_n,T_n)+cH_{p_nT_nW}(B|U)
 =\phi_c+n^{-1/2}\langle j,z_n\rangle+O(n^{-1}).
 \label{comp:hi:eq:fixed-objective-expansion}
\end{equation}
Its source tangent derivative is $j$ by
\eqref{comp:hi:eq:contact}, and its channel derivative vanishes at an
interior optimizer.  This calculation differentiates the fixed
positive test channel, not an optimized value in a neighborhood.
Subtracting \eqref{comp:hi:eq:helper-prescription} proves
\eqref{comp:hi:eq:aggregate-mean-ach}.

Finally, set $\rho_n(u,c)=\log[Q_n(c|u)/P_C(c)]$.  Independence at
the limiting law and strict positivity imply
$\sup_{u,c}|\rho_n(u,c)|=O(n^{-1/2})$.  For a compatible pair,
the conditional expectation of its sum is exactly
\begin{equation}
 n\E_{Q_n}\rho_n(U,C)
   =n\{I_{Q_n}(U;C)+D(Q_{n,C}\Vert P_C)\}=O(1).
 \label{comp:hi:eq:weighted-correction-mean}
\end{equation}
The last equality follows from the quadratic local expansion of
relative entropy at a positive independent law.  Conditional on
$(a^n,u^n)$, the summands are independent and their total variance is
$O(1)$.  Chebyshev's inequality proves
\eqref{comp:hi:eq:weighted-matching}.  This is stronger than equality of
first-order entropies: it controls the actual conditional likelihood
used by the decoder at order $\sqrt n$.
\end{proof}



\section{Proof of the Gaussian Reconstruction-Constrained Image Theorem}\label{app:hb-image}
This appendix proves Theorem~\ref{thm:gaussian-distortion-image}. All constants may depend on the fixed interior parameters $(p,q,a,b)$, but not on the source realization. The Gaussian posterior regularization and sampled-description device follow the construction of Xu and Zhou~\cite[Apps.~E--G]{xuzhou}. The cell-dependent likelihood and its coarse-distortion accounting are derived below; these are needed in addition to their scalar Wyner--Ziv theorem.

\subsection{Exact image identity and eventwise localization}
Fix a sequence of labelled images with
$\log M_n\le nR_{a,b}+\sqrt nL+o(\sqrt n)$, and first clip both reproductions coordinatewise to $[-h_n,h_n]$, where $h_n=\sqrt{8p\log n}$. On $\|X^n\|_\infty\le h_n$ clipping cannot increase either loss, and the omitted source mass tends to zero. \edit{Write $E_n$ for the success event of this clipped image. Its probability bounds that of the original code up to $o(1)$. It is enough to consider a subsequence with $\Pr(E_n)\ge c_0>0$.} Let $u=u(M)$ and $v=v(M,Y^n)$ denote the clipped reproductions.

The densities in the proof of Theorem~\ref{thm:gaussian-distortion-image} give
\begin{align}
 g(x,u,v,y)
 &:=\log\frac{\phi_{bI_n}(x-v)}{\phi_{pI_n}(x)}
       +\log\frac{\phi_{(a+q)I_n}(y-u)}{\phi_{qI_n}(y-x)}\nonumber\\
 &=nR_{a,b}+\frac{\|x\|^2}{2p}-\frac{\|x-v\|^2}{2b}
       -\frac{\|y-u\|^2}{2(a+q)}+\frac{\|y-x\|^2}{2q}.
 \label{eq:hb-g-expanded}
\end{align}
Writing $w=X^n-u$ and $Y^n=X^n+Z^n$ proves \eqref{eq:hb-exact-identity}, including all constants. Conditionally on $X^n$, the noise term $N_{0,n}$ in that identity has mean zero and variance
\begin{equation}
 \frac{q\|w\|^2}{(a+q)^2}+\frac{na^2}{2(a+q)^2}.\label{eq:hb-N0-variance}
\end{equation}
On the coarse distortion event this is at most $C n$. The source score $A_n$ is tight under the original law. Lemma~\ref{lem:cell-reference} bounds $H_r$ above and $Z_s$ below by using $\gamma=3\log(n+1)$. Both deficits in \eqref{eq:hb-exact-identity} are nonnegative on $E_n$. Therefore, in the eventwise sense,
\begin{equation}
 \one_{E_n}(na-\|w\|^2)=O_P(\sqrt n),\qquad
 \one_{E_n}(nb-\|X^n-v\|^2)=O_P(\sqrt n).\label{eq:hb-event-gaps}
\end{equation}
Precisely, for either deficit divided by $\sqrt n$, the probability of exceeding $K$ jointly with $E_n$ tends to zero as $K\to\infty$, uniformly in sufficiently large $n$. No conditional independence after conditioning on $E_n$ is invoked.

Let $s_n(x)=\Pr(E_n\mid X^n=x)$ and define the source-only flag
\begin{equation}
 F_n(x)=\one\{s_n(x)\ge n^{-1},\ \|x\|^2\le2pn,
                      \ \|x\|_\infty\le h_n\}.\label{eq:hb-flag}
\end{equation}
Since a source failing the coarse distortion has $s_n(x)=0$, the flag implies $\|x-u(f(x))\|^2\le na$. Gaussian tail bounds give
$\Pr(E_n,F_n=0)=o(1)$. Write $\beta_n=\Pr(F_n=1)$, which is bounded away from zero on the chosen subsequence, and $\widetilde P=P(\cdot\mid F_n=1)$. Under $\widetilde P$, the channel $Y^n\mid X^n$ is still $\cN(X^n,qI_n)$.

Append the flag to the message, $M_0=(M,F_n)$. Let $\Pi$ be an independent uniform coordinate permutation, set $k=\lceil n^{3/4}\rceil$, and quantize $X_{\Pi(1)},\ldots,X_{\Pi(k)}$ into intervals of width $\Delta_n=n^{-2}$, with two overflow labels. Denote this sampled description by $V^{\rm pre}$ and put $W=(M_0,V^{\rm pre})$. Conditional on $\Pi$, its alphabet satisfies
\begin{equation}
 \log|\mathcal W_\Pi|\le\log(2M_n)+O(k\log n).\label{eq:hb-prefix-alphabet}
\end{equation}
The source-only flag remains a coordinate of $W$, so posteriors on its retained cells agree under $P$ and $\widetilde P$. \edit{Conditional expectations given a retained $(X^n,\Pi)$ also agree under these two laws. This permits the centered likelihood estimate under $\widetilde P$ to be used in the original-law reduction of Lemma~\ref{lem:matched-cost}.}

\subsection{A reconstruction repair with a controlled cost}
We give the quantitative estimate used for posterior regularization. For a retained $x$, let
$\mathcal A_x=\{y:\|x-v(f(x),y)\|^2\le nb\}$.
Its probability under $\cN(x,qI_n)$ is at least $n^{-1}$. The Gaussian reverse H\"older inequality~\cite{chen} implies that for correlated standard Gaussians $(G_0,G_1)$ with cross covariance $\vartheta I_n$,
\[
 \Pr(G_0\in B,G_1\in A)\ge
       [\Gamma_n(B)\Gamma_n(A)]^{1/(1-\vartheta)}.
\]
Consequently, for the set where $\Pr(G_1\in A\mid G_0)<e^{-u}$,
\begin{equation}
 \Gamma_n(B)\le\Gamma_n(A)^{-1/\vartheta}
                       e^{-u(1-\vartheta)/\vartheta}.\label{eq:hb-reverse-tail}
\end{equation}
Choose $\delta_n=\sqrt{\log n/n}$, $\vartheta_n=1-\delta_n$, and
$u_n=(1+3\vartheta_n)\log n/\delta_n$.
For the ideal perturbation
$Y^*=\vartheta_nY+\delta_n x+\sqrt{q(1-\vartheta_n^2)}G'$, the conditional chance of belonging to $\mathcal A_x$ is at least $e^{-u_n}$ except on an observation event of probability $n^{-3}$.

The decoder can generate instead
$Y^{(j)}=\vartheta_nY+\sqrt{q(1-\vartheta_n^2)}G_j$.
For two Gaussians with common covariance $\sigma^2I$ and mean shift $v$, Cauchy--Schwarz gives
$P_0(A)\ge P_*(A)^2e^{-\|v\|^2/\sigma^2}$.
Since $\|x\|^2\le2pn$, each candidate hits $\mathcal A_x$ with conditional probability at least $e^{-c_n}$ outside the preceding exceptional event, where
\[
 c_n=2u_n+\frac{2pn\delta_n}{q(1+\vartheta_n)}
       =O(\sqrt{n\log n}).
\]
A list of $N_n=\lceil ne^{c_n}\rceil$ independent candidates therefore fails with probability at most $n^{-3}+e^{-n}$.

Clip all candidate reconstructions to $[-h_n,h_n]$. Their coordinate losses lie in $[0,4h_n^2]$, and the list is independent of $\Pi$. Sampling without replacement~\cite{hoeffding} and a union bound over the list give
\begin{equation}
 \widetilde P\left\{\max_{j\le N_n}
  \left|\frac1k\sum_{i\le k}(X_{\Pi(i)}-\widehat X_{\Pi(i)}^{(j)})^2
        -\frac1n\|X^n-\widehat X^{(j)}\|^2\right|>n^{-1/16}\right\}
 \le 2N_n\exp\left\{-\frac{k n^{-1/8}}{8h_n^4}\right\}=o(n^{-3}).
 \label{eq:hb-list-sampling}
\end{equation}
The exponent $n^{5/8}/(\log n)^2$ dominates $\log N_n$.
Select the candidate of smallest sampled loss using the quantized prefix. Quantization changes each sample loss by at most $4h_n\Delta_n+\Delta_n^2$. On list and sampling success its full-block distortion is at most
$b+2n^{-1/16}+8h_n\Delta_n+2\Delta_n^2$.
The exceptional loss is at most $4h_n^2$. Hence the resulting reconstruction, formed from $(W,Y^n,\Pi)$ and independent list randomness, satisfies
\begin{equation}
 \E_{\widetilde P}\|X^n-\widehat X^{\rm rep}\|^2
       \le nb+O(n^{15/16}).\label{eq:hb-repair}
\end{equation}
Independent private list randomness cannot improve on the posterior mean in squared error, so the same bound holds for the posterior mean error.

\subsection{The coarse deficit and posterior entropy}
Let $m_0=\E_{\widetilde P}[X^n\mid W,Y^n,\Pi]$ and
\[
 K_n(0)=\E_{\widetilde P}D(P_{X^n\mid W,Y^n,\Pi}\|\cN(m_0,bI_n)).
\]
Use the cell image reference $r_u=\phi_{(a+q)I_n}(\cdot-u)$ for each augmented cell and always normalize its likelihood under $P$. The posterior on a retained cell is unchanged by the source-only conditioning. Taking expectation in \eqref{eq:hb-exact-identity}, now with posterior reference centered at $m_0$, yields the exact relation
\begin{align}
 K_n(0)+\frac{\E_{\widetilde P}(na-\|X^n-u\|^2)}{2(a+q)}
 ={}&\E_{\widetilde P}H_{r,W}-nR_{a,b}
       -\E_{\widetilde P}\frac{\|X^n\|^2-np}{2p}\nonumber\\
 &+\frac{\E_{\widetilde P}\|X^n-m_0\|^2-nb}{2b}.
 \label{eq:hb-K-identity}
\end{align}
The centered channel term has expectation zero because the flag and $u$ are source-measurable. For each fixed permutation, the image-mass inequality is evaluated under $P$. Restriction to the source-only flag gives
$\E_{\widetilde P}e^{H_{r,W}}\le|\mathcal W_\Pi|/\beta_n$.
Jensen's inequality therefore applies without renormalizing the reference densities and gives
\[
 \E_{\widetilde P}H_{r,W}
 \le\log|\mathcal W_\Pi|+\log(1/\beta_n).
\]
The energy expectation in \eqref{eq:hb-K-identity} is $O(\sqrt n)$ by Cauchy--Schwarz and $\beta_n\ge c>0$. Equations~\eqref{eq:hb-prefix-alphabet} and \eqref{eq:hb-repair} imply
\begin{equation}
 K_n(0)=O(n^{15/16}),\qquad
 \E_{\widetilde P}(na-\|X^n-u\|^2)=O(n^{15/16}).\label{eq:hb-K-bounds}
\end{equation}
Both quantities on the left of \eqref{eq:hb-K-identity} are nonnegative. This is the additional entropy control supplied by the coarse-cell reference.

\subsection{A vanishing Gaussian enhancement}
Add an independent observation of strength $t\ge0$ to $Y^n$. Let $m_t,V_t$ be the posterior mean and covariance given $W$, the observations, and $\Pi$. Set $b_t=(b^{-1}+t)^{-1}$ and
$K_n(t)=\E_{\widetilde P}D(P_{X^n\mid W,\mathrm{obs},\Pi}\|\cN(m_t,b_tI_n))$.
Conditional I--MMSE and its derivative~\cite{immse,payaro} give
\begin{equation}
 K_n'(t)=-\frac1{2b_t}\E_{\widetilde P}\|V_t-b_tI_n\|_{\rm F}^2.\label{eq:hb-K-derivative}
\end{equation}
Indeed, differentiate
$K_n=\frac n2\log(2\pi b_t)+\E\tr V_t/(2b_t)-h(X^n\mid W,\mathrm{obs},\Pi)$,
use $b_t'=-b_t^2$, $h'=-\E\tr V_t/2$, and
$(\E\tr V_t)'=-\E\tr V_t^2$, and complete the square.
The retained prior is bounded, justifying differentiation and averaging.
Integrating over $[0,n^{-1/32}]$ gives a deterministic
$\nu_n\in(0,n^{-1/32}]$ such that
\begin{equation}
 K_n(\nu_n)=o(n),\qquad
 \E_{\widetilde P}\|V_{\nu_n}-b_{\nu_n}I_n\|_{\rm F}^2
       =O(n^{31/32})=o(n).\label{eq:hb-posterior-flat}
\end{equation}
\edit{Indeed, $K_n$ is nonnegative and nonincreasing, so the integral of the covariance defect over an interval of length $n^{-1/32}$ is at most $2bK_n(0)=O(n^{15/16})$. Its average is $O(n^{31/32})$. Choose a deterministic point attaining at most twice that average. The enhancement can depend on the code sequence, but not on the realized source, permutation, or noise.}

Write $\nu=\nu_n$,
$q_\nu=(q^{-1}+\nu)^{-1}$, $b_\nu=(b^{-1}+\nu)^{-1}$,
and let $A^{\rm ext}=\sqrt\nu X^n+G^{\rm ext}$.
The sufficient enhanced observation is
\[
 Y_\nu=q_\nu(q^{-1}Y^n+\sqrt\nu A^{\rm ext})=X^n+Z_\nu,
 \qquad Z_\nu\sim\cN(0,q_\nu I_n).
\]
\edit{The complementary observation is
\[
 V^{\rm anc}=\frac{G^{\rm ext}-\sqrt\nu Z^n}{\sqrt{1+\nu q}},
 \qquad Y^n=Y_\nu-\frac{q\sqrt\nu}{\sqrt{1+\nu q}}V^{\rm anc}.
\]
Its cross-covariance with $Z_\nu$ is zero. Joint Gaussianity makes it independent of $Z_\nu$, and both are independent of the source and permutation. Since $W$ is determined by the source and permutation, $V^{\rm anc}$ is independent of $(X^n,W,Y_\nu,\Pi)$. The independence remains valid after conditioning on the source-only flag.}
The enhanced center is
\begin{equation}
 R_\nu=\frac12\log\frac{p q_\nu}{b_\nu(a+q_\nu)},\qquad
 R_{a,b}-R_\nu=\frac12\log\frac{1+\nu s_a}{1+\nu b},\quad
 s_a=\frac{aq}{a+q}.\label{eq:hb-rate-enhance}
\end{equation}

\subsection{Upper image likelihood after enhancement}
Let $H_0=\log[r_u(Y^n)/p_{M_0,Y^n}(M_0,Y^n)]$ and
$H_{0,\nu}=\log[\phi_{(a+q_\nu)I_n}(Y_\nu-u)/p_{M_0,Y_\nu}(M_0,Y_\nu)]$.
The reference experiment $X^n\mid u\sim\cN(u,aI_n)$ induces, given $(u,Y^n)$,
\[
 A^{\rm ext}\sim\cN(\sqrt\nu m_a,(1+\nu s_a)I_n),\qquad
 m_a=u+\frac a{a+q}(Y^n-u).
\]
The ancillary factor cancels between this reference and the original experiment. Therefore
\begin{equation}
 H_0-H_{0,\nu}
 =\log\frac{p(A^{\rm ext}\mid M_0,Y^n)}
                {\phi_{(1+\nu s_a)I_n}(A^{\rm ext}-\sqrt\nu m_a)}.
 \label{eq:hb-enhance-gain}
\end{equation}
Compare the numerator first with
$\kappa=\phi_{(1+\nu b)I_n}(A^{\rm ext}-\sqrt\nu v)$.
The log ratio $p/\kappa$ has lower tail at most $e^{-\gamma}$.
Expanding the other Gaussian ratio gives
\begin{align}
 \log\frac\kappa{\phi_{(1+\nu s_a)I_n}(A^{\rm ext}-\sqrt\nu m_a)}
 ={}&n(R_{a,b}-R_\nu)
 +\frac\nu2\left[
 \frac{\|X^n-m_a\|^2-ns_a}{1+\nu s_a}
 -\frac{\|X^n-v\|^2-nb}{1+\nu b}\right]+N_n^{\rm ext},
 \label{eq:hb-gain-expansion}
\end{align}
where, conditionally on $(X^n,Y^n)$, $N_n^{\rm ext}$ is centered and has variance at most
$C\nu(\|X^n-m_a\|^2+\|X^n-v\|^2+n)$.
On $E_n$, the second bracketed term is nonnegative. Moreover
\[
 X^n-m_a=\frac q{a+q}(X^n-u)-\frac a{a+q}Z^n.
\]
Its centered energy is $O_P(\sqrt n)$ on $E_n$, by \eqref{eq:hb-event-gaps} and conditional Gaussian moments. Consequently the gain in \eqref{eq:hb-enhance-gain} is at least $n(R_{a,b}-R_\nu)-o_P(\sqrt n)$ on success. Using the budget tail for $H_0$ proves
\begin{equation}
 H_{0,\nu}\le nR_\nu+\sqrt nL+o_P(\sqrt n)
 \quad\text{on }E_n\cap\{F_n=1\}.\label{eq:hb-enhanced-upper}
\end{equation}
This calculation retains the change in rate in \eqref{eq:hb-rate-enhance}; the condition $\nu\to0$ alone would not justify discarding it.

Define the enhanced actual reconstruction
\[
 v^{\rm enh}=\frac{v+b\sqrt\nu A^{\rm ext}}{1+\nu b},\qquad e^{\rm enh}=X^n-v^{\rm enh}.
\]
It is measurable from $(M_0,Y_\nu,V^{\rm anc})$, independently of $\Pi$. Expansion of its squared error and conditional Chebyshev give
\begin{equation}
 \Pr\{E_n,F_n=1,\ \|e^{\rm enh}\|^2>nb_\nu+\eta\sqrt n\}\to0
 \quad(\eta>0).\label{eq:hb-enhanced-distortion}
\end{equation}
The centered enhancement-noise variance is $O(n\nu)$ on original success. Its coordinate maximum is $O_P(\sqrt{\log n})$ on the flag.

\subsection{Retaining the same sampled-description cost}
Condition on $\mathcal B=\sigma(M_0,Y_\nu,V^{\rm anc},\Pi)$.
Let $p_{\mathcal B}$ be the actual posterior and
$\kappa_{\mathcal B}=\phi_{b_\nu I_n}(\cdot-v^{\rm enh})$.
Let $\zeta=\log(p_{\mathcal B}(X^n)/\kappa_{\mathcal B}(X^n))$ and $\zeta_k$ be the corresponding likelihood ratio of the sampled quantization labels. Both $\zeta_k$ and $\zeta-\zeta_k$ have lower tails at most $e^{-\gamma}$ under $\widetilde P$.

Put
\[
 E_n^{\rm en}=\frac{\|e^{\rm enh}\|^2-nb_\nu}{2b_\nu},\qquad
 E_k^{\rm en}=\frac{\sum_{j\le k}(e^{\rm enh}_{\Pi(j)})^2-kb_\nu}{2b_\nu},
 \quad c_n^{\rm pre}=k\left[\log\frac1{\Delta_n}+\frac12\log(2\pi e b_\nu)\right].
\]
Integration over the Gaussian quantization intervals gives
\begin{equation}
 -\log\kappa^{\rm pre}(V^{\rm pre}\mid\mathcal B)
       =c_n^{\rm pre}+E_k^{\rm en}+o_P(1).\label{eq:hb-prefix-reference}
\end{equation}
On the high-probability coordinate cutoff the error is at most
$Ck\Delta_n\sqrt{\log n}=o(1)$.
Since the enhanced error vector is independent of $\Pi$,
\[
 \Var_\Pi\left(E_k^{\rm en}-\frac knE_n^{\rm en}\ \middle|\ e^{\rm enh}\right)
       \le Ck\|e^{\rm enh}\|_\infty^4.
\]
The right side is $O_P(k(\log n)^2)=o(n)$. Combining with \eqref{eq:hb-enhanced-distortion} gives, on original success,
\begin{equation}
 E_k^{\rm en}\le o_P(\sqrt n),\qquad
 E_k^{\rm en}-E_n^{\rm en}\ge-o_P(\sqrt n).\label{eq:hb-sampled-signs}
\end{equation}

Let
$H_{W,\nu}=\log[\phi_{(a+q_\nu)I_n}(Y_\nu-u)/p_{W,Y_\nu\mid\Pi}(W,Y_\nu\mid\Pi)]$.
Factoring the prefix probability gives
\begin{equation}
 H_{W,\nu}=H_{0,\nu}+c_n^{\rm pre}+E_k^{\rm en}-\zeta_k+o_P(1).
 \label{eq:hb-prefix-factor}
\end{equation}
Its upper bound follows from \eqref{eq:hb-enhanced-upper} and \eqref{eq:hb-sampled-signs}.
For the lower bound, expand $H_{0,\nu}$ as in \eqref{eq:hb-exact-identity} with reference mean $v^{\rm enh}$ and noise $q_\nu$. Substitution into \eqref{eq:hb-prefix-factor} yields
\begin{align}
 H_{W,\nu}={}&nR_\nu+c_n^{\rm pre}+\sqrt n A_n+N_{0,n}^{\nu}
 +\frac{na-\|X^n-u\|^2}{2(a+q_\nu)}\nonumber\\
 &+(E_k^{\rm en}-E_n^{\rm en})+(\zeta-\zeta_k)+o_P(1).
 \label{eq:hb-prefix-lower-identity}
\end{align}
The coarse deficit is nonnegative on the flag. Therefore, on original success outside vanishing-probability exceptions,
\begin{align}
 H_{W,\nu}&\le nR_\nu+c_n^{\rm pre}+\sqrt nL+o_P(\sqrt n),\label{eq:hb-paired-upper}\\
 H_{W,\nu}&\ge nR_\nu+c_n^{\rm pre}+\sqrt n A_n+N_{0,n}^{\nu}-o_P(\sqrt n).
 \label{eq:hb-paired-lower}
\end{align}
The same $c_n^{\rm pre}$ appears in both. Its order is larger than $\sqrt n$ and it is canceled, not neglected.

\subsection{Centered image law and the residual covariance}
Let $m=\E_{\widetilde P}[X^n\mid W,Y_\nu,\Pi]$ and define the source-measurable vectors
\[
 d_n=X^n-\E_{\widetilde P}[m\mid X^n,\Pi],\qquad
 \ell_n=\frac{d_n}{q_\nu}-\frac{X^n-u}{a+q_\nu}.
\]
Define
\begin{align}
 J_n&=\frac{\ell_n^{\mathsf T}Z_\nu}{\sqrt n}
 +\frac{a/(a+q_\nu)-b_\nu/q_\nu}{2q_\nu\sqrt n}
               (\|Z_\nu\|^2-nq_\nu),\label{eq:hb-J}\\
 B_n&=-\frac{d_n^{\mathsf T}Z_\nu}{q_\nu\sqrt n}
 +\frac{b_\nu}{2q_\nu^2\sqrt n}(\|Z_\nu\|^2-nq_\nu).
 \label{eq:hb-B}
\end{align}
Their sum is $N_{0,n}^{\nu}/\sqrt n$.
Differentiating the cell likelihood and posterior mean~\cite{dytso} gives
\[
 \nabla_y H_{W,\nu}=\frac{y-m}{q_\nu}-\frac{y-u}{a+q_\nu},\qquad
 \nabla_y m=\frac V{q_\nu}.
\]
For $v(y)=m(W,y,\Pi)-b_\nu y/q_\nu$, its derivative is
$(V-b_\nu I_n)/q_\nu$. Gaussian Poincar\'e~\cite{ledoux}, first for the centered likelihood residual and then for every component of $v$, gives
\begin{equation}
 \E_{\widetilde P}\left|
 H_{W,\nu}-\E_{\widetilde P}[H_{W,\nu}\mid X^n,\Pi]-\sqrt nJ_n
 \right|^2
 \le q_\nu^{-2}\E_{\widetilde P}\|V-b_\nu I_n\|_{\rm F}^2=o(n).
 \label{eq:hb-Poincare}
\end{equation}
It also gives a remainder $r_n^m$ with $\E\|r_n^m\|^2=o(n)$ in
\begin{equation}
 X^n-m=d_n-(b_\nu/q_\nu)Z_\nu-r_n^m.\label{eq:hb-error-decomposition}
\end{equation}

To identify the moments, compare $\widetilde P$ with the law having the same $(W,Y_\nu,\Pi)$ marginal and conditional source $\cN(m,b_\nu I_n)$. Its relative entropy is $K_n(\nu)=o(n)$. For an event with reference probability at most $e^{-cn}$, binary data processing gives
$\widetilde P(E)\le(K_n(\nu)+\log2)/(cn)$.
On each retained cell, $\|X^n\|\le\sqrt{2pn}$ and $\|X^n-u\|\le\sqrt{na}$. Conditional Jensen bounds the posterior mean, and the same cell bounds its fixed center $u$, so $\|m-u\|\le C\sqrt n$. Under the Gaussian reference the error is conditionally $\cN(0,b_\nu I_n)$. Its squared-norm and inner-product deviations of order $n$ have probability $e^{-c_\eta n}$ for every fixed tolerance $\eta>0$. The preceding relative-entropy comparison consequently transfers
\[
 \frac{\|X^n-m\|^2}{n}\to b,\qquad
 \frac{(m-u)^{\mathsf T}(X^n-m)}n\to0.
\]
Equation~\eqref{eq:hb-K-bounds} gives $\|X^n-u\|^2/n\to a$. Using \eqref{eq:hb-error-decomposition}, conditional Gaussian projections, and $\E\|d_n\|^2\le\E\|X^n\|^2\le2pn$, we obtain
\begin{equation}
 \frac{\|X^n-u\|^2}{n}\to a,\qquad
 \frac{(X^n-u)^{\mathsf T}d_n}{n}\to b,\qquad
 \frac{\|d_n\|^2}{n}\to b-\frac{b^2}{q}.
 \label{eq:hb-three-moments}
\end{equation}
\edit{For the two projections needed here, conditional noise variances divided by $n^2$ are bounded by $q_\nu a/n$ and $q_\nu\E\|d_n\|^2/n^2$, respectively. Both tend to zero. Also $\|r_n^m\|/\sqrt n\to0$ by its second-moment bound.} ''All limits are in $\widetilde P$-probability. The second limit uses
$(X^n-u)^{\mathsf T}(X^n-m)/n\to b$; the third follows by expanding \eqref{eq:hb-error-decomposition}.

Set $\rho=a/(a+q)$ and $\beta=b/q$. The conditional covariance of \eqref{eq:hb-J}--\eqref{eq:hb-B}, using \eqref{eq:hb-three-moments}, is
\[
 \Var J_n\to(\rho-\beta)-\tfrac12(\rho-\beta)^2,\quad
 \Var B_n\to\beta-\tfrac12\beta^2,\quad
 \Cov(J_n,B_n)\to-\tfrac12\beta(\rho-\beta).
\]
The linear projections of $Z_\nu$ are exactly Gaussian. Rotate them into at most two coordinate directions; the remaining squared noise norm is an independent chi-square fluctuation, and the two removed squared coordinates contribute $o_P(1)$ after normalization. This proves convergence along every deterministic source sequence satisfying \eqref{eq:hb-three-moments}. Conditional second moments are bounded on a slowly shrinking source-measurable good set. A finite-grid and tail argument therefore makes the conditional CDF convergence uniform over all thresholds. This permits unbounded source-measurable allocations.

Apply Lemma~\ref{lem:matched-cost} to \eqref{eq:hb-paired-upper}--\eqref{eq:hb-paired-lower} and \eqref{eq:hb-Poincare}. Remove the nonnegative localization indicator only after conditioning on the source and evaluating the Gaussian probability. Under the original source law $A_n\Rightarrow\cN(0,1/2)$. The resulting converse is $\limsup\Pr(E_n)\le F_{a,b}(L)$.

\subsection{Construction attaining the image bound}
Fix a source-energy cutoff $B<\infty$. Partition $A_n(x)$ in $[-B,B]$ into intervals of width $n^{-1}$; write $\alpha_j$ for the upper endpoint. There are polynomially many branches.
Use the backward Gaussian reference $X=U_0+W_0$ with independent variances $p-a$ and $a$. For a fixed $x$ in branch $j$, retain coarse words satisfying
\begin{equation}
 na-1\le\|x-u_0\|^2\le na.\label{eq:hb-coarse-shell}
\end{equation}
Under $Q_{U_0\mid X=x}$ the hit probability is at least $c_Bn^{-1/2}$: rotate along $x$ and apply Stirling's formula to the remaining chi-square density in a fixed-width interval about $n+O_B(\sqrt n)$. On this shell,
\[
 \log\frac{q_{X\mid U_0}(x\mid u_0)}{p_X(x)}
 =\frac n2\log\frac pa+\sqrt nA_n(x)+\frac{na-\|x-u_0\|^2}{2a}.
\]
Changing measure to the marginal $Q_{U_0}$ shows that
$N_j=\lceil\exp\{n\log(p/a)/2+\sqrt n\alpha_j+C\log n\}\rceil$
independent coarse candidates have failure at most $n^{-4}$ at every retained source word, for large enough $C$. Select the first hit using the source alone. Averaging later fixes one deterministic family of coarse lists.

Given a selected $u_0$, write $w=x-u_0$ and $y'=y-u_0=w+z$.
Let $T=bq/(q-b)<a$, and use the refinement reference
$W_0=V+E$ with independent Gaussian variances $a-T$ and $T$.
The reconstruction is
\[
 \widehat x=u_0+\frac bq y'+\left(1-\frac bq\right)V.
\]
For a bounded allocation $t$, restrict the conditional kernel of $V\mid W_0=w$ to a fixed-width shell about
\begin{equation}
 \|w-V\|^2=nT-\frac{\sqrt n\,t}{\lambda(1-b/q)^2},\qquad \lambda=\frac1{2b}.
 \label{eq:hb-refinement-shell}
\end{equation}
Its probability is at least $c_{B,t}n^{-1/2}$ uniformly over \eqref{eq:hb-coarse-shell}. Gaussian tail estimates divided by this shell probability give
$\langle w,w-V\rangle/n\to T$ uniformly. The fine conditional information and scaled actual distortion have means
$nR_1+\sqrt n t+O_B(1)$ and $nb-\sqrt n t/\lambda+O_B(1)$, where
$R_1=\log(s_a/b)/2$.
Their centered law, conditional on the source, converges uniformly to \eqref{eq:hb-covariance}.

Use the Gaussian decoder reference for $V\mid Y'$ and the conditional Poisson matching bound~\cite[Lemma~2 and Theorem~3]{poisson}. It applies to the actual conditional source law in each coarse cell: that law need not be Gaussian. Only the prescribed encoder and decoder reference kernels enter the bound. The conditioning cost of \eqref{eq:hb-refinement-shell} is logarithmic. A refinement alphabet of size
\[
 M_{1,j}=\left\lceil\exp\{nR_1+\sqrt n(L-\alpha_j)+C'\log n\}\right\rceil
\]
\edit{gives a one-shot success bound whose integrand is
$\Phi_{\Sigma_{a,b}}(L-A_n-t,t)+o(1)$, uniformly over source words in the retained branch. The guarantee for a fixed refinement code is the average of this integrand under the actual coarse-cell source law.} Choose a Borel maximizing allocation from $L-\alpha_j$. Maximizers are bounded for fixed $B$.

The common message enumerates $(j,\text{coarse index},\text{refinement label})$. Its cardinality is at most $\sum_jN_jM_{1,j}$, whose logarithm is
$nR_{a,b}+\sqrt nL+O_B(\log n)$.
For each branch and each of its finitely many coarse indices, the conditional Poisson construction is formed for the actual residual-source law induced by that coarse choice. Its reference kernels are those displayed above, not an assumption that the residual-source law is Gaussian. Averaging the bounded conditional failure estimates over these cells fixes a single finite family of constructions. \edit{Condition first on the coarse lists. For each realized coarse cell, the one-shot bound controls the error averaged under that cell's actual source law. Selecting a refinement realization with no larger error and weighting by the cell probabilities gives one deterministic collection of refinement codes. This does not assert a pointwise guarantee for that selected deterministic code.} All lists, branches, and labels are fixed before coding. The weak decoder uses only the branch and coarse index; the strong decoder uses the refinement label and its side information. Finally average the resulting total error over the coarse lists and fix a realization with no larger error. The limit on the cutoff is
$\E[\one_{\{|G_0|\le B\}}h_{a,b}(L-G_0)]$.
Let $B\to\infty$ by a slow diagonal choice. For arbitrary prescribed budgets, first use $L-\eta$ and then $\eta\downarrow0$. This completes the proof of Theorem~\ref{thm:gaussian-distortion-image}.



\section{Agreed-Reproduction Images and Their Gaussian Evaluation}\label{app:cr-details}
The common-reconstruction criterion in Section~\ref{sec:cr-gaussian} charges disagreement and actual excess distortion to one fixed error probability. All projection operations below are transformations of the encoder and decoder, with their finite headers included in the message size.

\subsection{A logarithmic-header geometric repair}
Fix $0<D<p$ and choose $\eta>0$ so small that $p-\eta>D$. On
$|\|x\|^2/n-p|\le\eta$ and $\|x-u\|^2\le nD$, the reproduction satisfies
\[
 \sqrt n(\sqrt{p-\eta}-\sqrt D)\le\|u\|
 \le\sqrt n(\sqrt{p+\eta}+\sqrt D).
\]
Thus the least-squares scale $\alpha_* =\langle x,u\rangle/\|u\|^2$ belongs to a fixed bounded interval. Take a grid of spacing $n^{-2}$ on a larger fixed interval and also include $1$. The encoder chooses the grid point minimizing $\|x-\alpha u\|^2$ and transmits its index. The decoder applies this scale to its reconstruction. There are polynomially many choices, so the header costs $O(\log n)$ nats. On original agreement the new reproductions agree, and inclusion of $1$ prevents an increase of distortion.

The nearest grid point to $\alpha_*$ has excess least-squares loss at most $\|u\|^2n^{-4}$. The minimizing choice has no larger loss, so $|\alpha-\alpha_*|\le n^{-2}$. For $u'=\alpha u$,
\begin{equation}
 |\langle x,x-u'\rangle-\|x-u'\|^2|
   =|\langle u',x-u'\rangle|
   =|\alpha|\|u\|^2|\alpha_*-\alpha|=O(n^{-1}).\label{eq:cr-projection-cross}
\end{equation}
Outside the displayed source/reconstruction set use the original scale; its index is already available. The source-energy exceptional mass tends to zero. This construction preserves original success on the retained set, so any converse for the repaired code bounds the original code.

For a coarse reproduction $u_0$ known from the message and a fine reproduction $u_1$ recovered on agreement, choose an orthonormal basis of $\operatorname{span}(u_0,u_1)$ by Gram--Schmidt, with deterministic rules for zero vectors and ties. The first basis vector depends only on $u_0$. Project $x$ first onto $\operatorname{span}(u_0)$ and then onto $\operatorname{span}(u_0,u_1)$. Quantize the projection amplitudes on $[-C\sqrt n,C\sqrt n]$ at spacing $n^{-2}$. At each stage include the original reproduction as an additional candidate and choose the candidate of least distortion. The original candidate is represented by one extra label rather than by transmitting its coordinates.

The weak decoder uses its coarse amplitude and its known original $u_0$. The strong decoder knows $u_0$ from the message and, on agreement, has the original $u_1$; hence it can form the same basis and use the fine amplitudes. Polynomial header alphabets suffice even if the original vectors are nearly collinear: coefficients refer to an orthonormal basis, not to the possibly ill-conditioned original vectors. With exact collinearity the second basis vector is omitted. Each repaired vector is within $O(n^{-2})$ of its exact projection on the retained source set, because its excess projection loss is at most that of the nearest grid point. For repaired residuals $w_0=x-u'_0$ and $w_1=x-u'_1$,
\begin{equation}
 \langle x,w_j\rangle=\|w_j\|^2+o(n),\qquad
 \langle w_0,w_1\rangle=\|w_1\|^2+o(n),\label{eq:cr-nested-cross}
\end{equation}
uniformly on that set. The exact relations follow from orthogonal projection onto nested subspaces; the quantization error is $o(n)$. Distortion and agreement are preserved by the inclusion of the original candidates. This argument uses that the stronger decoder can recover both original vectors on success.

\subsection{Single Gaussian common reconstruction}
Choose $q_{X\mid u}=\phi_{DI_n}(\cdot-u)$, so its channel output is
$q_{Y\mid u}=\phi_{(D+q)I_n}(\cdot-u)$. Lemma~\ref{lem:agreed-image} applies to the joint agreement/distortion event. The score is
\begin{equation}
 g=nR_{\rm CR}+\frac{\|x\|^2}{2p}-\frac{\|x-u\|^2}{2D}
              -\frac{\|y\|^2}{2(p+q)}+\frac{\|y-u\|^2}{2(D+q)}.
 \label{eq:cr-g-score}
\end{equation}
For $w=x-u$, $y=x+z$, $c=q/(p+q)$, and
$\lambda=q/[2D(D+q)]$, expansion gives
\begin{align}
 g+\lambda(\|w\|^2-nD)
 &=nR_{\rm CR}+\sqrt n A_n^{\rm CR}+J_n^{\rm CR},\label{eq:cr-spectral-identity}\\
 A_n^{\rm CR}&=\frac c{2\sqrt n}\left(\frac{\|x\|^2}{p}-n\right),\nonumber\\
 J_n^{\rm CR}&=\left\langle\frac w{D+q}-\frac x{p+q},z\right\rangle
 +\frac12\left(\frac1{D+q}-\frac1{p+q}\right)(\|z\|^2-nq).
 \nonumber
\end{align}
Apply the repair \eqref{eq:cr-projection-cross} first. Its logarithmic message increase does not change the second-order offset. On success, the positive deficit $nD-\|w\|^2$, the score tail in Lemma~\ref{lem:agreed-image}, and conditional variance $O(n)$ of $J_n^{\rm CR}$ imply
$nD-\|w\|^2=O_P(\sqrt n)$ eventwise. Thus successful mass can be retained on a source-measurable set where
\begin{equation}
 \frac{\|x\|^2}{n}\to p,\qquad
 \frac{\|w\|^2}{n}\to D,\qquad
 \frac{\langle x,w\rangle}{n}\to D.\label{eq:cr-three-moments}
\end{equation}
Let $\Delta=q/(D+q)-q/(p+q)$. Conditional Gaussian moments give
\begin{align}
 \lim\frac1n\Var(J_n^{\rm CR}\mid x,u)
 &=q\left[\frac D{(D+q)^2}+\frac p{(p+q)^2}
             -\frac{2D}{(D+q)(p+q)}\right]
   +\frac{q^2}{2}\left[\frac1{D+q}-\frac1{p+q}\right]^2\nonumber\\
 &=\Delta-\frac{\Delta^2}{2}.\label{eq:cr-conditional-variance}
\end{align}
Rotation into the span of $x,w$ proves a uniform conditional Gaussian limit as in Appendix~\ref{app:hb-image}. Under the original source law $A_n^{\rm CR}\Rightarrow\cN(0,c^2/2)$. Removing the nonnegative source-localization indicator after conditional comparison yields
\[
 \limsup P_{\rm c}\le\Phi\left(\frac L{\sqrt{V_{\rm CR}}}\right).
\]
The argument does not condition the channel on agreement.

For achievability use the backward Gaussian reference $X=U+E$ with independent variances $p-D$ and $D$. On a bounded source-energy fluctuation window, restrict the encoder kernel to
$nD-1\le\|x-U\|^2\le nD$.
Its probability under $Q_{U\mid X=x}$ is at least $c_Bn^{-1/2}$ by the scalar local chi-square bound. Gaussian tails, divided by this shell probability, give \eqref{eq:cr-three-moments}. Use the Gaussian decoder reference $Q_{U\mid Y}$ in the conditional Poisson matching construction~\cite{poisson}. On auxiliary recovery the decoder equals the encoder-selected $U$, and distortion is already satisfied by the shell. The true score is \eqref{eq:cr-g-score} minus the logarithm of the shell probability. Thus its normalized limit is $\cN(0,V_{\rm CR})$, and $O(\log n)$ overhead attains its CDF. Fixing a random-code realization makes the encoder's selected reproduction and the decoder deterministic. Cutoff removal and strict quantile inversion prove \eqref{eq:cr-rate}.

\subsection{Nested Gaussian common reconstruction}
Use Lemma~\ref{lem:agreed-image} with
$q_{X\mid u_1}=\phi_{bI_n}(\cdot-u_1)$ and
$r_{u_0}=\phi_{(a+q)I_n}(\cdot-u_0)$. Its score is
\begin{equation}
 g=nR_{\rm HB,CR}+\frac{\|x\|^2}{2p}-\frac{\|w_1\|^2}{2b}
                 -\frac{\|w_0+z\|^2}{2(a+q)}
                 +\frac{\|w_1+z\|^2}{2(b+q)},\quad w_j=x-u_j.
 \label{eq:hbcr-g}
\end{equation}
The exact distortion identity is
\begin{align}
 g+\frac{\|w_0\|^2-na}{2(a+q)}
       +\frac{q(\|w_1\|^2-nb)}{2b(b+q)}
 &=nR_{\rm HB,CR}+\frac{\|x\|^2-np}{2p}+N_n^{\rm CR},\label{eq:hbcr-identity}\\
 N_n^{\rm CR}
 &=\left\langle\frac{w_1}{b+q}-\frac{w_0}{a+q},z\right\rangle
 +\frac12\left(\frac1{b+q}-\frac1{a+q}\right)(\|z\|^2-nq).
 \nonumber
\end{align}
Both multipliers are positive. Apply the nested repair \eqref{eq:cr-nested-cross}. The image tail and conditional variance bound localize the two successful residual energies at $a,b$. Consequently the relevant source-measurable moment limits are
\[
 \|w_0\|^2/n\to a,\quad \|w_1\|^2/n\to b,\quad
 \langle w_0,w_1\rangle/n\to b.
\]
Their conditional noise variance is
\begin{align}
 \lim\frac1n\Var(N_n^{\rm CR})
 &=q\left[\frac b{(b+q)^2}+\frac a{(a+q)^2}
                 -\frac{2b}{(b+q)(a+q)}\right]
    +\frac{q^2}{2}\left[\frac1{b+q}-\frac1{a+q}\right]^2\nonumber\\
 &=\delta-\frac{\delta^2}{2}.
\end{align}
The source score has variance $1/2$. The same conditional Gaussian comparison proves the converse in \eqref{eq:hbcr-rate}. If the original reconstructions are collinear, their repaired residuals coincide up to negligible error; such sources cannot carry a nonvanishing successful mass with both distinct constraints active at this rate. No separate noncollinearity assumption is needed.

For the construction, use the coarse shells and branch headers from Appendix~\ref{app:hb-image}. Given the coarse residual $w_0$ with energy $na+O(1)$, use the backward reference $W_0=V+E$ with independent variances $a-b$ and $b$, restrict $\|w_0-V\|^2$ to $[nb-1,nb]$, and use its Gaussian decoder reference for $V\mid(Y-u_0)$. The fine common-reconstruction rate is
$\frac12\log[a(b+q)/(b(a+q))]$.
The two fixed residual energies make its conditional mean equal to this rate times $n$, up to $O(1)$, while its residual variance is $\delta-\delta^2/2$.
A source branch using $\sqrt nA_n$ extra coarse nats has $\sqrt n(L-A_n)$ fine nats available. Poisson matching recovers the encoder-selected fine word; both distortions are already satisfied by the two shells. Average the conditional scalar CDF over $A_n\Rightarrow\cN(0,1/2)$ and enumerate the branch indices as in Appendix~\ref{app:hb-image}. This gives the normal CDF with variance \eqref{eq:hbcr-parameters}. Strict inversion proves \eqref{eq:hbcr-rate}.



\section{Controlled Partitions and Structured Image Cells}\label{app:partition-tools}
\subsection{Proof of simultaneous auxiliary-to-partition transfer}
The strong functional representation lemma~\cite[Theorem~1]{sfrl}, with its bound converted from bits to nats, supplies a seed $S\perp X$ and a representation $T=g(X,S)$ with $I(X;S\mid T)\le\kappa$. Generate all physical outputs conditionally on $X$, independently of $S$. In the Poisson representation, fixing the seed leaves a countable set of possible values of $T$. Enumerating distinct values gives a label $J_S$ in bijection with $T$ at that fixed seed. Consequently
\begin{align}
 H(J_S\mid E_r,S)
 &=I(X;T\mid E_r,S)\nonumber\\
 &=I(X;T\mid E_r)+I(X;S\mid T,E_r)\nonumber\\
 &\le h_r+\kappa.\label{eq:app-partition-entropy}
\end{align}
The second equality uses $I(X;S\mid E_r)=0$. The final inequality follows from
\[
 I(X;S\mid T,E_r)=I(X;S\mid T)-I(S;E_r\mid T),
\]
since $I(S;E_r\mid X,T)=0$. Thus the same physical Markov constraint is retained after introducing the seed.

For the $j$th output, the relative-entropy chain rule yields
\begin{align}
 \E_{T,S}D(P_{B_j|T,S}\|Q_{j|T})
 &=d_j+I(B_j;S\mid T)\nonumber\\
 &\le d_j+I(X;S\mid T)\le d_j+\kappa.
 \label{eq:app-partition-divergence}
\end{align}
Here $S-X-B_j$ holds conditionally on $T$. Divide each of the $m+s$ nonnegative seed-dependent quantities by its positive upper bound in \eqref{eq:app-partition-entropy}--\eqref{eq:app-partition-divergence}, and sum. Its expectation is at most $m+s$, so one seed has sum at most $2(m+s)$. That seed gives every claimed bound. A quantity with zero expectation is zero almost surely and can be discarded. The same selection includes finitely many additional nonnegative approximation errors with their respective expectation bounds. It does not preserve a signed moment merely from its mean; use an absolute or squared error for such a moment.

For a finite existing source-cell index $M$, the construction can be performed under each $P_{X|M=m}$. Uniform component estimates allow the same constants. More generally, average the conditional overhead and use concavity of $u\mapsto\log(1+u/\log2)$. The resulting header still belongs to the terminal observing $(M,X)$.

\subsection{An exponentially large family of regular components}\label{app:structured-example}
We prove the assertions of Proposition~\ref{prop:structured-cells}. All binary entropies are in nats. Choose fixed rational parameters $t,q,\delta$ with
$0<t,q<1/2$ and $2(t+q)<\delta<1/2$. Take compatible even blocklengths and a rate
$0<R<\log2-h(\delta)$. Greedy packing inside the balanced type class gives at least
$\exp\{n[\log2-h(\delta)]-O(\log n)\}$ centers at pairwise Hamming distance greater than $n\delta$. Indeed, that type class has size $\exp(n\log2-O(\log n))$, while a full binary Hamming ball of radius $n\delta$ has at most $\exp(nh(\delta)+O(\log n))$ points. Retain $L_n=\lfloor e^{nR}\rfloor$ centers $a^n(j)$.

For each center, let $\mathcal A_j$ consist of sequences obtained by flipping exactly a fraction $t$ of the positions of each center symbol. These are disjoint because they lie in radius-$nt$ Hamming balls. All have size
\begin{equation}
 |\mathcal A_j|=\binom{n/2}{nt/2}^{\!2}
       =\exp\{nh(t)+O(\log n)\}.
 \label{eq:structured-cardinality}
\end{equation}
Let $J$ be uniform on $[L_n]$, draw $X^n$ uniformly on $\mathcal A_J$, and observe
$C^n=X^n\oplus N^n$, $N_i\sim\operatorname{Bern}(q)$ independently. Then $J$ is a deterministic function of $X^n$. Choose $\eta>0$ with $2(t+q+\eta)<\delta$. Nearest-center decoding of $C^n$ recovers $J$ whenever $\|N^n\|_1\le n(q+\eta)$, so its error is at most $e^{-c n}$. Fano's inequality gives
\begin{equation}
 H(J\mid C^n)\le h(e^{-cn})+e^{-cn}\log L_n=O(ne^{-cn}).
 \label{eq:structured-recoverability}
\end{equation}
A smaller positive constant absorbs the polynomial factor.

For a fixed center let $Q_{X|A}$ be the binary symmetric test channel with crossover $t$ and $W_{C|X}$ the BSC with crossover $q$. Write $v=t+q-2tq$ and let
\[
 s(a,x,c)=-\log Q_{X|C,A}(x\mid c,a).
\]
For every $x^n\in\mathcal A_j$, the sum $\sum_i s(a_i(j),x_i,C_i)$ has mean $nH_*$ and variance $nV_*$, where
\begin{align}
 H_*&=h(t)+h(q)-h(v),\nonumber\\
 V_*&=\sum_{a,x}Q_{AX}(a,x)
                  \Var_{W(\cdot|x)}s(a,x,C),\quad Q_A(0)=Q_A(1)=1/2.
 \label{eq:structured-spectrum-parameters}
\end{align}
Its law is the same for every source word in the component, and, by binary symmetry, for every component. The increments are bounded, and $V_*>0$: for a fixed $(a,x)$ the posterior score differs between the two positive-probability observation values.

The reference product law conditioned on the conditional-type event gives exactly the uniform component. Its conditioning probability is polynomially bounded below. The output-reference correction and its lower-tail estimate therefore give
\begin{equation}
 -\log P(X^n\mid C^n,J)
  =\sum_i s(a_i(J),X_i,C_i)+O_P(\log n).
 \label{eq:structured-component-score}
\end{equation}
This statement is uniform over the components. It follows either from Theorem~\ref{v5:screened} or directly from Bayes' rule and
$P_{C^n|J}/Q_{C^n|a^n(J)}\le\operatorname{poly}(n)$, with an $e^{-\gamma}$ lower likelihood tail. Hence the centered conditional spectrum has limit $\mathcal N(0,V_*)$.

A section image for $(X^n,C^n)$ has at most $M$ source words for each observation. Thresholding the component conditional spectrum constructs such an image with successful mass at least
$P\{-\log P(X^n|C^n,J)\le\log M\}$.
Section counting gives the converse with threshold $\log M+\gamma$ and error $e^{-\gamma}$. Corollary~\ref{cor:partition-profile}, using \eqref{eq:structured-recoverability}, therefore gives the same Gaussian image limit without revealing $J$.
For an actual fixed-length source code, assign random bin labels separately in each component and decode the most likely component source word after recovering $J$ from $C^n$. The conditional binning bound is
\[
 P_{\rm e}\le P\{-\log P(X^n|C^n,J)>\log M-\gamma\}
                  +e^{-\gamma}+e^{-cn}.
\]
All component maps can be fixed in advance by averaging, and the same message alphabet is reused across components. The encoder computes $J$ from $X^n$ and the decoder recovers it from its own observation. The preceding image converse applies to every code, proving the rate formula in Proposition~\ref{prop:structured-cells}.

Finally, the union $\mathcal A=\bigcup_j\mathcal A_j$ lies inside the balanced source type and has global logarithmic density deficit
\[
 \log\frac{|T_n(1/2)|}{|\mathcal A|}
       =n[\log2-R-h(t)]+O(\log n),
\]
which is positive and linear. Under the unconditional fair-input product reference $R_{X^nC^n}$, its conditional divergence is
\begin{align}
 K_n&=D(P_{X^nC^n}\|R_{X^nC^n})
                    -D(P_{C^n}\|R_{C^n})\nonumber\\
    &=H(C^n)-H(X^n)\nonumber\\
    &=n[h(v)-h(t)]+O(\log n)>0\quad\text{to first order}.
 \label{eq:structured-screening-failure}
\end{align}
To justify the last equality, $H(X^n)=\log L_n+nh(t)+O(\log n)$ by \eqref{eq:structured-cardinality}. The component output reference has entropy $nh(v)$ and relative-entropy correction $O(\log n)$, so
$H(C^n|J)=nh(v)+O(\log n)$. Then use
$H(C^n)=H(J)+H(C^n|J)-H(J|C^n)$ and \eqref{eq:structured-recoverability}. Thus both the global density test and the small conditional-divergence test fail, while the recoverable-partition theorem applies.





\begin{thebibliography}{99}

\bibitem{ck}
I.~Csisz{\'a}r and J.~K{\"o}rner, \emph{Information Theory: Coding Theorems for
  Discrete Memoryless Systems}, 2nd~ed.\hskip 1em plus 0.5em minus 0.4em\relax
  Cambridge University Press, 2011.

\bibitem{imagecomp}
D.~Cao, ``Exact second-order asymptotics for the Wyner--Ahlswede--K\"orner problem,'' \href{https://arxiv.org/abs/2608.21991v2}{arXiv:2608.21991.v2}, 2026.

\bibitem{wak}
D.~Cao, ``Tight weighted second-order asymptotics for the
Wyner--Ahlswede--K\"orner problem under regular posterior geometry,''
\href{https://arxiv.org/abs/2608.21991v1}{arXiv:2608.21991.v1}, 2026.

\bibitem{blfr}
J.~Liu, T.~A. Courtade, P.~W. Cuff, and S.~Verd\'u,
``A forward-reverse Brascamp--Lieb inequality: Entropic duality and Gaussian optimality,'' \emph{Entropy}, vol.~20, no.~6, Art.~418, 2018.

\bibitem{blinfo}
J.~Liu, T.~A. Courtade, P.~Cuff, and S.~Verd\'u,
``Information-theoretic perspectives on Brascamp--Lieb inequality and its reverse,'' arXiv:1702.06260, 2017.

\bibitem{lvhv}
J.~Liu, R.~van Handel, and S.~Verd\'u, ``Second-order converses via reverse hypercontractivity,'' \emph{Mathematical Statistics and Learning}, vol.~2, no.~2, pp.~103--163, 2019. arXiv:1812.10129.

\bibitem{wkt}
S.~Watanabe, S.~Kuzuoka, and V.~Y.~F. Tan, ``Non-asymptotic and second-order
  achievability bounds for coding with side-information,'' arXiv:1301.6467,
  2013.

\bibitem{poisson}
C.~T. Li and V.~Anantharam, ``A unified framework for one-shot achievability
  via the {Poisson} matching lemma,'' \emph{IEEE Transactions on Information
  Theory}, vol.~67, no.~5, pp. 2624--2651, 2021.

\bibitem{lili}
X.~Li and C.~T. Li, ``New second-order achievability bounds for coding with
  side information via type deviation convergence,'' arXiv:2510.20241, 2025.

\bibitem{consolidated}
L.~Zhou, ``Exact second-order asymptotics for the discrete Wyner--Ziv
problem,'' \href{https://arxiv.org/abs/2610.11998v1}{arXiv:2610.11998v1}, 2026.

\bibitem{kaspi}
L.~Zhou and M.~Motani, ``Non-asymptotic converse bounds and refined asymptotics
  for two lossy source coding problems,'' arXiv:1708.05496, 2017.

\bibitem{htw}
M.~Hayashi, H.~Tyagi, and S.~Watanabe, ``Secret key agreement: General capacity
  and second-order asymptotics,'' arXiv:1411.0735, 2014.

\bibitem{rd}
V.~Kostina and S.~Verd{\'u}, ``Fixed-length lossy compression in the finite
  blocklength regime,'' \emph{IEEE Transactions on Information Theory},
  vol.~58, no.~6, pp. 3309--3338, 2012.

\bibitem{nit}
A.~{El Gamal} and Y.-H. Kim, \emph{Network Information Theory}.\hskip 1em plus
  0.5em minus 0.4em\relax Cambridge University Press, 2011.

\bibitem{conditionalrd}
S.-Q. Le, V.~Y.~F. Tan, and M.~Motani, ``Second-order coding rates for
  conditional rate-distortion,'' arXiv:1410.2687, 2014.

\bibitem{xuzhou}
J.~Xu and L.~Zhou, ``Exact second-order asymptotics for the Gaussian Wyner--Ziv problem,'' working paper, version~12, 2026. [Online]. Available: \url{https://lzhou-sustech.github.io/preprint-pdfs/Gaussian_Wyner_Ziv_Second_Order.pdf}

\bibitem{zhouyu}
L.~Zhou and L.~Yu, ``Exact second-order asymptotics for the Gaussian Gray--Wyner problem,'' \href{https://arxiv.org/abs/2610.38306v1}{arXiv:2610.38306v1}, 2026.

\bibitem{oufkir}
A.~Oufkir and M.~Berta, ``Channel dichotomies,'' \href{https://arxiv.org/abs/2610.41858v1}{arXiv:2610.41858v1}, 2026.

\bibitem{steinbergcr}
Y.~Steinberg, ``Coding and common reconstruction,'' \emph{IEEE Transactions on Information Theory}, vol.~55, no.~11, pp.~4995--5010, 2009.

\bibitem{ahmadi}
B.~Ahmadi, R.~Tandon, O.~Simeone, and H.~V.~Poor, ``Heegard--Berger and cascade source coding problems with common reconstruction constraints,'' arXiv:1112.1762v4, 2012.

\bibitem{heegard}
C.~Heegard and T.~Berger, ``Rate distortion when side information may be absent,'' \emph{IEEE Transactions on Information Theory}, vol.~31, no.~6, pp.~727--734, 1985.

\bibitem{steinbergmerhav}
Y.~Steinberg and N.~Merhav, ``On successive refinement for the Wyner--Ziv problem,'' \emph{IEEE Transactions on Information Theory}, vol.~50, no.~8, pp.~1636--1654, 2004.

\bibitem{chen}
W.-K.~Chen, N.~Dafnis, and G.~Paouris, ``Improved H\"older and reverse H\"older inequalities for Gaussian random vectors,'' \emph{Advances in Mathematics}, vol.~280, pp.~643--689, 2015.

\bibitem{hoeffding}
W.~Hoeffding, ``Probability inequalities for sums of bounded random variables,'' \emph{Journal of the American Statistical Association}, vol.~58, no.~301, pp.~13--30, 1963.

\bibitem{immse}
D.~Guo, S.~Shamai (Shitz), and S.~Verd\'u, ``Mutual information and minimum mean-square error in Gaussian channels,'' \emph{IEEE Transactions on Information Theory}, vol.~51, no.~4, pp.~1261--1282, 2005.

\bibitem{payaro}
M.~Payar\'o and D.~P.~Palomar, ``Hessian and concavity of mutual information, differential entropy, and entropy power in linear vector Gaussian channels,'' \emph{IEEE Transactions on Information Theory}, vol.~55, no.~8, pp.~3613--3628, 2009.

\bibitem{dytso}
A.~Dytso, H.~V.~Poor, and S.~Shamai (Shitz), ``Conditional mean estimation in Gaussian noise: A meta derivative identity with applications,'' \emph{IEEE Transactions on Information Theory}, vol.~69, no.~3, pp.~1883--1898, 2023.

\bibitem{ledoux}
M.~Ledoux, \emph{The Concentration of Measure Phenomenon}, ser.~Mathematical Surveys and Monographs, vol.~89. American Mathematical Society, 2001.

\bibitem{hentila}
H.~Hentil\"a, Y.~Y.~Shkel, and V.~Koivunen, ``Communication-constrained secret key generation: Second-order bounds,'' \emph{IEEE Transactions on Information Theory}, vol.~70, no.~11, pp.~8180--8203, 2024.

\bibitem{tanbc}
\edit{V.~Y.~F.~Tan, ``Second-order asymptotics for the scalar Gaussian broadcast channel,'' \href{https://arxiv.org/abs/2611.01583v1}{arXiv:2611.01583v1}, 2026.}

\bibitem{sfrl}
\edit{C.~T.~Li and A.~El Gamal, ``Strong functional representation lemma and applications to coding theorems,'' \emph{IEEE Transactions on Information Theory}, vol.~64, no.~11, pp.~6967--6978, 2018. \href{https://arxiv.org/abs/1701.02827v4}{arXiv:1701.02827v4}.}
\end{thebibliography}
\end{document}